\documentclass[12pt]{article}
\usepackage[latin9]{inputenc}
\usepackage{geometry}
\usepackage{color}
\usepackage{verbatim}
\usepackage{mathtools}
\usepackage{dsfont}
\usepackage{amsmath}
\usepackage{amsthm}
\usepackage{amssymb}
\usepackage{stmaryrd}
\usepackage{setspace}
\usepackage{verbatim}
\usepackage{calc}
\usepackage{mathrsfs}
\usepackage{dsfont}
\usepackage{bbm}
\usepackage{mathrsfs}
\usepackage{forest}
\usepackage{lscape}
\usepackage{tikz}

\def\mP{\mathbb{P}}

\def\mR{\mathbb{R}}
\def\mE{\mathbb{E}}
\def\mF{\mathbb{F}}
\def\mG{\mathbb{G}}
\def\mH{\mathbb{H}}
\def\mN{\mathbb{N}}
\def\mT{\mathbb{T}}
\def\mX{\mathbb{X}}
\def\cA{\mathcal{A}}
\def\cF{\mathcal{F}}
\def\cG{\mathcal{G}}
\def\cH{\mathcal{H}}

\def\cJ{\mathcal{J}}
\def\cB{\mathcal{B}}
\def\cP{\mathcal{P}}
\def\cO{\mathcal{O}}

\def\a{\alpha}
\def\b{\beta}

\def\G{\Gamma}
\def\w{\omega}
\def\W{\Omega}
\def\s{\sigma}

\def\r{\rho}
\def\l{\lambda}
\def\L{\Lambda}
\def\m{\mu}
\def\n{\nu}
\def\p{\pi}
\def\P{\Pi}
\def\t{\tau}
\def\d{\delta}
\def\D{\Delta}

\def\e{\varepsilon}
\def\z{\zeta}
\def\|{\lVert}
\def\lb{\llbracket}

\def\sup{\text{sup }}

\def\min{\text{min }}

\def\loc{\text{loc }}
\def\inf{\text{inf }}

\def\1{\mathbbm{1}}
\def\sA{\mathscr{A}}
\def\sX{\mathscr{X}}
\def\sR{\mathscr{R}}
\def\sZ{\mathscr{Z}}
\def\sL{\mathscr{L}}

\newcommand{\ov}[1]{\overline{#1}}
\newcommand{\wt}[1]{\widetilde{#1}}
\newcommand{\wh}[1]{\widehat{#1}}

\makeatletter
\theoremstyle{plain}
\newtheorem{thm}{\protect\theoremname}
\theoremstyle{plain}
\newtheorem{cor}[thm]{\protect\corollaryname}
\theoremstyle{plain}
\newtheorem{lem}[thm]{\protect\lemmaname}
\theoremstyle{definition}
\newtheorem{defn}[thm]{\protect\definitionname}
\theoremstyle{remark}
\newtheorem{rem}[thm]{\protect\remarkname}
\theoremstyle{definition}
\newtheorem{example}[thm]{\protect\examplename}

\makeatother

\providecommand{\corollaryname}{Corollary}
\providecommand{\definitionname}{Definition}
\providecommand{\examplename}{Example}
\providecommand{\lemmaname}{Lemma}
\providecommand{\remarkname}{Remark}
\providecommand{\theoremname}{Theorem}

\begin{document}

\title{Enlargement of Filtrations: An Exposition \\ of Core Ideas with Financial Examples}
\author{\textcolor{black}{Karen Grigorian}\thanks{\textcolor{black}{Operations Research and Industrial Engineering,
Cornell University, Ithaca, N.Y. 14853. email: kgg44@cornell.edu.}}\textcolor{black}{{} and Robert A. Jarrow}\thanks{\textcolor{black}{Samuel Curtis Johnson Graduate School of Management,
Cornell University, Ithaca, N.Y. 14853. email: raj15@cornell.edu.}}}
\maketitle

\begin{abstract}
In this paper we provide an exhaustive survey of the current state of the mathematics of filtration enlargement and an interpretation of the key results of the literature from the viewpoint of mathematical finance. The emphasis is \textcolor{black}{on providing a well-structured compendium of known mathematical results that can be used by researchers in mathematical finance. We mainly state the results and discuss their role and significance}, with references provided for the omitted proofs. The discussion of mathematical results is accompanied by numerous examples from mathematical finance.
\end{abstract}

\section{Introduction to Enlargements}

\subsection{A Philosophical Aside}

The fundamental role of filtrations as a mathematical model of information flow is well-established in modern finance. The filtration imposes "an arrow of time" on the model. In the words of P.Medvegyev, "...A fundamental property of time is its 'irreversibility'. This property of time is expressed with the introduction of the filtration." Furthermore, quadratic variation is one of the core notions in stochastic analysis, and its role is highlighted in the interplay between enlargement of filtrations and Doob's decomposition. The vague ideas of "signal" and "noise" are given a precise meaning in the context of Doob's decomposition and enlargement. As a model of information, an expanded filtration (e.g. access to inside information) may potentially create an arbitrage opportunity. Lastly, a wide array of results and notions are critically sensitive to the underlying filtration and the role of the filtration as a core part of the stochastic basis  should not be neglected in any discussion of a given financial model as well as stochastic processes in general. These and many other ideas are discussed more rigorously in the following sections. 

\par
An outline of this paper is as follows. Section 1 serves as a brief general introduction to the mathematical theory of filtration enlargements and its connections to mathematical finance. Section 2 discusses the relationship between filtration enlargement and different notions of arbitrage. Section 3 analyzes the relationship between filtration enlargement and change of measure techniques.  Section 4 studies the preservation/loss of the martingale representation property under enlargement. Section 5 provides a detailed account of filtration enlargement in the setting of Levy processes, processes with independent increments and point processes. Section 6 concludes. 

\subsection{Basic Notions}

We begin by introducing a number of definitions. \par

\begin{defn}
Filtration $\mathbb{G} = (\mathcal{G}_t, t \geq 0)$ is called \textbf{an enlargement} (expansion) of filtration $\mathbb{F} = (\mathcal{F}_t, t \geq 0)$ if $\mathcal{F}_t \subset \mathcal{G}_t, \quad \forall t \geq 0$. 
\end{defn}

\par
In general, it is not true that an $\mathbb{F}$-martingale is a $\mathbb{G}$-martingale, or even a $\mathbb{G}$-semimartingale, hence two hypotheses arise quite naturally from these considerations.We are interested in conditions that ensure that every $\mathbb{F}$-martingale is a $\mathbb{G}$-martingale (known in the literature as  hypothesis $\mathcal{H}$, or \textbf{immersion}, denoted as $\mathbb{F} \hookrightarrow \mathbb{G}$) or that every $\mathbb{F}$-martingale is a $\mathbb{G}$-semimartingale (known as the $\mathcal{H}^{\prime}$ hypothesis). \par
There are two main kinds of enlargement of filtration. We speak of \textbf{initial enlargement}, when the larger filtration is taken to be $\mathbb{G} \coloneqq \mathbb{F} \lor \sigma(\zeta)$ where $\zeta$ is some random variable (or, more generally, a random element). \textbf{Progressive enlargement} is defined to be the smallest $\mathbb{G}$ that turns a given random time $\tau$ (any nonnegative random variable) into a stopping time. \par

\subsection{Initial Enlargement}

A standard result on immersion under initial enlalrgement is the following (proof can be found, e.g., in Aksamit, Jeanblanc, 2017, p.11). \textcolor{black}{Here and throughout $\cF\perp\!\!\!\perp_{\cH}
\cG$ stands for conditional independence of $\s$-algebras $\cF$ and $\cG$ given $\cH.$ When no subscript is present, then the independence is unconditonal.}\par

\begin{thm}
 Let $\mathbb{F}$ be a filtration and $\zeta$ be a RV. Then $\mathbb{F} \hookrightarrow \mathbb{F} \lor \sigma(\zeta)$ if and only if $\sigma(\zeta)\perp\!\!\!\perp_{\mathcal{F}_0}
\mathcal{F}_{\infty}$. In particular, $\mathbb{F} \hookrightarrow \mathbb{F} \lor \sigma(\zeta)$ if $\sigma(\zeta)\perp\!\!\!\perp
\mathcal{F}_{\infty}$. 
\end{thm}

\par
One of the most famous and important results in the whole area is the well-known Jacod's criterion for $\mathcal{H}^{\prime}$ (sometimes denoted as $\mathcal{J})$. We state the result only, proof can be found in, e.g., (Protter, 2005, p.371). \par
\begin{thm} Let $\zeta$ be a random element in a standard Borel space ($E, \mathcal{E}$), and let $Q_t(\omega, dx)$ be the regular conditional distribution of $\zeta$ given $\mathcal{F}_t$, $\forall t \geq 0$. Suppose that there exists a positive $\sigma$-finite measure $\eta$ on $(E, \mathcal{E})$ s.t. $Q_t(\omega, dx) \ll \eta(dx)$ a.s., $\forall t \geq 0$. Then $\mathcal{H}^{\prime}$ holds. 
\end{thm}
\par
An immediate yet important corollary is the following. \par
\begin{cor} Let $\zeta$ be a random element in ($E, \mathcal{E}$) s.t. $\zeta$ is independent of $\mathbb{F}$, and let $\mathbb{G} \coloneqq \mathbb{F} \lor \sigma(\zeta)$. Then Jacod's criterion is satisfied and hence   $\mathcal{H}^{\prime}$ holds. \par
\begin{proof} By independence, the regular conditional distributions $Q_t(\omega, dx)$ of $\zeta$ are equal to the law $\eta(dx)$ of $\zeta.$
\end{proof}
\end{cor}
\par
While Jacod's condition is arguably one of the most general results in the area, it may fail even in relatively simple setups. \par

\begin{example}(Corcuera, Valdivia, 2012) Let $L$ be the $n$-th jump of a Poisson process $(N_t)_{t \in (0,T)}$ with intensity $\lambda$ and $\mathbb{F}$ its natural filtration. Then
\begin{align*}
    \mathbb{P}(L > x \vert \mathcal{F}_t) = \mathbbm{1}_{\{N_x < n, N_t \geq n\}} + \mathbbm{1}_{\{N_t < n\}} \int_{(x-t)_+}^{\infty}\frac{\lambda e^{-\lambda u}(\lambda u)^{n - N_t - 1}}{(n - N_t -1)!} du,
\end{align*}
and the conditional probability cannot be dominated by a non-random measure.
\end{example}
\par
\begin{example} (Ouwehand, 2007; Amedinger et al, 1998) We consider a standard example where at time 0 the value of the Brownian motion at some terminal date $T$ (e.g. at $T=1$) is known.  Let $B$ be a Brownian motion, $\mathbb{F}$ its natural filtration and we enlarge it with the $\sigma$-field generated by $B_1$, i.e. $\mathbb{F}^{\sigma(B_1)} \coloneqq \cap_{s > t} (\mathcal{F}_s \lor \sigma(B_1))$. The process $B$ is no longer an $\mathbb{F}^{\sigma(B_1)}$-martingale. To see this, consider, for example, the process $(\mathbb{E}[B_1 | \mathcal{F}^{\sigma(B_1)}_t], t \geq 0)$ and note that it is identically equal to $B_1)$. $B$ is, however, an $\mathbb{F}^{\sigma(B_1)}$-semimartingale, and the process $\beta$ defined as
\begin{align*}
     \beta_t \coloneqq B_t - \int_0^{t \land 1} \frac{B_1 - B_s}{1-s}ds
\end{align*}
is an $\mathbb{F}^{\sigma(B_1)}$-Brownian motion. Jacod's criterion fails to hold on $[0, \infty)$. It does, however, hold on $[0, T)$. $B_T$ is clearly Gaussian conditional on $\mathcal{F}_t$ for $t < T$, hence if $Q_t(\omega, dx)$ is a regular version of $\mathbb{P}(B_T \in dx | \mathcal{F}_t)$, then  $Q_t(\omega, dx)$ is absolutely continuous with respect to Lebesgue measure. Note, however, for $t \geq T$, $Q_t(\omega, dx)$ is the point mass $\delta_{B_T(\omega)}(dx)$ and it is impossible to find a single measure $\eta$ s.t. $\delta_c(dx) \ll \eta(dx)$ for all $c \in \mathbb{R} $. 
\end{example}

\par
\begin{example} (Amedinger, 1998) Consider a random variable $G$ with values in a countable set $U$ such that $\mathbb{P}(G = l) > 0$ for all $l \in U$. Hence the sets in the $\sigma$-algebra generated by $G$ are of the form $A = \cup_{l \in J}\{G = l\}$ for some subset $J \subseteq U$. Thus we have 
\begin{align*}
    \mathbb{P}(G \in A \vert \mathcal{F}_t) = \sum_{l \in J}\mathbb{P}(G = l \vert \mathcal{F}_t) = \sum_{l \in J} p^l_t\mathbb{P}(G=l) = \int_A p^l_T\mathbb{P}(G \in dl) = \int_A p^l_T\eta(dl) \ \forall t \in [0,T], 
\end{align*}
where $p^l_t = \frac{\mathbb{P}(G = l \vert \mathcal{F}_t)}{\mathbb{P}(G=l)}$, and the conditional probability $\mathbb{P}(G \in \cdot \vert \mathcal{F}_t)$ is absolutely continuous with respect to the law of $G$ (i.e. $\eta(dl)=\mathbb{P}(G \in dl)$) \ $\forall t \in [0,T]$. Note, however, that the conditional laws of $G$ given $\mathcal{F}_t$ are equivalent to the law of $G$ for $t < T$ only if $\mathbb{P}(G = l \vert \mathcal{F}_t) > 0$ $\mathbb{P}$-a.s. for all $l \in U$. The equivalence also fails to hold if $G$ is $\mathcal{F}_T$-measurable, since in this case $\mathbb{P}(G = l \vert \mathcal{F}_T) = \mathbbm{1}_{\{G = l\}}$ can be zero with positive probability unless $G$ is a constant. \par
An instructive special case is when $G \coloneqq \mathbbm{1}_{\{B_T \in [a,b]\}}$, i.e. it signals whether the terminal value $B_T$ of a standard Brownian motion lies in some interval $[a,b]$. Then 
\begin{align}
    p^1_t = \frac{\mathbb{P}(G = 1 \vert \mathcal{F}_t)}{\mathbb{P}(G=1)},\quad  p^0_t = \frac{1 - \mathbb{P}(G = 1 \vert \mathcal{F}_t)}{1 - \mathbb{P}(G=1)},
\end{align}
and it can be easily calculated that the conditional law is 
\begin{align*}
    & \mathbb{P}(G = 1 \vert \mathcal{F}_t) = \frac{1}{\sqrt{2\pi(T-t)}}\int_a^b\text{exp}(- \frac{(u - B_t)^2}{2(T-t)})du, \quad t \in [0,T), \\
    & \mathbb{P}(G = 1) = \mathbb{P}(G = 1 \vert \mathcal{F}_0) = \Phi(b/\sqrt{T}) - \Phi(a/\sqrt{T}),
\end{align*}
where $\Phi$ is the standard normal distribution function. Thus, the conditional probability $\mathbb{P}(G \in \cdot \vert \mathcal{F}_t)$ is absolutely continuous with respect to the law of $G$ for all $t \in [0, T)$ \textcolor{black}{and Jacod's condition holds on $[0,T)$}.
\end{example}

\par
\begin{example} (Amedinger) Let $G = B_T + \epsilon$, where $B_T$ is the terminal value of a standard Brownian motion $B$ and $\epsilon$ is an $\mathcal{N}(0,1)$-random variable independent of $\mathcal{F}_T$ that can be interpreted as distortion of the available information $B_T$ by some "noise". Then $\forall t \in [0,T]$
\begin{align*}
    \mathbb{P}(G \in dx \vert \mathcal{F}_t) & =\mathbb{P}(B_T - B_t + B_t + \epsilon \in dx \vert \mathcal{F}_t) \\
    & = \mathbb{P}(B_T - B_t + \epsilon \in dx - y \vert \mathcal{F}_t) \vert_{y = B_t} \\
    & = \frac{1}{\sqrt{2\pi(T - t +1)}} \text{exp}(-\frac{(x - B_t)^2}{2(T - t +1)})dx \\
    & = p_t(x)\mathbb{P}(B_T + \epsilon \in dx),
\end{align*}
where $p_t(x) = \sqrt{\frac{T + 1}{T - t + 1}}\text{exp}(-\frac{(x - B_t)^2}{2(T - t +1)} + \frac{x^2}{2(T+1)})$. An application of Ito's formula to $\frac{(x - B_t)^2}{(T - t +1)}$ yields
\begin{align*}
    p_t(x) = \mathcal{E}\left(\int \frac{(x - B_s)^2}{(T - s +1)}dB_s\right)_t,
\end{align*}
which is an $\mathbb{F}$-martingale as can be shown by checking Novikov's condition. Thus, the conditional probability $\mathbb{P}(G \in \cdot \vert \mathcal{F}_t)$ is absolutely continuous with respect to the law of $G$ for all $t \in [0, T]$, including the endpoint of the interval.\par
This example can be generalized in a natural way if we consider instead a mixture of the terminal value and the noise term, i.e. $G = \lambda B_T + (1 - \lambda) \epsilon$. Calculations similar to the above yield that $\mathbb{P}(G \in dx \vert \mathcal{F}_t) = p_t(x)\mathbb{P}(G \in dx)$, where $p_t(x) = \sqrt{\frac{\lambda^2T + (1 - \lambda)^2}{\lambda^2(T-t) + (1 - \lambda)^2}}\text{exp}(-\frac{(x - \lambda B_t)^2}{2\lambda^2(T-t) + (1 - \lambda)^2} + \frac{x^2}{2(\lambda^2T + (1 - \lambda)^2)})$, $x \in \mathbb{R}$. Thus, the conditional probability $\mathbb{P}(G \in \cdot \vert \mathcal{F}_t)$ is the normal probability distribution with mean $\mu_t = \lambda B_t$ and variance $\sigma^2_t = \lambda^2(T-t) + (1 - \lambda)^2$, \textcolor{black}{and Jacod's condition holds on $[0.T]$}.
\end{example}

\par
Of key importance is the following lemma which allows us to chose a nice version of the density process (see, e.g., Jacod, 1985, for a proof). \textcolor{black}{Below and throughout, $\mF^{\s(\zeta)}$ is the right-continuous augmentation of $\mF\lor\s(\zeta)$, i.e. the initial enlargement of $\mF$ by a random variable $\zeta$ taking values in $\ov{\mR}\coloneqq [-\infty,\infty]$.}  \par
\begin{lem} Under $\mathcal{J}$, there exists a non-negative $\mathcal{O}(\mathbb{F}) \otimes \mathcal{B}(\overline{\mathbb{R}})$-measurable function $\Omega \times \mathbb{R}_+ \times \overline{\mathbb{R}} \ni (\omega, t, x) \mapsto p_t (\omega, x)$ cadlag in $t$ such that  \par
(i) for every $ t \geq 0$, we have $Q_t(\omega, dx) = p_t(\omega, x) \eta(dx)$, \par
(ii) for each $x \in \overline{\mathbb{R}}$, the process $(p_t(x))_{t \geq 0}$ is an $\mathbb{F}$-martingale, \par
(iii) $p(x) > 0$ and $p_-(x) > 0$ on $\llbracket0, \zeta^x \llbracket$ and $p(x) = 0$ on $\llbracket \zeta^x, \infty \llbracket$, where $\zeta^x \coloneqq \inf\{t \vert p_-(x)=0\}$. 
\end{lem}

\par
While Jacod's criterion is a general result that ensures the validity of $\mathcal{H}^{\prime}$, a more specific question is what the semimartingale decomposition of the process is in the enlarged filtration, resolved by the following theorem. Note that the conditional density process $p(x)$ ensured by the lemma is crucial to obtain this decomposition. \par

\begin{thm} Suppose that the random element $\zeta$ satisfies $\mathcal{J}$. If $X$ is an $\mathbb{F}$-local martingale, the process $\widetilde{X}$ defined as 
\begin{align*}
    \widetilde{X}_t \coloneqq X_t - \int_0^t \frac{1}{p_{s-}(\zeta)}d\langle X, p(u)\rangle ^{\mathbb{F}}_s\vert_{u = \zeta}, \quad t \leq T
\end{align*}
is an $\mathbb{F}^{\sigma(\zeta)}$-local martingale. In particular, $\mathcal{H}^{\prime}$ holds. 
\end{thm}
\begin{proof}
See, e.g., in (Aksamit, Jeanblanc, 2017, p. 88). 
\end{proof}

\subsection{Progressive Enlargement}

Consider the filtration $\mathbb{A} \coloneqq (\mathcal{A}_t, t \geq 0)$ generated by the so-called \textit{default indicator process} $A \coloneqq \mathbbm{1}_{\llbracket \tau, \infty \llbracket}$, where $\tau$ is a \textit{random time}, i.e. a nonnegative random variable. Progressive enlargement of the filtration is defined as $\mathbb{G} \coloneqq \mathbb{F} \lor \mathbb{A}$. Of fundamental importance is the \textbf{Azema supermartingale} associated with this process, defined to be $Z = {}^o(1 - A) = (\mathbb{P}(\tau > t \vert \mathcal{F}_t))_{t \geq 0}$, \textcolor{black}{i.e. the optional projection of $1-A$}. Sometimes a related process is studied $\widetilde{Z} = {}^o(1 - A_-) = (\mathbb{P}(\tau \geq t \vert \mathcal{F}_t))_{t \geq 0}$, which only has right and left limits. These processes can be shown to be of class ($D$). It can also be shown that $Z = m - A^o = n - A^p$ for some $\mathbb{F}$-martingales $m$ and $n$ (where $A^o$ and $A^p$ are the dual optional and predictable projections in $\mathbb{F}$ respectively). The following are fundamental results on the semimartingale decomposition of an $\mathbb{F}$-local martingale stopped at $\tau$ (see Aksamit, Jeanblanc, 2017, pp.102-103 for proofs). \textcolor{black}{Using this notation, we state the following results.} \par

\begin{thm} Every cadlag $\mathbb{F}$-local martingale $X$  stopped at $\tau$ is a special $\mathbb{G}$-semimartingale with the canonical decomposition 
\begin{align*}
    X_t^{\tau} = \widehat{X}_t + \int_0^{t \land \tau}\frac{d \langle X,m \rangle ^{\mathbb{F}}_s}{Z_{s-}} = \widehat{X}_t + \int_0^{t \land \tau} \frac{d \langle X,n \rangle ^{\mathbb{F}}_s + dJ_s}{Z_{s-}},
\end{align*}
where $\widehat{X}$ is a $\mathbb{G}$-local martingale and $J$ is the $\mathbb{F}$-dual predictable projection of the process $\mathbbm{1}_{\llbracket \tau, \infty \llbracket} \Delta X_{\tau}$. 
\end{thm}

\par
Another decomposition is given by the following theorem. \par

\begin{thm} Every cadlag $\mathbb{F}$-local martingale $X$  stopped at $\tau$ is a $\mathbb{G}$-semimartingale with decomposition 
\begin{align*}
     X_t^{\tau} = \bar{X}_t + \int_0^{t \land \tau}\frac{d [ X,m ]_s}{\widetilde{Z}_s} - (\mathbbm{1}_{\llbracket \widetilde{R}, \infty \llbracket} \Delta X_{\widetilde{R}})^{p,\mathbb{F}}_{t \land \tau},
\end{align*}
where $\bar{X}$ is a $\mathbb{G}$-local martingale and $\widetilde{R} \coloneqq R_{\{\widetilde{Z}_R=0 < Z_{R-}\}}$, where $R \coloneqq $ inf $\{t \vert Z_t = 0\}$. 
\end{thm}

\par
There are several different kinds of random times analyzed in the literature. We briefly introduce them in the following. \par

\begin{defn}A random time $\tau$ is called \textbf{an honest time} if and only if, for every $t > 0$, there exists an $\mathcal{F}_{t-}$-measureable random variable $\tau_t$ such that $\tau = \tau_t$ on $\{\tau < t\}$. 
\end{defn}

\par
\begin{defn} A random time $\tau$ is called a \textbf{$\mathcal{J}$-time} if it satisfies Jacod's absolute continuity criterion.
\end{defn}

\par
\begin{defn} A random time $\tau$ is called an $\mathbb{F}$-\textbf{thin time }if its graph $\llbracket \tau \rrbracket$ is contained in a thin set, i.e. there exists an exhausting sequence of $\mathbb{F}$-stopping times $(T_n)_{n=1}^{\infty}$ with disjoint graphs s.t.  $\llbracket \tau \rrbracket \subset \cup_n \llbracket T_n \rrbracket$. \par
This means $\tau = \infty \mathbbm{1}_{C_0} + \sum_nT_n\mathbbm{1}_{C_n},$ where $C_0 \coloneqq \{ \tau = \infty\}$ and $C_n \coloneqq \{ \tau = T_n < \infty\}$ for $n \geq 1$. 
\end{defn}

\par
\begin{defn} A random time $\tau$ is called an $\mathbb{F}$-\textbf{pseudo-stopping time} if, for any bounded $\mathbb{F}$-martingale $Y$, $\mathbb{E}[Y_{\tau}] = \mathbb {E}[Y_0]$. 
\end{defn}

\par
For progressive enlargements with these random times we have the respective semimartingale decompositions. 

\par
\begin{thm} Let $\tau$ be an $\mathbb{F}$-honest time. Then every cadlag $\mathbb{F}$-local martingale $X$ is a special $\mathbb{G}$-semimartingale with the canonical decomposition
\begin{align*}
    X_t = \widetilde{X}_t + \int_0^{t \land \tau}\frac{d \langle X,m \rangle ^{\mathbb{F}}_s}{Z_{s-}}  - \int_{\tau}^{t \lor \tau}\frac{d \langle X,m \rangle ^{\mathbb{F}}_s}{1 -Z_{s-}},
\end{align*}
where $\widetilde{X}_t$ is a $\mathbb{G}$-local martingale. 
\end{thm} 

\par
\begin{thm} Let $\tau$ be a $\mathcal{J}$- time. Then, any cadlag $\mathbb{F}$-local martingale $X$ is a special $\mathbb{G}$-semimartingale with the canonical decomposition 
\begin{align*}
     X_t = \widehat{X}_t + \int _0^{t \land \tau}\frac{d \langle X,m \rangle ^{\mathbb{F}}_s}{Z_{s-}}  + \int^t_{t \land \tau}\frac{1}{p_{s-}(\tau)}d\langle X, p(x)\rangle ^{\mathbb{F}}_s\vert_{x = \tau},
\end{align*}
where $\widehat{X}_t$ is a $\mathbb{G}$-local martingale. 
\end{thm}

\par
\begin{thm} Let $\tau$ be a thin time. Then $\mathbb{F} \subset \mathbb{G} \subset \mathbb{F}^{\mathcal{C}}$ and $\mathcal{H}^{\prime}$ is satisfied for $\mathbb{G}$, where $\mathbb{F}^{\mathcal{C}}$ is the initial enlargement of $\mathbb{F}$ with the atomic $\sigma$-field $\mathcal{C} \coloneqq \sigma(C_n, n \geq 0)$. Moreover, Then, any $\mathbb{F}$-local martingale $X$ has the following $\mathbb{G}$-semimartingale canonical decomposition 
\begin{align*}
    X_t = \widehat{X}_t + \int _0^{t \land \tau}\frac{d \langle X,m \rangle ^{\mathbb{F}}_s}{Z_{s-}} + \sum_n\mathbbm{1}_{C_n} \int_0^t\mathbbm{1}_{\{s > T_n\}} \frac{d \langle X,z^n \rangle ^{\mathbb{F}}_s}{z^n_{s-}},
\end{align*}
where $\widehat{X}_t$ is a $\mathbb{G}$-local martingale, $z^n_t \coloneqq \mathbb{P}(C_n\vert \mathcal{F}_t)$. 
\end{thm} 

\par
The discussions of random times and their associated Azema supermatingales make frequent use of the following two assumptions. \par
\textbf{Assumption (A)} We say that a random time $\tau$ satisfies assumption (A) if $\tau$ \textbf{a}voids all $\mathbb{F}$-stopping times, i.e. $\mathbb{P}(\tau = \sigma) = 0$ for all $\mathbb{F}$-stopping times $\sigma$.
\par
\textbf{Assumption (C)} We say that assumption (C) is satisfied if all $\mathbb{F}$-martingales are \textbf{c}ontinuous. \par
When both assumptions are satisfied we say that (CA) holds. 

\par
\begin{defn} A process $X$ is said to be of class (D) if the family $\{X_{\tau}\mathbbm{1}_{\{\tau < \infty\}} : \tau \ \text{is a stopping time} \}$ is uniformly integrable. 
\end{defn} 

The following well-known \textbf{multiplicative decomposition} of a cadlag positive supermartingale of class (D) is often useful, although much less than the additive decompositions. \par
\begin{thm}  (Ito-Watanabe) Let $X$ be a positive cadlag supermartingale of class (D). Then there exists a local martingale $N$ and predictable decreasing process $D$ such that $X = ND$, where $D_0 = X_0$. 
\end{thm} 

\par
\begin{thm} (Nikeghbali, Yor, 2006) Let $L$ be an honest time that satisfies (CA). Then there exists a continuous and nonnegative local martingale $(N_t)_{t \geq 0}$ with $N_0 = 1$ and $\lim_{t \to \infty} N_t = 0$, such that 
\begin{align*}
    Z_t \coloneqq \mathbb{P}(L > t \vert \mathcal{F}_t) = \frac{N_t}{S_t},
\end{align*}
where $S_t = \sup_{s\leq t}N_s$ (sometimes also denoted as $\overline{N}_t)$. 
\end{thm} 

\par
\begin{cor} \textcolor{black}{In the notation of the previous theorem}, the supermartingale  $Z_t = \mathbb{P}(L > t \vert \mathcal{F}_t)$ admits the following \textcolor{black}{multiplicative and additive} decompositions:
\begin{align*}
       Z_t &= \frac{N_t}{S_t}, \quad \quad\textcolor{black}{(Ito-Watanabe),}\\
       Z_t &= M_t - A_t, \quad \quad \textcolor{black}{(Doob-Meyer)}.
\end{align*}
where the corresponding terms are related as follows:
\begin{align*}
    N_t &= \text{exp}\left(\int_0^t\frac{dM_s}{Z_s} - \frac{1}{2} \int_0^t\frac{d\langle M \rangle_s}{Z^2_s}\right), \\
    S_t &=\text{exp}(A_t)  
\end{align*}
and
\begin{align*}
    M_t &= 1 + \int_0^t\frac{dN_s}{S_s} = \mathbb{E}(\text{log} S_{\infty} \vert \mathcal{F}_t), \\
    A_t &= \text{log} \ S_t.
\end{align*}
\end{cor} 

\par
We now proceed to some examples with explicit calculations of the Azema supermartingale associated with a given random time and path decompositions of the processes in an enlarged filtration. But first, consider an important theorem that establishes that last passage times are in fact honest times. \par

\begin{thm}  Let $X$ be an $\mathbb{F}$-adapted process and $\tau \coloneqq \text{sup} \{t : X_t \leq a \}$. Then $\tau$ is an honest time. 
\end{thm} 
\begin{proof} Consider $\tau_u \coloneqq \text{sup} \{t \leq u : X_t \leq a \}$ and note that it is $\mathcal{F}_u$-measureable. Trivially, $\tau = \tau_u$ on $\tau < u$, hence the claim. 
\end{proof}

\par
\begin{example} (Nikeghbali, 2006) Let $B$ be a standard Brownian motion, $S_t \coloneqq \text{sup}_{s \leq t}B_s$ and $g \coloneqq \text{sup}\{t : B_t = S_t\}.$ Then from the results above it follows that 
\begin{align*}
  \mathbb{P}(g > t \vert \mathcal{F}_t) = \frac{B_t}{S_t}.  
\end{align*} 
\end{example}

\par
\begin{example} (Nikeghbali, 2006) Let $N_t \coloneqq \text{exp}(2\nu B_t - 2\nu^2t)$ for $\nu>0$, \textcolor{black}{where $B$ is a standard Brownian motion}. We have
\begin{align*}
    S_t &= \text{exp}\left(\text{sup}_{s \leq t}2\nu (B_s - \nu s) \right), \\
    g &= \text{sup}\left\{t :(B_t - \nu t) = \text{sup}_{s \geq 0} (B_s - \nu s)\right\},
\end{align*}
therefore
\begin{align*}
    \mathbb{P}( g > t \vert \mathcal{F}_t) = \text{exp}\left( 2\nu (B_s - \nu s) - \text{sup}_{s \geq 0} (B_s - \nu s)\right).
\end{align*} 
\end{example}

\par
\begin{example} (Nikeghbali, 2006) Consider a Bessel process \textcolor{black}{$R$} of dimension $2(1 - \mu)$ starting from $0$, where $\mu \in (0,1)$. Define $g_{\mu} \coloneqq \text{sup}\{t \leq 1 : R_t = 0\}$. Then the associated Azema supermartingale \textcolor{black}{$Z^{g_{\m}}$} is (we refer to the original paper for a derivation)
\begin{align*}
    Z_t^{g_{\mu}} = \mathbb{P}(g_{\mu} > t \vert \mathcal{F}_t) = \frac{1}{2^{\mu -1}}\Gamma(\mu) \int^{\infty}_{\frac{R_t}{\sqrt{1-t}}} \text{exp}(- \frac{y^2}{2})y^{2\mu - 1}dy,
\end{align*}
and the dual predictable projection \textcolor{black}{$A^{g_{\m}}$} of $\mathbbm{1}_{\{g_{\mu} \leq t\}}$ is 
\begin{align*}
    A_t^{g_{\mu}} = \frac{1}{2^{\mu}\Gamma(1 + \mu)}\int_0^{t\land1}\frac{dL_u}{(1-u)^{\mu}},
\end{align*}
where $L_y \coloneqq \text{sup} \{t : R_t = y \}$. \par
We now apply these results to the case of a standard Brownian motion. Note that when $\mu = 1/2$, $R_t$ can be viewed as $\lvert B_t \rvert$. Then for $g \coloneqq \text{sup}\{t \leq 1 : B_t = 0\}$ we have 
\begin{align*}
     Z_t^g = \mathbb{P}(g > t \vert \mathcal{F}_t) = \sqrt{\frac{2}{\pi}} \int^{\infty}_{\frac{B_t}{\sqrt{1-t}}} \text{exp}(- \frac{y^2}{2})dy,
\end{align*}
and
\begin{align*}
      A_t^g = \sqrt{\frac{2}{\pi}}\int_0^{t\land1}\frac{dL_u}{\sqrt{1-u}}.
\end{align*} 
\end{example}

\par
For the last example, we need an auxillary result provided by the following lemma. \par

\begin{lem} (Nikeghbali, Yor, 2006) \textcolor{black}{Given the Azema supermartingale $Z$ associated with a random time $g$ or $g^{\m}$} \textcolor{black}{as in the examples above}, define the nonincreasing process $(r_t)$ by 
\begin{align*}
    r_t \coloneqq \inf_{u\leq t} Z_u = \text{inf}_{u \leq t} \frac{N_u}{S_u},
\end{align*}
\textcolor{black}{where $N$ and $S$ are as defined in the Ito-Watanabe decomposition above}.
Then 
\begin{align*}
    \rho \coloneqq \text{sup} \left\{ t < g : \frac{N_t}{S_t} = \text{inf}_{u \leq t} \frac{N_u}{S_u} \right\}
\end{align*}
is a pseudo-stopping time and $r_{\rho}$ is uniformly distributed on $(0,1)$.
\end{lem}

\par
\begin{example} In the setup above, let $N_t = B_t$. \textcolor{black}{Consider the initial enlargements $\mF^g$ and $\mF^{\s(S_{T_0})}$ of $\mF$ by $g$ and $\s(S_{T_0})$ respectively.} Then \\
(i) $\frac{B_{\rho}}{S_{\rho}}$ is uniformly distributed on $(0,1)$ \\
(ii) $(B_t)$ is an $\mF^g-$ and $\mF^{\s(S_{T_0})}$-semimartingale with the canonical decomposition 
\begin{align*}
    B_t = \widetilde{B}_t + \int_0^{t \land g}\frac{ds}{B_s} - \int_g^{t \land T_0}\frac{ds}{S_{T_0} - B_s},
\end{align*}
where $\widetilde{B}_t$ is an $\mF^{\sigma(S_{T_0})}$-Brownian motion stopped at the first hitting time of 0, i.e. $T_0$. 
\end{example}

\par
\begin{example} (Cox's construction). Consider a random variable $\Theta$ independent of $\mathcal{F}_{\infty}$ such that $\mathbb{P}(\Theta > t) = \text{exp}(-t)$ for all $t \geq 0$. Let $(\Lambda_t)$ be an $\mathbb{F}$-adapted continuous increasing process with $\Lambda_{\infty} = \infty$. Define $\sigma \coloneqq \text{inf}\{t \geq 0 : \Lambda_t \geq \Theta\}$, and note that it is a pseudo-stopping time since 
\begin{align*}
    Z_t \coloneqq \mathbb{P}(\sigma > t \vert \mathcal{F}_t) = \mathbb{P}(\Lambda_t < \Theta \vert \mathcal{F}_t) = \text{exp}(-\Lambda_t).
\end{align*}
As an important special case, consider $\mathbb{F}$ to be the filtration generated by a Poisson process $N$, and $\tau = \text{inf}\{t \geq 0 : \psi(N_t) \geq \Theta\}$, where $\psi$ is some non-decreasing function $\psi : \mathbb{R}_+ \to \mathbb{R}$ such that $\psi(0) = 0$ and $\lim_{x \to \infty}\psi(x) = \infty$. Then the associated Azema supermartingale is $Z_t = e^{-\psi(N_t)}$ for all $t \in \mR_+$. Note that it is decreasing but not predictable. The Azema supermartingale admits the decomposition $Z = \mu - A^p$, where
\begin{align*}
    \mu_t = 1 + \int_0^t\left(e^{-\psi(N_{s-}+1)} - e^{-\psi(N_{s-})}\right)dM_s, \\
    A^p_t = \int_0^t\lambda\left((e^{-\psi(N_s+1)} - e^{-\psi(N_s + 1)}\right)ds,
    \end{align*}
and the process 
\begin{align*}
    M_t^{(\tau)} \coloneqq A_t - \int_0^{t \land \tau}e^{\psi(N_s)}dA^p_s = A_t - \int_0^{t \land \tau}\lambda\left((e^{-\psi(N_s+1)} - e^{-\psi(N_s + 1)}\right)ds
\end{align*}
is a $\mathbb{G}$-martingale, which implies $\tau$ is a totally inaccessible $\mathbb{G}$-stopping time. Indeed, $Z_t = Z_0 + \sum_{0<s\leq t}e^{-\psi(N_s)} - e^{-\psi(N_{s-})} = 1 + \int_0^t\left(e^{-\psi(N_{s-}+1)} - e^{-\psi(N_{s-})}\right)dM_s$, and $M_t = N_t - \lambda t$. The proof of the second statement can be found in (Aksamit et al 2018, p.1243).
\end{example}

\subsection{Connections to Mathematical Finance}

The germ of the main ideas that relate the theory of enlargement to finance can be found, inter alia, in an initial enlargement setting where the natural filtration of a stochastic process is enlarged by the $\sigma$-algebra generated by the terminal value, i.e. if an insider has at time 0 some information about an event which will occur in the future, then that information can be used to make a profit. \par
More generally, the cornerstone of modern mathematical finance is the notion of (absence of) \textbf{arbitrage}. Properly defined, it can be shown to be equivalent to the existence of an \textbf{equivalent probability measure}, under which the price process of an asset becomes a (local) martingale. The high level idea is that enlarging a filtration can potentially create arbitrage opportunities. Specifically, \textbf{the First Fundamental Theorem of Asset Pricing} claims that for a price process $S$ that is a general semimartingale the existence of an equivalent probability measure under which $S$ is a $\sigma$-martingale is equivalent to no arbitrage (defined as \textbf{no free lunch with vanishing risks, NFLVR} - see the next section). This leads us to the idea that if the enlargement of the underlying filtration destroys the semimartingale property of $S$, then there are very likely to be arbitrage opportunities for those with access to this enlarged information set. The semimartingale property does not, however, guarantee no arbitrage. For example, it is preserved by the addition of $B_T$ to the natural filtration of a standard Brownian motion. However, in the Black-Scholes model, knowing $B_T$ is equivalent to knowing the terminal value of the stock price $S_T$, which obviously implies arbitrage (Ouwehand, 2007). \par
Before we proceed to actual applications in finance, a word of caution is in order. While the theory of enlargement of filtrations is a powerful tool in the study of information-related aspects of finance, it is by no means general and does not exhaust the topic. As an example, (Ernst et al, 2017) consider the "enlarged" filtration $(\mathcal{F}'_t) \coloneqq (\mathcal{F}_{t+a})$ for some fixed $a > 0$, taken to represent the "foresight" of the insider.  In this new filtration, the standard Brownian motion $B$ is not even a semimartingale. (Ernst et al, 2017) conclude that "none of the results from enlargement of filtrations
will help us here $-$ the problem addressed is concrete, challenging, and not amenable to general theory".\par
The discussion above, specifically the semimartingale decompositions given for the processes in the enlarged filtration, points to a number of interesting connections between enlargement of filtrations and several core ideas in stochastic analysis. From a very high level vantage point, semimartingales can informally be understood as a generalization of the "noise-signal" decomposition of a random phenomenon evolving in time. This is also the idea present in Doob's decomoposition of a supermartingale. Note that enlargement results in the "noise" process (the local martingale term in our analogy) acquiring a trend (the finite variation term, or the "signal") or becoming so "wild" that it even ceases to be a semimartingale. The appearance of the drift term under enlargement bears a resemblance to the Girsanov transform and change of measure techniques in general. This resemblance is not superficial and indeed the two ideas are closely intertwined. Moreover, initial enlargement of filtrations is naturally linked to the well established theory of bridges. Indeed, even in the classic example the standard Brownian motion turns into a Brownian bridge if the natural filtration is expanded to include the terminal value of the process (i.e. the "arrival" point for the bridge process). This connection is also not merely a useful analogy, but rather a robust relationship. As M.Yor pointed out (Mansuy, Yor, p.36), almost all initial enlargement formulas obtained by appending $\sigma$-algebras generated by different random variables "may be interpreted in terms of semimartingale decompositions of Markovian bridges". Given these considerations, we provide the necessary definitions and theorems that will be useful in the following discussion. More details can be found in, e.g., (Protter, 2005). \par

\begin{thm} (Doob-Meyer Decomposition). Let $Z$ be a cadlag supermartingale. Then $Z$ has a decomposition $Z = Z_0 + M - A$, where $M$ is a local martingale and $A$ is a predictable increaing process, and $M_0 = A_0 = 0$. Such a decomposition is unique. 
\end{thm}

\par
\begin{thm}(Girsanov). Let $\mathbb{P} \sim \mathbb{Q}$. Let $X$ be a semimartingale under $\mathbb{P}$ with decomposition $X = M + A$. Then $X$ is also a semimartingale under $\mathbb{Q}$ with the decomposition $X = L + C$, where  
\begin{align*}
    L_t = M_t - \int_0^t \frac{1}{Z_s}d[Z,M]_s
\end{align*}
is a $\mathbb{Q}$-local martingale, $Z_t = \mathbb{E}_{\mathbb{P}}(\frac{d\mathbb{Q}}{d\mathbb{P}} \vert \mathcal{F}_t)$, and $C = X - L$ is a $\mathbb{Q}$-finite variation process. 
\end{thm}

\par
\begin{thm}(Lenglart-Girsanov). Let $X$ be a $\mathbb{P}$-local martingale with $X_0 = 0$. Let $\mathbb{Q}$ be a probability measure such that $\mathbb{Q} 
\ll \mathbb{P}$, and let $Z_t = \mathbb{E}_{\mathbb{P}}(\frac{d\mathbb{Q}}{d\mathbb{P}} \vert \mathcal{F}_t)$, $R = \text{inf}\{ t > 0 : Z_t = 0, Z_{t-} > 0\}$, $U_t = \Delta X_R\mathbbm{1}_{\{t \geq R\}}$, and $\widetilde{U}_t$ is the predictable copensator of $U_t$ such that $U - \widetilde{U}$ is a $\mathbb{P}$-local martingale. Then 
\begin{align*}
    X_t - \int_0^t \frac{1}{Z_s}d[Z,M]_s + \widetilde{U}_t
\end{align*}
is a $\mathbb{Q}$-local martingale. 
\end{thm}

\par
\begin{defn} A \textbf{Markov bridge} is a process obtained by conditioning a Markov process $X$ to start in some state $x$ at time 0 and arrive at some state $z$ at time $t$. We call this process the $(x,t,z)$-bridge derived from $X$. 
\end{defn}

\bigskip

\section{Enlargement of Filtrations and Different Notions of Arbitrage}

\subsection{Preliminaries and Definitions}

In the sequel, we work on a given filtered probability space (stochastic basis) $(\Omega, \mathcal{F}, \mathbb{F}, \mathbb{P})$, where the filtration $\mathbb{F} = (\mathcal{F}_t)_{0 \leq t \leq T}$ satisfies the usual hypotheses of right-continuity and augmentation by $\mathbb{P}$-null sets, $T \in (0, \infty]$ and $\mathcal{F}_T = \mathcal{F}$. We assume the price process of $n$ risky assets to be a general positive semimartingale $S = (S_1(t), ..., S_n(t))^\top$.\par
We denote by $\mathcal{L}(S)$ the set predictable processes integrable with respect to $S$, and by $\mathcal{O}$ the set of optional processes. The set of \textit{admissible self-financing trading strategies} is defined to be \\
$\mathcal{A}(x) = \{(\alpha_0, \alpha) \in (\mathcal{O}, \mathcal{L}(S)) : X_t = \alpha_0(t) + \alpha_t \cdot S, \exists c \leq 0 , X_t = x + \int_0^t\alpha_u \cdot dS_u \geq c, \forall t \in [0, T]\}$, where $X_t = x + \int_0^t\alpha_u \cdot dS_u$ is the (normalized) value process. We will also use  $\alpha \bullet S$ to denote the stochastic integral process $(\int_0^t\alpha_u \cdot dS_u)_{0 \leq t \leq T}$.
We begin by recalling a number of different definitions that clarify the notion of arbitrage (see Jarrow, 2021, for more details). \par

\begin{defn} An admissible s.f.t.s. $(\alpha_0, \alpha) \in \mathcal{A}(x)$ is called a (simple) \textit{arbitrage opportunity} if \\
(i) $X_0 = x = 0$ \\
(ii) $X_T = \int_0^T\alpha_u \cdot dS_u \geq 0$ \quad $\mathbb{P}$-a.s. \\
(iii) $\mathbb{P}(X_T > 0) > 0.$ 
\end{defn}

\par
\begin{defn} A sequence of admissible s.f.t.s.'s $(\alpha_0, \alpha)_n \in \mathcal{A}(x)$ with associated value processes $X^n$ is said to generate \textit{unbounded profits with bounded risk} if \\
(i) $X_0^n = x > 0$ \\
(ii) $X_t^n \geq 0$ $\mathbb{P}$-a.s. $\forall t \in [0,T]$ \\
(iii) the sequence $(X^n_T)_n \geq 0 $ is unbounded in probability, i.e. \mbox{$\lim_{m \to \infty}(\sup_n \mathbb{P}(X^n_T > m)) > 0.$} 
\end{defn}

\par
\begin{defn} An $\mathcal{F}_T$-measurable random variable $\xi$ is said to generate \textit{an arbitrage opportunity of the first kind} if \\
(i) $\mathbb{P}(\xi \geq 0) = 1$ \\
(ii) $\mathbb{P}(\xi > 0) >0$ \\
(iii) $\forall x > 0$ there exists an admissible s.f.t.s. $(\alpha_0, \alpha) \in \mathcal{A}(x)$ with the value process $X^{x,\alpha} \geq 0$ such that $X^{x,\alpha}_T \geq \xi$. 
\end{defn}

\par
\begin{defn} \textit{A free lunch with vanishing risk} is a sequence of admissible s.f.t.s.'s $(\alpha_0, \alpha)_n \in \mathcal{A}(x)$ with $x \geq 0$, value processes $X^n_t$, admissibility lower bounds $c_n \leq 0$ s.t. $\exists c \leq c_n$ for all $n$, and an $\mathcal{F}_T$-measurable random variable $\xi \geq x$ s.t. $\mathbb{P}(\xi > x) > 0$ and $X^n_T \to \xi$ in probability. 
\end{defn}

\par
If no arbitrage opportunities of the above kind exist we say that NA, NUPBR, NA1 and NFLVR holds respectively. We may also use the following equivalent formulation of the definition of NFLVR (Karatzas, Kardaras, 2007): \par

\begin{defn} A sequence $(\alpha_0, \alpha)_n$ of admissible s.f.t.s. with value processes $X^n_t$ generates \textit{a free lunch with vanishing risk} if there exist an $\epsilon > 0$ and an increasing sequence $(\delta_n)_{n\in\mathbb{N}}$ s.t. $0 \leq \delta_n \nearrow 1$ and $\mathbb{P}(X_T^n > -1 + \delta_n) = 1$ and $\mathbb{P}(X_T^n > \epsilon) \geq \epsilon$, for all $ n \in \mathbb{N}$. 
\end{defn}

\par
These notions of arbitrage are closely connected. For example, it can be shown (Kardaras, 2012; Fontana, 2014) that NA1 is equivalent to NUPBR. Moreover, NFLVR is equivalent to both NA and NUPBR being satisfied (Delbaen Schachermayer, 1994; Fontana, 2014). We also recall that the most general version of the Fundamental Theorem of Asset Pricing states that for a general semimartingale $S$ NFLVR is equivalent to the existence of a probability measure $\mathbb{Q}$ on $(\Omega, \mathcal{F})$ equivalent to $\mathbb{P}$ under which $S$ is a $\sigma$-martingale. Since we assumed $S$ to be nonnegative and every nonnegative $\sigma$-martingale is a local martingale, then NFLVR is equivalent to the existence of a probability measure $\mathbb{Q}$ on $(\Omega, \mathcal{F})$ equivalent to $\mathbb{P}$ under which $S$ is a local martingale. \par
The following definitions will also be useful in the following discussion. Consider a general semimartingale financial market given by $(S, \mathbb{F}, \mathbb{P})$. \par

\begin{defn} A stricly positive $\mathbb{F}$-local martingale $(L_t)$ with $L_0 = 1$ and $L_{\infty} > 0$ a.s. is called \textit{a local martingale deflator}, if the process $(S_tL_t)$ is an $\mathbb{F}$-local martingale. 
\end{defn}

\par
\begin{rem} NA1 is equivalent to the existence of a local martingale deflator. 
\end{rem}

\par
\begin{defn} $\mathbb{Q} \coloneqq L_{\infty} \mathbb{P}$ is \textit{an equivalent local martingale measure}, if there exists a local martingale deflator $(L_t)$ which is a uniformly integrable martingale closed by $L_{\infty}.$
\end{defn}

\par
With these notions established, we now proceed to explicate the relationship between different notions of arbitrage and different types of enlargement of the underlying filtration.\par

\subsection{The relationship between EoF and NFLVR}

We begin by observing that the first FTAP establishes the connection between arbitrage opportunities encapsulated in the notion of NFLVR and the existence of an equivalent local martingale measure for $S$. Thus, for a general expansion $\mathbb{G}$ of $\mathbb{F}$ we have the following observation (Neufcourt, 2017, p.37, Corollary 2.3). \par

\begin{thm} Suppose that \textcolor{black}{a semimartingale} $S = S_0 + M + A$ satisfies NFLVR in $\mathbb{F}$, \textcolor{black}{where $M$ is a local martingale and $A$ is a process of finite variation}. Then $S$ satisfies NFLVR in $\mathbb{G}$ if and only if there exists an equivalent probability measure under which $M$ is a $\mathbb{G}$-semimartingale. 
\end{thm}

Consider an initial enlargement setup. The proof of the following theorem can be found in (Aksamit, Jeanblanc, 2017, p.97). \par

\begin{thm} Assume that $\zeta$ satisfies Jacod's equivalence condition $\mathcal{E}$ under $\mathbb{P}$. Let $[0,T]$ be a finite time horizon. If $S$ satisfies NFLVR in $(\mathbb{F}, \mathbb{P})$, then $S$ satisfies NFLVR in $(\mathbb{F}^{\sigma(\zeta)}, \mathbb{P}).$ 
\end{thm}

\par
It is well-known (Jacod, 1985) that in an initial enlargement setting under $\mathcal{J}$ there exists a jointly measurable process $\alpha^x$ with $(x, t) \in \mathbb{R} \times [0,T)$ such that $\int_0^t(\alpha_u^x)^2du < \infty$, $\langle p^x, B \rangle_t = \int_0^t p^x_u \alpha_u^x du$, \textcolor{black}{$B$ is a standard $\mF$-Brownian motion}
and $\widetilde{B}_t = B_t - \int_0^t \alpha_u^x du$ is a $\mathbb{G}$-Brownian motion, where $p^x$ is the conditional density process defined above that can also be written in the stochastic exponential form as $p^x_t = \mathcal{E}\left(-\int_0^t \alpha_u^x dB_u\right)$. The following result is a standard tool for showing that a semimartingale model violates NFLVR (see Ankirchner, Imkeller 2005 for a proof). \par

\begin{thm} If $\mathbb{P}\left(\int_0^T(\alpha^x_t)^2dt = \infty \right) > 0$, then $S$ admits a FLVR in $\mathbb{G} \coloneqq \mathbb{F} \lor \sigma(\zeta)$. \par
A corollary of this result can be used to check if NFLVR holds in an initially enlarged filtration. 
\end{thm}

\par
\begin{cor} If the process $\alpha^x$ satisfies $\mathbb{E}\left[\text{exp}(\frac{1}{2}\int_0^T(\alpha^x_t)^2dt)\right] < \infty$, then $S$ satisfies NFLVR in $\mathbb{G} \coloneqq \mathbb{F} \lor \sigma(\zeta)$. 
\end{cor}
\begin{proof}  
 Novikov's condition and the Girsanov theorem guarantee the existence of the ELMM, hence NFLVR holds true by the first FTAP. 
 \end{proof}

\par 
Consider a financial market given by $(S, \mathbb{F}, \mathbb{P})$ where NFLVR holds. Recall the stopping time $R$ in the notation above $R = T_0^Z = \text{inf}\{ t > 0 : Z_t = 0\}$. \par

\begin{thm} (Kreher, p.22) Let $T$ be an $\mathbb{F}$-stoppping time, and $\sigma$ a random time. Assume (A) is satisfied. \textcolor{black}{Let $Z$ be the Azema supermartingale associated with $\s$}. If $\mathbb{P}(T_0^Z  \leq T) = 0$, then NFLVR holds in the progressively enlarged filtration $\mathbb{G}$ on $[0, \sigma \land T]$.
\end{thm}

\par
It can be shown that the process $( \frac{\rho_{t \land \sigma}}{N_{t \land \sigma}})_{t \geq 0}$ is a local martingale deflator for $(S_{t \land \sigma})$ in $\mathbb{G}$, i.e. NA1 holds in $\mathbb{G}$ on $[0,\sigma]$, where $\rho_t = \mathbb{E}_{\mathbb{P}}(\rho \vert \mathcal{F}_t)$, i.e. the Radon-Nikodym density process \textcolor{black}{in the original $\mF$-market and $N$ is the corresponding local martingale term in the Ito-Watanabe decomposition $Z=ND$, as given above. Here $D$ is the predictable decreasing process given by $D=\frac{1}{\ov{N}}, \ov{N}_t\coloneqq \sup_{s\leq t} N_s$}. We omit the statement since more general results on NA1 stability will be provided in the next subsection. We can, however, apply this result with respect to conditions that would ensure NFLVR. Specifically, the following theorem provides a sufficient and necessary criterion for $( \frac{\rho_{t \land \sigma}}{N_{t \land \sigma}})_{t \geq 0}$ to be a uniformly integrable martingale on $[0, \sigma \land T]$, where $T$ is an $\mathbb{F}$-stopping time and $\sigma$ is a random time. \par

\begin{thm} (Kreher, p.26) Let $T$ be an $\mathbb{F}$-stopping time, and $\sigma$ a random time. Then,
\begin{align*}
    \left( \frac{\rho_{t \land \sigma}}{N_{t \land \sigma}}\right)_{t \geq 0} \in \mathcal{M}_{\text{u.i.}}(\mathbb{P}, \mathbb{G}) \Leftrightarrow \mathbb{E}_{\mathbb{P}}\left(D_{\infty}\mathbbm{1}_{T_0^N \leq T}\right) = 0.
\end{align*}
\end{thm}

\par
This theorem has two immediate implications. One is a different proof of the stability of NFLVR on $\mathbb{G}$ on $[0, \sigma \land T]$ given above, the other is the following corollary. \par

\begin{cor} If $D_{\infty} = 0$ a.s., then NFLVR holds in $\mathbb{G}$ on $[0, \sigma]$. 
\end{cor}

\par
We note that the condition is fulfilled for every pseudo-stopping time since $D_{\infty} = 1 - A_{\infty} = 1 - 1 = 0$. Emery provided an example that shows other stopping times can also satisfy the condition $D_{\infty} = 0$ a.s. and hence admit an ELMM in $\mathbb{G}$ on $[0, \sigma]$.  
\par
We now look at some examples of the stability of NFLVR under filtration enlargement. We begin from some well-known basic examples and proceed to more specific (and involved) ones. \par

\begin{example} (Karatzas, Pikovski, 1996; see also Aksamit, Jeanblanc, 2017) Consider a portfolio optimization problem in a Black-Scholes setting. Assume the dynamics of a risky asset to be driven by the SDE of the form $dS_t = S_t \mu dt + S_t \sigma dB_t$, where $B$ is a Brownian motion and $\mu$ and $\sigma$ are constants. Our goal is to maximize expected logarithmic utility $\mathbb{E}(\text{ln} X^{\pi, x}_T)$ for a given $T$, where the superscript stresses the dependence on the proportion of wealth invested in the risky asset $\pi$ and initial wealth $x$. It is easily verified that the dynamics of the wealth process are given by 
\begin{align*}
    \text{ln}X^{\pi,x}_t = \text{ln}x + \int_0^t (r - \frac{1}{2}\pi_s^2\sigma^2 + \theta \pi_s\sigma)ds + \sigma \int_0^t\pi_sdB_s,
\end{align*}
where $\theta \coloneqq \frac{\mu - r}{\sigma}$ (the so-called the market price for risk) and $\pi_t \coloneqq \frac{\alpha_tS_t}{X_t}$, and $\alpha_t$ is the position in the risky asset. Consider admissible strategies for which $\pi \bullet B$ is a true martingale, i.e. $\pi$'s s.t. $\mathbb{E}(\int_0^t\pi^2_sds) < \infty$. Taking expectations, and noting that the stochastic integral on the right vanishes, we get the following expression
\begin{align*}
    \mathbb{E}(\text{ln}X^{\pi,x}_t) = \text{ln}x + \int_0^t \mathbb{E}((r - \frac{1}{2}\pi_s^2\sigma^2 + \theta \pi_s\sigma))ds,
\end{align*}
which is optimized by setting $\pi^*_s = \frac{\theta}{\sigma}$, where $\pi$ belongs to the set of admissible strategies. With a little more effort it can be argued that this solution holds for all admissible strategies, including those for which $\pi \bullet B$ is only a local martingale. \par
We now consider the optimization problem for the insider who has access to $B_{T^*}$, i.e. the terminal value of the Brownian motion that generates the randomness in the model. It is a standard result that in this
enlarged filtration the process $\beta_t = B_t - \int_0^t \frac{B_1 - B_s}{1-s}ds$ is a $\mathbb{G}$-Brownian motion. Thus the dynamics of the wealth process on $(\Omega, \mathcal{F}, \mathbb{G}, \mathbb{P})$ for $T < T^*$ are given by the following:
\begin{align*}
    \text{ln}X^{\pi,x}_T = \text{ln}x + \int_0^T (r - \frac{1}{2}\pi_s^2\sigma^2 + \widetilde{\theta}_s \pi_s\sigma)ds + \sigma \int_0^T\pi_sd\beta_s,
\end{align*}
where $\widetilde{\theta}_s = \theta + \frac{B_{T^*} - B_s}{T^* - s}$
Standard calculations show that the optimal portfolio for the insider (i.e. that maximizes expected terminal log-utility for $T < T^*$) is given by $\pi^* = \frac{\widetilde{\theta}_s}{\sigma}.$ Denoting by $V^{\mathbb{F}}$ and $V^{\mathbb{G}}$ the respective optimal values of the objective functions under different filtrations, it follows that for $T < T^*$ we have $V^{\mathbb{G}}(x) = V^{\mathbb{F}}(x) + \frac{1}{2}\mathbb{E}(\int_0^T(\frac{B_{T^*} - B_s}{T^* - s})^2ds = V^{\mathbb{F}}(x) + \frac{1}{2}\text{ln}\frac{T^*}{T^* - T}.$ From this expression it is obvious that the objective function explodes to infinity as $T \uparrow T^*$. In other words, at $T = T^*$ the integral $(\int_0^{T^*}\mathbb{E}(\frac{B_{T^*} - B_s}{T^* - s})^2ds$ diverges and there is no EMM in $\mathbb{G}$. This measure, however, does exist on $[0, T]$ for all $T < T^*$, i.e. inside information creates arbitrage opportunities only at time $T^*$. This toy example is elaborated on in a more general setting in the following example. 
\end{example}

\par
\begin{example} (Imkeller, Perkowski, 2013) Consider a stochastic basis $(\Omega, \mathcal{F}, \mathbb{F}, \mathbb{P})$,  where $\mathcal{F}_{\infty}  = \bigvee_{t \geq 0}\mathcal{F}_t$, and $\mathcal{F}_0$ is trivial, i.e. $\forall A \in \mathcal{F}_0$ we have $\mathbb{P}(A) \in \{0,1\}$. Assume $S$ to be a general semimartingale and the market to be complete (i.e. every contingent claim can be hedged). Take an $\mathcal{F}_{\infty}$-measureable random variable $L$ not a.s. constant and consider the initially enlarged filtration $\mathbb{G} = (\mathcal{G}_t \coloneqq \mathcal{F}_t \lor \sigma(L), t \geq 0)$. Since $L$ is not constant, there exists $A \in \sigma(L)$ s.t. $\mathbb{P}(A) \in (0,1).$ Assume $\mathbb{Q}$ is the equivalent $\mathbb{G}$-local martingale measure for $S$. Consider the $\mathbb{F}$-martingale $N_t = \mathbb{E}_{\mathbb{Q}}(\mathbbm{1}_A \vert \mathcal{F}_t)$, $t \geq 0$. By completess, $\mathbbm{1}_A$ can be replicated, i.e. there exists an $\mathbb{F}$-predictable admissible trading strategy $\alpha$ such that $N_t = \mathbb{E}_{\mathbb{Q}}(\mathbbm{1}_A) + \int_0^t\alpha_sdS_s$. This implies that $\int_0^t\alpha_sdS_s$ is a bounded $\mathbb{G}$-local martingale under $\mathbb{Q}$, hence a true martingale. Consider an event $A \in \mathcal{G}_0$ with positive measure under $\mathbb{Q}$ and note
\begin{align*}
    0 = \mathbb{E}_{\mathbb{Q}}(\mathbbm{1}_{A^c}\mathbbm{1}_A) = \mathbb{E}_{\mathbb{Q}}\left(\mathbbm{1}_{A^c}(\mathbb{E}_{\mathbb{Q}}(\mathbbm{1}_A) + \int_0^{\infty}\alpha_sdS_s)\right) = \mathbb{Q}(A^c)\mathbb{Q}(A) > 0,
\end{align*}
a contradiction. Hence an ELMM does not exist and NFLVR is violated.
\end{example}

\par
Nevertheless, not any information about future values of the price process creates arbitrage opportunities. (D'Auria, Salmeron, 2019) provide an example where having some information about the terminal value of the risky asset (defined as some function of $B_T$, for example) does not necessarily lead to arbitrage opportunities. Since the calculations are rather tedious, we only follow the main steps in their derivation. \par

\begin{example} Consider an initial enlargement setting with $\mathbb{G} \coloneqq \mathbb{F}^{\sigma(\zeta)}$, where the additional information is provided by $\zeta = \mathbbm{1}_{\{B_T \in \cup_{k = -\infty}^{+\infty}[2k-1,2k]\}}$ that represents knowledge whether the Brownian motion will end up in a particular infinite union of intervals of size 1. The new drift term in the $\mathbb{G}$-semimartingale decomposition of the price process is given by 
\begin{align*}
    \alpha^x_t = \begin{cases}
    \frac{1}{\sqrt{T - t}}\frac{\sum_{k = -\infty}^{+\infty}\Phi'(\frac{2k- B_t}{\sqrt{T - t}}) - \Phi'(\frac{2k- 1 - B_t}{\sqrt{T - t}})} {\sum_{k = -\infty}^{+\infty}\Phi'(-\frac{2k- B_t}{\sqrt{T - t}}) - \Phi'(\frac{2k- 1 - B_t}{\sqrt{T - t}})} \quad \text{when} \quad x = 0, \\
    \frac{1}{\sqrt{T - t}}\frac{\sum_{k = -\infty}^{+\infty}\Phi'(\frac{2k - 1 - B_t}{\sqrt{T - t}}) - \Phi'(\frac{2k - B_t}{\sqrt{T - t}})} {\sum_{k = -\infty}^{+\infty}\Phi'(-\frac{2k- B_t}{\sqrt{T - t}}) - \Phi'(\frac{2k- 1 - B_t}{\sqrt{T - t}})} \quad \text{when} \quad x = 1,
    \end{cases}
\end{align*}
where $\Phi$ is the distribution of a standard Gaussian random variable. Then it can be shown that $\alpha^x$ satisfies Novikov's condition $\mathbb{E}\left[\text{exp}(\frac{1}{2}\int_0^T(\alpha^x_t)^2dt)\right] < \infty$, hence by the corollary above $S$ satisfies NFLVR in $\mathbb{G}$.
\end{example}

\par
\begin{example} (Kreher, 2016, p. 25) Consider a financial market model $(S_t, \mathbb{F}, \mathbb{P})$. Let $\sigma$ be a pseudo-stopping time bounded by 1. Since $1 - Z_1 = A_1 = 1$, then $\mathbb{P}(T_0^Z \leq 1) = 1$, where $Z_t = \mathbb{P}(\sigma > t \vert \mathcal{F}_t)$ is the Azema supermatingale and $T_0^Z$ is the first hitting time of $0$ by $Z$. Moreover, $\mathbb{E}(\rho_{\sigma}) = 1$ and $N \equiv 1$, where $\rho$ is the ELMM for the market, i.e. $\rho_t \coloneqq \mathbb{E}_{\mathbb{P}}(\rho \vert  \mathcal{F}_t)$, and $Z = ND$ is the Ito-Watanabe decomposition of the Azema supermartingale (i.e. $N$ is a continuous non-negative local martingale starting from $N_0 = 1$ and $D$ is a continuous decreasing process such that both $N$ and $D$ are constant on the set $\{ Z = 0 \}$). Therefore $(\rho_{t \land \sigma})$ is a local martingale deflator with $\mathbb{E}(\rho_{\sigma}) = 1$, and thus a uniformly integrable martingale. Therefore there is an ELMM for the enlarged market $(S_{t \land \sigma}, \mathbb{G}, \mathbb{P})$, where $\mathbb{G} = (\mathcal{G}_{t \land \sigma} : t \geq 0)$, and hence NFLVR holds on $[0, \sigma] = [0, \sigma \land 1]$.
\end{example}

\par
\begin{example} (Emery, also Kreher, p. 27). Let $W$ be an $\mathbb{F}$-Brownian motion and define $\sigma \coloneqq \text{sup}\{t \leq 1 : 2W_t = W_1\}$. Then the associated Azema supermartingale is 
\begin{align*}
    Z_t = \sqrt{\frac{2}{\pi}}\int_{\frac{\lvert W_t \rvert}{\sqrt{1 - t}}} ^{\infty}x^2e^{-x^2/2}dx = m_t - \sqrt{\frac{2}{\pi}}\int_0^t\frac{\lvert W_u \rvert}{(1-u)^{3/2}}\text{exp}\left(-\frac{W^2_u}{2(1-u)}\right)du,
\end{align*}
where $m \not\equiv 1$ \textcolor{black}{is the local martingale part in the Doob-Meyer decomposition of $Z$}. Now $\forall n \in \mathbb{N}$ define the set
\begin{align*}
    B_n = \left\{\lvert W_u \rvert > \sqrt{\frac{2}{n}} \quad \forall u \in [1 - \frac{1}{n}, 1] \right\}.
\end{align*}
Note that $\lim_{n \to \infty} \mathbb{P}(B_n) = \mathbb{P}(W_1 \not= 0) = 1$, and one the set $B_n$ for all $u \in [1 - \frac{1}{n}, 1]$ it holds that $\frac{\lvert W_u \rvert} {\sqrt{1-u}} > 2$, therefore
\begin{align*}
    \frac{1}{2}\int_{\frac{\lvert W_u \rvert}{\sqrt{1 - u}}} ^{\infty}x^2e^{-x^2/2}dx \leq \int_{\frac{\lvert W_u \rvert}{\sqrt{1 - u}}} ^{\infty}(x^2-1)e^{-x^2/2}dx = \frac{\lvert W_u \rvert}{\sqrt{1-u}}\text{exp}\left(-\frac{W^2_u}{2(1-u)}\right),
\end{align*}
and on $B_n$ we have, \textcolor{black}{for $A=1 - D$ in the notation above,}
\begin{align*}
    \int_0^1\frac{dA_t}{Z_t} &\geq \int_{1-\frac{1}{n}}^1 \frac{dA_t}{Z_t} \geq \int_{1-\frac{1}{n}}^1 \frac{dA_t}{\sqrt{\frac{2}{\pi}} \int_{\frac{\lvert W_t \rvert}{\sqrt{1 - t}}} ^{\infty}x^2e^{-x^2/2}dx} \\ &\geq \int_{1-\frac{1}{n}}^1 \frac{dA_t}{\sqrt{\frac{2}{\pi}} \frac{\lvert W_t \rvert}{\sqrt{1-t}}\text{exp}\left(-\frac{W^2_t}{2(1-t)}\right)} = \frac{1}{2}\int_{1-\frac{1}{n}}^1 \frac{dt}{1-t} = \infty.
\end{align*}
Thus it follows that on the set $B_n$ we have $D_{\infty} = D_1 = \text{exp}\left(-\int_0^1\frac{dA_t}{Z_t}\right) = 0$, and hence by the monotone convergence theorem we get
\begin{align*}
    \lim_{n \to \infty}\mathbb{E}_{\mathbb{P}}\left(D_{\infty}\mathbbm{1}_{B_n}\right) = \mathbb{E}_{\mathbb{P}}\left(D_{\infty}\right) = 0 \Leftrightarrow D_{\infty} = 0 \quad \mathbb{P}-\text{a.s.}
\end{align*}
By the corollary stated above we conclude that NFLVR holds in $\mathbb{G}$ on $[0, \sigma]$.
\end{example}

\par
\begin{example} (Imkeller, 2001) Consider the 1-dimensional canonical Wiener space $(\Omega, \mathcal{F}, \mathbb{P})$ equipped with the canonical Wiener process $W = (W_t)_{t \geq 0}$ (i.e. the coordinate process), where $\Omega = C(\mathbb{R}_+, \mathbb{R})$ is the set of continuous functions on $\mathbb{R}_+$ starting at 0, $\mathcal{F}$ is the $\sigma$-algebra of Borel sets with respect to uniform convergence on compact subsets of $\mathbb{R}_+$, and $\mathbb{P}$ is the Wiener measure. 
A financial market is defined to be a pair $(b, \sigma)$, where $b$ is a progressively measureable mean rate of return process such that $\int_0^1 \lvert b_t \rvert dt < \infty$ $\mathbb{P}$-a.s., and $\sigma$ is the progressively measureable volatility process    such that $\int_0^1 \sigma_t^2 dt < \infty$ $\mathbb{P}$-a.s. This pair determines the (stock) price process given by 
\begin{align*}
    \frac{dS_t}{S_t} = \int_0^tb_udu + \int_0^t \sigma_u dW_u
\end{align*}
A progressively measureable process $\pi$ is a portfolio process if $\int_0^1 \lvert \pi_t b_t \rvert dt < \infty$ $\mathbb{P}$-a.s. and $\int_0^1 \lvert \pi_t \sigma_t \rvert^2 dt < \infty$ $\mathbb{P}$-a.s. We also define the excess yield process $R$ and gains process $C$ by $dR_t = \frac{dS_t}{S_t}$ and $G_t = \int_0^t \pi_u dR_u$ for all $t \in [0,1]$.
Consider an honest time $L$ and a progressive enlargement of $\mathbb{F}$ defined as $\mathbb{G} \coloneqq \mathbb{F} \lor \sigma(L)$. Then the process 
\begin{align*}
    \widetilde{W}_t = W_t - \int_0^t \mathbbm{1}_{[0, L]}(s)\frac{\frac{d}{dt}\langle M^{L}, W \rangle_s}{Z^L_{s-}}ds + \int_0^t \mathbbm{1}_{(L, 1]}(s)\frac{\frac{d}{dt}\langle M^{L}, W \rangle_s}{1 -Z^L_{s-}}ds
\end{align*}
is a $\mathbb{G}$-Wiener process, \textcolor{black}{where $M^L$ is the martingale in the Doob-Meyer decomposition of the Azema supermartingale $Z^L$ associated with $L$}. Set
\begin{align*}
    \alpha_t = \int_0^t \mathbbm{1}_{[0, L]}(s)\frac{\frac{d}{dt}\langle M^{L}, W \rangle_s}{Z^L_{s-}}ds - \int_0^t \mathbbm{1}_{(L, 1]}(s)\frac{\frac{d}{dt}\langle M^{L}, W \rangle_s}{1 -Z^L_{s-}}ds.
\end{align*}
From the point of view of the insider the new financial market is given by the pair $(\tilde{b}, \tilde{\sigma})$, where $\tilde{b}_t = b_t + \sigma_t \alpha_t$ and $\tilde{\sigma}_t = \sigma_t$ for all $t \in [0,1]$ with respect to the $\mathbb{G}$-Wiener process $\widetilde{W}$. Note that $R_t = \int_0^tb_udu + \int_0^t \sigma_u dW_u = \int_0^t \tilde{b}_udu + \int_0^t \sigma_u d\widetilde{W}_u$ satisfies NFLVR if and only if there exists and equivalent measure $\mathbb{Q}$ such that 
\begin{align*}
    \frac{d\mathbb{Q}}{d\mathbb{P}} \vert_{\mathcal{G}_t} = \text{exp} \left( - \int_0^t\theta_s d\widetilde{W}_s - \frac{1}{2}\int_0^t\theta^2_sds \right), \quad t \in [0,1],
\end{align*}
where $\theta_t = \frac{\tilde{b}_t}{\sigma_t} = \frac{b_t}{\sigma_t} + \alpha_t$. Thus, $\int_0^1\theta_t^2dt = \infty$ on a set of positive $\mathbb{P}$-measure implies a free lunch for the insider. It can be shown (Prop. 1.1 in Imkeller 2002) that $\int_0^t\alpha_t^2dt = \infty$ on a set of positive $\mathbb{P}$-measure implies $\int_0^t\theta_t^2dt = \infty$ on a set of positive $\mathbb{P}$-measure, where $t \in (0,1]$. A specific example can be obtained by considering the honest time $L = \text{sup} \{ t \in [0,1] : W_t = a \}$ for some $a \in \mathbb{R}$. In this case (see Prop. 2.1 and 2.2 for explicit calculations) for $t \in [0,1]$
\begin{align*}
    Z^L_t \coloneqq \mathbb{P}( L > t \vert \mathcal{F}_t) = 1 - F_{1-t}(\lvert W_t - a \rvert)
\end{align*}
and
\begin{align*}
    \alpha_t = -\mathbbm{1}_{[0,L]}(t)\frac{p_{1-t}(\lvert W_t - a \rvert)}{1 - F_{1-t}(\lvert W_t - a \rvert)}\text{sgn}(W_t - a) - \mathbbm{1}_{(L,1]}(t)\frac{p_{1-t}(\lvert W_t - a \rvert)}{F_{1-t}(\lvert W_t - a \rvert)}\text{sgn}(W_t - a),
\end{align*}
where $p_{1-t}$ and $F_{1-t}$ are the density and the distribution function of the law of $\lvert W_t \rvert$. Then it can be shown (Prop. 2.3) that for any $t \in (0,1]$ we have $\int_0^t\alpha_t^2dt = \infty$ on a set of positive $\mathbb{P}$-measure, hence $R$ does not satisfy NFLVR.
\end{example}

\subsection{The relationship between EoF and NUPBR/NA1} 

We begin as usual by arguing from general principles. It is known that $S$ satisfies NA1 if and only if there exists a positive local martingale $Z$ s.t. $ZS$ is a local martingale. Therefore the following result follows (Neufcourt, 2017, p.38, Corollary 2.5, Theorem 2.6). \par

\begin{thm} Suupose that $S$ satisfies NA1 in $\mathbb{F}$. Then $S$ satisfies NA1 in $\mathbb{G}$ if and only if there exists a positive $\mathbb{G}$-local martingale $Z$ s.t. $ZM$ is a local martingale, i.e. $M$ satisfies NA1 in $\mathbb{G},$ \textcolor{black}{where $M$ is the local martingale in the semimartingale decomposition of $S$}. Equivalently, an $\mathbb{F}$-local martingale $M$ satisfies NA1 in $\mathbb{G}$ if and only if there exists a predictable process $\alpha$ s.t. $M - \int_0^{\cdot}\alpha_sd[M,M]_s$ is a $\mathbb{G}$-local martingale.
\end{thm}

\par
Now consider again an initial enlargement setup. Recall that under $\mathcal{J}$ (hence under $\mathcal{E}$) there exists a density process $(p_t(u), t \geq 0)$ that is an $\mathbb{F}$-martingale and $\forall t \geq 0$ the measure $p_t(u)\eta(du)$ equals $\mathbb{P}(\zeta \in du \vert \mathcal{F}_t)$, \textcolor{black}{where $\eta$ is the dominating measure in Jacod's condition.} The proof of the following theorem can be found in (Aksamit, Jeanblanc, 2017, p.97). \par

\begin{thm} Assume that $\zeta$ satisfies Jacod's equivalence condition $\mathcal{E}$ under $\mathbb{P}$. If $S$ admits a local martingale deflator $Y$ (i.e. NUPBR holds in $\mathbb{F}$), then $S$ admits $Y/p(\zeta)$ as a local martingale deflator in $\mathbb{F}^{\sigma(\zeta)}.$ 
\end{thm}

\par
The next result concerns the preservation of NA1 under initial enlargement. Since NA1 is equivalent to NUPBR, the previous theorem applies, but we present the result regardless to establish the link explicitly (Acciaio, Fontana, Kardaras, 2015). Define for every $x \in E$ and $n \in \mathbb{N}$ a collection of stopping times on $(\Omega, \mathbb{F})$ as $\zeta^x_n \coloneqq \text{inf} \{ t \in \mathbb{R}_+ \vert p_t(x) < 1/n \}$, $\zeta^x \coloneqq \text{inf} \{t \in \mathbb{R}_+ \vert p_t(x) = 0\}$. For every $x \in E$ we define $\Lambda^x \coloneqq \{\zeta^x < \infty, p_{\zeta^x-}(x) > 0\}$. Finally, define a stopping time on $(\Omega, \mathbb{F})$ as $\eta^x \coloneqq \zeta^x_{\Lambda^x} = \zeta^x\mathbbm{1}_{\Lambda^x} + \infty\mathbbm{1}_{\Omega\backslash\Lambda^x}$  for all $x \in E$, which is the time at which the density process $p(x)$ jumps to 0. \par

\begin{thm} Assume that $\mathcal{J}$ holds, $\mathbb{L}^1(\Omega, \mathcal{F}, \mathbb{P})$ is separable, and that $\mathbb{P}(\eta^x < \infty, \Delta S_{\eta^x} \neq 0) = 0$ $\gamma$-a.e. $x \in E$. If $S$ satisfies NA1 in $\mathbb{F}$, then $S$ satisfies NA1 in $\mathbb{G}$. 
\end{thm}

\par
Two remarks are in order. Firstly, the condition $\mathbb{P}(\eta^x < \infty, \Delta S_{\eta^x} \neq 0) = 0$ cannot be dropped as various counterexamples show. Secondly, the statement of the theorem concerns a \textit{specific} semimartingale $S$ and does not address the preservation of NA1 for \textit{any} semimartingale. This is resolved in the following result, for which it is necessary to introduce preliminary notation. We follow (Acciaio, Fontana, Kardaras, 2015). Denote by $D^x$ the predictable compensator of $\mathbbm{1}_{\llbracket \eta^x, \infty\llbracket}$ on $(\Omega, \mathcal{F}, \mathbb{F}, \mathbb{P})$ for all $x \in E$. Define $(S^x)_{x \in E}$ to be $S^x \coloneqq \mathcal{E}(-D^x)^{-1}\mathbbm{1}_{\llbracket 0, \eta^x \llbracket}$, for all $x \in E$. The process $S^J$ is defined as $S^J(\omega, t) \coloneqq S_t^{J(\omega)}(\omega)$ for all $(\omega, t) \in \Omega \times \mathbb{R}_+$ and is a semimartingale on $(\Omega, \mathcal{F}, \mathbb{G}, \mathbb{P})$. $S^J$ has a particularly intuitive financial interpretation of an insider having knowledge of $J$ and at time $0$ taking a position on a single unit of the stock with index $J$ that she keeps indefinitely. \par

\begin{thm} Under $\mathcal{J}$ the following statements hold:\\
(i) If $\mathbb{P}(\eta^x < \infty) = 0$ $\gamma$-a.e. $x \in E$, then any $S$ that satisfies NA1 in $\mathbb{F}$ also satisfies NA1 in $\mathbb{G}$. \\
(ii) If $\mathbb{L}^1(\Omega, \mathcal{F}, \mathbb{P})$ is separable and $\int_E \mathbb{P}(\eta^x < \infty)\gamma(dx) > 0$, then the family $(S^x)_{x \in E}$ consists of local martingales on $(\Omega, \mathcal{F}, \mathbb{P})$, and $S^J$ is nondecreasing with $\mathbb{P}(S^J_t = S^J_0, \forall t \in \mathbb{R}_+) < 1.$ In particular, $S^x$ satisfies NA1 in $\mathbb{F}$ for any $x \in E$, but $S^J$ does not satisfy NA1 in $\mathbb{G}$.
\end{thm}

\par
We now proceed to concrete examples showing how NA1 (or, equivalently, NUPBR) may hold in certain cases, but in general fails in the initially enlarged filtration. \par

\begin{example} (Chau et al, 2016) Consier two independent Poisson processes $N^1$ and $N^2$ with intensity $\lambda = 1$. The discounted risky asset price process is taken to be $S_t = e^{N^1_t - N^2_t}$, i.e.
\begin{align*}
    dS_t = S_{t-}((e - 1)dN^1_t + (e^{-1} - 1)dN^2_t), \quad S_0 = 1, t \in [0,T].
\end{align*}
$\mathbb{F}$ is the natural filtration generated by the two Poisson processes, and $S$ satisfies NFLVR in $\mathbb{F}$ with the density $Z$ of any equivalent local martingale measure given by
\begin{align*}
    dZ_t = Z_{t-}((\alpha^1_t - 1)(dN^1_t-dt) + (\alpha^2_t - 1)(dN^2_t - dt)),
\end{align*} 
where $\alpha^1$ and $\alpha^2$ are positive integrable processes such that $\alpha^1_t = e^{-1}\alpha^2_t.$ \par 
We now consider $N_t \coloneqq N^1_t - N^2_t$ with the original filtration enlarged with the terminal value $N_T$. This is tantamount to knowing (i.e. having inside information on) $S_T$. Then it can be shown that the conditional density for all $t \in [0, T)$ takes the form 
\begin{align*}
    p_t(x) = \frac{\mathbb{P}(N_T = x \vert \mathcal{F}_t)}{\mathbb{P}(N_T = x )} = \frac{\sum_{k\geq0}e^{-(T-t)}\frac{(T-t)^k}{k!}e^{-(T-t)}\frac{(T-t)^{k + x - N_t}}{(k + x - N_t)!}\mathbbm{1}_{k + x - N_t \geq 0}}{\sum_{k\geq0}e^{-T}\frac{T^k}{k!}e^{-T}\frac{T^{x+k}}{(x+k)!}} > 0, 
\end{align*}
and for $t = T$
\begin{align*}
     p_T(x) = \frac{\mathbbm{1}_{N_T=x}}{\sum_{k\geq0}e^{-T}\frac{T^k}{k!}e^{-T}\frac{T^{x+k}}{(x+k)!}},
\end{align*}
from which we see that the density process does not jump to 0, i.e. $p_t(x)$ is strictly positive before time $T$. This implies that $S$ satisfies NUPBR (NA1) in $\mathbb{G}.$ 
\end{example}

\par
\begin{example} (Acciaio et al, 2015) Consider a Poisson process with intensity $\lambda > 0$ on $(\Omega, \mathcal{F}, \mathbb{F}, \mathbb{P})$ on the time interval $[0, T]$ for some $T \in (0, \infty).$ We enlarge the filtration by the terminal value $J \coloneqq N_T$. The hypothesis $\mathcal{J}$ holds, and the density process is known to be 
\begin{align*}
    p_t(x) = e^{\lambda t} \frac{(\lambda(T - t))^{x - N_t}}{(\lambda T)^x}\frac{x!}{(x - N_t)!}\mathbbm{1}_{\{N_t \leq x\}}, \quad \text{for all} \quad t \in [0, T), \quad p_T(x)=e^{\lambda T}\frac{x!}{(\lambda T)^x}\mathbbm{1}_{\{N_T = x\}}.
\end{align*}
Now consider the process $S_t \coloneqq$ exp($N_t - \lambda t( e - 1))$ for all $t \in [0, T]$, which can be shown to be a strictly positive $\mathbb{F}$-martingale, which implies that NA1 holds in $\mathbb{F}$. To show that it fails in $\mathbb{G}$ consider the strategy ($\mathbbm{-1}_{\rrbracket \sigma, T \rrbracket}$), where $\sigma \coloneqq \text{inf} \{t \in [0,T] \vert N_t = N_T\}$. The associated value process is thus 
\begin{align*}
    (\mathbbm{-1}_{\rrbracket \sigma, T \rrbracket} \cdot S)_t = \mathbbm{1}_{\{t > \sigma \}} \text{exp}(N_{\sigma} - \lambda \sigma(e - 1))(1 - \text{exp}(-\lambda(t-\sigma)(e-1))).
\end{align*}
Thus it follows that $\mathbbm{-1}_{\rrbracket \sigma, T \rrbracket} \bullet S$ is nondecreasing and $\mathbb{P}(\sigma < T) = 1$, hence NA1 does not hold in $\mathbb{G}$. A careful examination of the example shows that the condition  $\mathbb{P}(\eta^x < \infty, \Delta S_{\eta^x} \neq 0) = 0$ for $\gamma$-a.e. $x \in E$ is violated since the processes $p(x)$ have a positive probability to jump to 0 (at the jump times of the Poisson process $N$). 
\end{example}

\par
We now consider a progressive enlargement setup which revolves primarily around the behavior of the Azema supermartingale. Specifically, we consider a random time defined as $\tau : \Omega \mapsto \overline{\mathbb{R}}_+$ such that $\mathbb{P}(\tau = \infty) = 0$. The filtration is progressively enlarged as
\begin{align*}
    \mathcal{G}_t = \{B \in \mathcal{F} \ \vert \ B \cap \{\tau > t\} = B_t \cap \{\tau > t\} \text{ for some } B_t \in \mathcal{F}_t\} \quad \text{for all} \quad t \in \mathbb{R}_+.
\end{align*}
Thus $\mathbb{G}$ turns $\tau$ into a stopping time. 
The finiteness of $\tau$ implies that $Z_{\infty} \coloneqq \lim_{t \to \infty}Z_t = 0$ a.s. It follows that $Z$ can be represented as $Z = m - A$, where $A$ is the dual optional projection of $\mathbbm{1}_{\llbracket \tau, \infty \llbracket}$ and $m$ is some uniformly integrable martingale on $(\Omega, \mathbb{F}, \mathbb{P})$ and $m_t = \mathbb{E}(A_{\infty} \vert \mathcal{F}_t)$. It can also be shown that for any $\mathbb{F}$-stopping time $\sigma$ we have $\Delta A_{\sigma} = \mathbb{P}(\tau = \sigma \vert \mathcal{F}_{\sigma})$ on $\{\sigma < \infty \}$. Similarly to the initial enlargement case, we define a number of stopping times. Define for every $n \in \mathbb{N}$ a collection of stopping times on $(\Omega, \mathbb{F})$ as $\zeta_n \coloneqq \text{inf} \{ t \in \mathbb{R}_+ \vert Z_t < 1/n \}$, $\zeta \coloneqq \lim_{n \to \infty} \zeta_n = \text{inf} \{t \in \mathbb{R}_+ \vert Z_{t-} = 0 \: \text{or} \: Z_t = 0\} = \text{inf} \{t \in \mathbb{R}_+ \vert  Z_t = 0\} $. Now define the $\mathcal{F}_{\zeta}$-measureable event $\Lambda \coloneqq \{\zeta < \infty, Z_{\zeta-} > 0, \Delta A_{\zeta} = 0\}$ and a stopping time on $(\Omega, \mathbb{F})$ as $\eta \coloneqq \zeta_{\Lambda} = \zeta\mathbbm{1}_{\Lambda} + \infty\mathbbm{1}_{\Omega\backslash\Lambda}$. Note that $Z_{\eta} = \mathbb{P}(\tau > \eta \vert \mathcal{F}_{\eta}) = 0$ and $\mathbb{P}(\tau = \eta < \infty \vert \mathcal{F}_{\eta})= \Delta A_{\eta}\mathbbm{1}_{\{\eta < \infty\}} = \Delta A_{\eta}\mathbbm{1}_{\Lambda} = 0$, where we have used the fact that $\mathbb{P}(\tau = \infty) = 0$. The following is an important result for a particular stopped semimartingale $S^{\tau} \coloneqq (S_{\tau \land t})_{t \in \mathbb{R}_+}$. \par

\begin{thm} Suppose $S$ satisfies NA1 in $\mathbb{F}$ and $\mathbb{P}(\eta < \infty, \Delta S_{\eta} \neq 0) = 0$. Then $S^{\tau}$ satisfies NA1 in $\mathbb{G}$.
\end{thm}

\par
The following is the general result for all semimartingales. \par

\begin{thm} The following statements are true: \\
    (i) Suppose $\mathbb{P}(\eta < \infty) = 0$. Then for any $S$ that satisfies NA1 in $\mathbb{F}$, it follows that $S^{\tau}$ satisfies NA1 in $\mathbb{G}$. \\
    (ii) Suppose $\mathbb{P}(\eta < \infty) > 0$. Then $S \coloneqq \mathcal{E}(-D)^{-1}\mathbbm{1}_{\llbracket0, \eta \llbracket}$ is an $\mathbb{F}$-local martingale, and $S^{\tau}$ is nondecreasing with $\mathbb{P}(S_{\tau} > 1) > 0$, where $D$ is the predictable compensator of $\mathbbm{1}_{\llbracket \eta, \infty \llbracket}$ on $(\Omega, \mathbb{F}, \mathbb{P}).$ 
\end{thm}

\par
We now proceed to concrete examples showing how NA1 (or, equivalently, NUPBR) may hold or fail in the progressively enlarged filtration. \par

\begin{example} (Acciaio et al, 2015, p.1769) Consider a random variable $\zeta\sim $Exp(1) on $(\Omega, \mathcal{F}, \mathbb{P})$. Define $\mathbb{F} = (\mathcal{F}_t)_{t \in \mathbb{R}_+}$ to be the smallest filtration that satisfies the usual hypotheses and turns $\zeta$ into a stopping time. Consider the progressive enlargement $\mathbb{G} \coloneqq \mathbb{F} \lor \sigma(\tau)$, where $\tau \coloneqq \zeta / 2$. \textcolor{black}{As usual, let $Z$ be the Azema supermartingale associated with $\t$ and $A$ the dual optional projection of $\1_{\llbracket\t,\infty\llbracket}$}. Note $Z_t = \text{exp}(-t)\mathbbm{1}_{\{t < \zeta\}}$ for all $t \in \mathbb{R}_+$. Moreover, $\Delta A_{\sigma} = \mathbb{P}(\tau = \sigma \vert \mathcal{F}_{\sigma}) = \mathbb{P}(\zeta = 2\sigma \vert \mathcal{F}_{\sigma}) = 0$ for all bounded $\mathbb{F}$-stopping times, hence $\Delta A =0$. Observe that $\zeta  = \text{inf} \{t \in \mathbb{R}_+ \vert Z_{t-} = 0 \: \text{or} \: Z_t = 0\}$ and $Z_{\zeta -} = \text{exp}(-\zeta) > 0$, and, by the construction of $\eta$, we have $\eta = \zeta$. Note that the $\mathbb{F}$-predictable compensator of $\mathbbm{1}_{\llbracket \eta, \infty \llbracket}$ is $D \coloneqq (\eta \land t)_{t \in \mathbb{R}_+}$, and $\eta = \zeta$ is totally inaccessible on $(\Omega, \mathbb{F}, \mathbb{P}).$ Since $\mathbb{P}(\eta < \infty) = 1$, we can apply part (ii) of the above theorem and procure $S \coloneqq \mathcal{E}(-D)^{-1}\mathbbm{1}_{\llbracket0, \eta \llbracket} = \text{exp}(D)\mathbbm{1}_{\llbracket0, \eta \llbracket} = (\text{exp}(t)\mathbbm{1}_{\{t<\zeta\}})_{t \in \mathbb{R}_+}$. Since $S$ is a quiasi-left-continuous nonnegative martingale, it satisfies NA1 in $\mathbb{F}$. However, since $S $ is strictly increasing up to $\tau$, it follows that $S^{\tau}$ does not satisfy NA1 in $\mathbb{G}$. 
\end{example}

\par
\begin{example} (Acciaio et al, 2015, p. 1769) \textcolor{black}{In the notation of the previous example}, consider a random variable $\zeta : \Omega \to \mathbb{N}$, where $p_k \coloneqq \mathbb{P}(\zeta = k) \in (0,1)$ for all $k \in \mathbb{N}$, and $\sum_{k = 1}^{\infty}p_k = 1$. Take $\mathbb{F} = (\mathcal{F}_t)_{t \in \mathbb{R}_+}$ to be the smallest filtration that satisfies the usual hypotheses and turns $\zeta$ into a stopping time. Note that $\zeta$ is accessible since it is $\mathbb{N}$-valued. Consider the progressive enlargement $\mathbb{G} \coloneqq \mathbb{F} \lor \sigma(\tau)$, where $\tau \coloneqq \zeta -1$. Note that $Z_t = 0$ on $\{\zeta \leq t\}$, and 
\begin{align*}
    Z_t = \mathbb{P}(\tau > t \vert \mathcal{F}_t) = \mathbb{P}( \zeta > t + 1 \vert \mathcal {F}_t) = \mathbb{P}( \zeta > \lceil t + 1 \rceil \vert \mathcal {F}_t) = \frac{q_{\lceil t+1 \rceil}}{q_{\lceil t \rceil}} \ \text{on} \ \{t < \zeta\},
\end{align*}
where $q_k \coloneqq \sum_{n = k +1}^{\infty}p_n$ for all $ k \in \{0, 1, 2, ...\}.$ Observe that $\zeta  = \text{inf} \{t \in \mathbb{R}_+ \vert Z_{t-} = 0 \: \text{or} \: Z_t = 0\}$ and $Z_{\zeta -} = \frac{q_{\lceil \zeta \rceil}}{q_{\lceil \zeta - 1 \rceil}} > 0$, $\Delta A_{\zeta} = \mathbb{P}(\tau = \zeta \vert \mathcal{F}_{\zeta}) = 0$. Thus, $\eta = \zeta$, and $\eta$ is accessible on $(\Omega, \mathbb{F}, \mathbb{P})$.
\end{example}

\bigskip

\section{Enlargement of Filtrations and a Change of Measure} 

\subsection{Enlargement of Filtrations under a Change of Measure}

The discussion above leads to the question whether $\mathcal{H}$ or $\mathcal{H}^{\prime}$ continue to hold under a change of the underlying probability measure. We begin from the most intuitively clear-cut case of immersion. The following general result was proved in (Jeulin, Yor, 1978). See also (Coculescu et al, 2012, p. 9) or (Aksamit, Jeanblanc, 2017, p. 55) for a proof. \par

\begin{thm} Assume $\mathbb{F} \hookrightarrow \mathbb{G}$ under $\mathbb{P}$. Let $\mathbb{Q}$ be equivalent to $\mathbb{P}$ on $\mathcal{G}_{\infty}$. For $t \in [0,\infty)$ define $\frac{d\mathbb{Q}}{d\mathbb{P}} \vert_{\mathcal{G}_t} = L_t$ and $\frac{d\mathbb{Q}}{d\mathbb{P}} \vert_{\mathcal{F}_t} = l_t = \mathbb{E}(L_t \vert \mathcal{F}_t).$ \\
(i) if $X$ is an $\mathbb{F}$-local martingale under $\mathbb{Q}$, then the process 
\begin{align*}
    X_t + \int_0^t\frac{L_{s-}}{L_s} \left( \frac{1}{l_{s-}}d[X,l]_s - \frac{1}{L_s}d[X.L]_s\right) = X_t + \int_0^t \frac{1}{Y_{s-}}d[X,Y]_s,
\end{align*}
where $Y = l / L$ is a $\mathbb{G}$-local martingale under $\mathbb{Q}.$ \\
(ii) $\mathbb{F} \hookrightarrow \mathbb{G}$ under $\mathbb{Q}$ if and only if for any $\zeta \in L^1(\mathcal{F}_{\infty})$ we have
\begin{center}
       $ \frac{\mathbb{E}_{\mathbb{P}}(\zeta L_{\infty} \vert \mathcal{G}_t)}{L_t} = \frac{\mathbb{E}_{\mathbb{P}}(\zeta l_{\infty} \vert \mathcal{F}_t)}{l_t}. $
\end{center}
(iii) Assume $L$ is $\mathbb{F}$-adapted, then  $\mathbb{F} \hookrightarrow \mathbb{G}$ under $\mathbb{Q}$. Moreover, if $\mathbb{G}$ is the progressive enlargement with a random time $\tau$, then the associated Azema supermartingales under $\mathbb{P}$ and $\mathbb{Q}$ are equal. 
\end{thm}

\par
We now consider a change of probability measure under  initial enlargement. 

\par
\begin{thm} Assume that $\zeta$ satisfies Jacod's equivalence condition and \textcolor{black}{$p$ is the usual associated density process}. Then the process $L \coloneqq \frac{p_0(\zeta)}{p(\zeta)}$ is a positive $\mathbb{F}^{\sigma(\zeta)}$-martingale under $\mathbb{P}$, with $L_0 = 1$ Then there exists a probability measure $\mathbb{P}^*$ on $\mathcal{F}_{\infty}^{\sigma(\zeta)}$ such that $\mathbb{P}^* \sim \mathbb{P}$ on $\mathcal{F}_t^{\sigma(\zeta)}$ for any $t \geq 0$ and 
\begin{align*}
    d\mathbb{P}^* \vert_{\mathcal{F}_t^{\sigma(\zeta)}} = L_t d\mathbb{P} \vert_{\mathcal{F}_t^{\sigma(\zeta)}}, \quad \mathbb{P}^* \vert_{\mathcal{F}_t} = \mathbb{P} \vert_{\mathcal{F}_t}, \quad \mathbb{P}^* \vert_{\sigma(\zeta)} = \mathbb{P} \vert_{\sigma(\zeta)}.
\end{align*}
Moreover, under $\mathbb{P}^*$, the $\sigma$-algebras $\sigma(\zeta)$ and $\mathcal{F}_{\infty}$ are conditionally independent given $\mathcal{F}_0$. 
\end{thm}

\par
This last statement has an obvious implication that is the point of our interest in the result. \par

\begin{cor} $\mathbb{F} \hookrightarrow \mathbb{\mathbb{F}^{\sigma(\zeta)}}$ under $\mathbb{P}^*$. 
\end{cor}

\par
We now consider a change of probability measure under  progressive enlargement. The proof of the following stability result can be found in, e.g., (Aksamit, Jeanblanc, 2017, p.62). Assume $\mathbb{F} \hookrightarrow \mathbb{G}$ under $\mathbb{P}$ and define an equivalent probability measure $\mathbb{Q}$ by $\frac{d\mathbb{Q}}{d\mathbb{P}} = \mathcal{E}(\varphi \bullet Y)_{\infty}\mathcal{E}(\gamma \bullet M)_{\infty}$, where $Y$ is an $\mathbb{F}$-martingale under $\mathbb{P}$, $\varphi $ is an $\mathbb{F}$-predictable process such that $\mathcal{E}(\varphi \bullet Y)$ is a uniformly integrable $\mathbb{G}$-martingale, and $\gamma $ is an $\mathbb{G}$-predictable process such that $\mathcal{E}(\gamma \bullet M)$ is a uniformly integrable $\mathbb{G}$-martingale. Finally, assume that $L = \frac{d\mathbb{Q}}{d\mathbb{P}} = \mathcal{E}(\varphi \bullet Y)_{\infty}\mathcal{E}(\gamma \bullet M)_{\infty}$ is a positive uniformly integrable $\mathbb{G}$-martingale. \par

\begin{thm} Assume $\mathbb{F} \hookrightarrow \mathbb{G}$ under $\mathbb{P}$. Let the equivalent probability measure $\mathbb{Q}$ be defined by  $\frac{d\mathbb{Q}}{d\mathbb{P}} = \mathcal{E}(\varphi \bullet Y)_{\infty}\mathcal{E}(\gamma \bullet M)_{\infty}$. Then, under assumption (C) or (A), $\mathbb{F} \hookrightarrow \mathbb{G}$ under $\mathbb{Q}$.
\end{thm}

\par
More can be said if we assume that both assumptions (C) and (A) are satisfied. We follow (Kreher, 2016). In the following, $\sigma$ will denote a random time, the enlarged filtration is taken to be defined by $\mathcal{G}_t \coloneqq \cap_{s>t} \left( \mathcal{F}_s \lor \sigma(\mathbbm{1}_{\{\sigma > r\}}; r \leq s) \right)$, $\rho$ is a non-negative $\mathcal{F}_{\infty}$-measurable random variable with $\mathbb{E}(\rho)=1$, $d\mathbb{Q} = \rho d\mathbb{P}$, $\rho_t \coloneqq \mathbb{E}_{\mathbb{P}}(\rho \vert \mathcal{F}_t)$, $\widetilde{\rho}_t \coloneqq \mathbb{E}_{\mathbb{P}}(\rho \vert \mathcal{G}_t)$, $h_t \coloneqq \mathbb{E}_{\mathbb{P}}(\rho \mathbbm{1}_{\{\sigma > t\}} \vert \mathcal{F}_t)$ for $t \geq 0$. Thus $h_t$ is an $\mathbb{F}$-supermartingale under $\mathbb{P}$, and, by the well-known Bayes formula, $h_t = \rho_t \cdot \mathbb{Q}(\sigma > t \vert \mathcal{F}_t) \eqqcolon \rho_t Z_t^{\mathbb{Q}}$, hence $h_t \to 0$ a.s. as $t \to \infty$. For a strictly positive $h$ we also define its stochastic logarithm $\mu$ via $h = \mathcal{E}(\mu)$. $\mu$ is also an $\mathbb{F}$-supermartingale under $\mathbb{P}$ with the Doob-Meyer decomposition $\mu = \mu^L - \mu^F$, where $\mu^L \in \mathcal{M}_{loc}(\mathbb{P}, \mathbb{F})$, and $\mu^F$ is an increasing finite variation process. $h, \mu, \mu^L$ and $\mu^F$ are all continuous. The following is a well-known result obtained in (Mortimer, Williams, 1991). \par

\begin{thm} Assume $h$ is strictly positive. Let $U= (U_t)_{t \geq0}$ be an $\mathbb{F}$-local martingale under $\mathbb{P}$. Then the process $\left(\mathbbm{1}_{\{\sigma > t\}}V_t\text{exp}(\mu^F_t)\right)_{t \geq 0}$ is a $\mathbb{G}$-local martingale under $\mathbb{Q}$, where $V \coloneqq U - \langle U, \mu \rangle$. Moreover, the process $\left( \mu^F_{t \land \sigma} \right)_{t \geq 0}$ is the $(\mathbb{Q}, \mathbb{G})$-dual predictable projection of $\left( \mathbbm{1}_{\{\sigma \leq t\}} \right)_{t \geq 0}$. 
\end{thm}

\par
\begin{cor} Assume $h$ is strictly positive. If $U \in \mathcal{M}_{loc}(\mathbb{P}, \mathbb{F})$, then $V_{t \land \sigma} = U_{t \land \sigma} - \langle U, \mu \rangle_{t \land \sigma} \in \mathcal{M}_{loc}(\mathbb{Q}, \mathbb{G})$. 
\end{cor}

\par
Note that if we set $\rho \equiv 1$, we get the familiar enlargement formula considered above, i.e. for any $M \in \mathcal{M}_{loc}(\mathbb{P}, \mathbb{F})$ we have
\begin{align*}
    M_{t \land \sigma} - \int_0^{t \land \sigma} \frac{d\langle M, Z^{\mathbb{P}} \rangle_s}{Z^{\mathbb{P}}} =  M_{t \land \sigma} - \int_0^{t \land \sigma} \frac{d\langle M, m^{\mathbb{P}} \rangle_s}{Z^{\mathbb{P}}} \in \mathcal{M}_{loc}(\mathbb{P}, \mathbb{G}),
\end{align*}
\textcolor{black}{where, as before, $Z$ is the Azema supermartingale associated wtih $\s$ and $m$ is its local martingale part in the Doob-Meyer decomposition}.
\par
The positivity of $h$ can be relaxed. The proof can be found in (Kreher, 2016, p. 4). \par

\begin{thm} Let $U= (U_t)_{t \geq0}$ be an $\mathbb{F}$-local martingale under $\mathbb{P}$. Then the processes $X = \left(\mathbbm{1}_{\{\sigma > t\}}V_t\text{exp}(\mu^F_t)\right)_{t \geq 0}$ and $\left( V_{t \land \sigma} \right)_{t \geq 0}$ are $\mathbb{G}$-local martingales under $\mathbb{Q}$, where $V_{t \land \sigma} = U_{t \land \sigma} - \langle U, \mu \rangle_{t \land \sigma}$.
\end{thm}

\par
The results above were concerned with a change of measure \textit{up to} a (general) random time $\sigma$. The difficulty in extending them beyond $\sigma$ lies in the fact that for general random times not all $\mathbb{F}$-semimartingales remain $\mathbb{G}$-semimartingales. However, for a very special subset - honest times - this is not a restriction, hence broader results can be obtained for the entire time domain. Assume from now on that $\sigma$ is an honest time. Define an $\mathbb{F}$-submartingale $k$ under $\mathbb{P}$ by $k_t = \mathbb{E}_{\mathbb{P}}(\rho \vert \mathcal{F}_t) - h_t = \mathbb{E}_{\mathbb{P}}(\rho\mathbbm{1}_{\{\sigma \leq t\}}\vert \mathcal{F}_t)$. For a fixed $u \geq 0$ denote by $\mathcal{M}^u(\mathbb{P}, \mathbb{F})$ the class of $\mathbb{F}$-martingales under $\mathbb{P}$ on $[u, \infty)$. Since $\sigma$ is an honest time, for a fixed $t \geq 0$ there exists an $\mathcal{F}_t$-measurable random variable $\sigma_t$ such that $\mathbbm{1}_{\{\sigma < t\}}\sigma = \mathbbm{1}_{\{\sigma < t\}}\sigma_t$. Two auxillary results are necessary to establish the main result. See (Kreher, 2016, pp. 10-11 for proofs). \par

\begin{lem} Fix $u \geq 0$. Let $Y$ be an $\mathbb{F}$-adapted process such that $(\mathbbm{1}_{\{\sigma_t \leq u\}}k_tY_t)_{t \geq u} \in \mathcal{M}^u_{loc}(\mathbb{P},\mathbb{F})$. Then $Y_t\mathbbm{1}_{\{\sigma \leq u\}} \in \mathcal{M}^u_{loc}(\mathbb{Q},\mathbb{G})$.
\end{lem}

\par
\begin{lem} Let $(Y_t)$ be a real-valued continuous $\mathbb{G}$-adapted process such that $(\mathbbm{1}_{\{\sigma \leq u\}}(Y_{t \lor u} - Y_u))_{t \geq 0} \in \mathcal{M}_{loc}(\mathbb{Q},\mathbb{G})$ for all $u > 0$. Then $(Y_{t \lor \sigma} - Y_{\sigma})_{t \geq 0} \in  \mathcal{M}_{loc}(\mathbb{Q},\mathbb{G})$.
\end{lem}

\par
We are now in a position to derive the general result for honest times. Although both (C) and (A) are assumed to hold, the result can be proved without them. \textcolor{black}{In the notation above, we have} \par

\begin{thm} Let $\sigma$ be an honest time and suppose that $(U_t)_{t \geq 0}$ is an $\mathbb{F}$-local martingale under $\mathbb{P}$. Then the process
\begin{align*}
    V_t \coloneqq U_t - \int_0^{\sigma \land t}\frac{d\langle U,h\rangle_s}{h_s} - \int_{\sigma}^{\sigma \lor t}\frac{d\langle U,k\rangle_s}{k_s} 
\end{align*}
is a $\mathbb{G}$-local martingale under $\mathbb{Q}$.
\end{thm}

\par

\par
\bigskip

\subsection{Enlargement of Filtrations as a Change of Measure}

As was noted earlier, there are many similarities between the enlargement of filtrations and Girsanov-type changes of measure where a drift term appears or is removed. This resemblance is especially vivid in the context of explicit enlargement formulas where the drift terms are identified. The main focus of this section is to present a view of enlargement \textit{as} a change of measure. Essentially, the discussion will revolve around the ideas of Ankirchner where enlargement is tied to a change of measure on a product space into which our original probability space in embedded. More specifically, once we have constructed the so-called \textbf{decoupling measure} on the product space (that is assumed to dominate the original product probability measure), it can be shown that the change of filtrations is basically a change from the decoupling measure to the dominated one. Within this setup general Girsanov-type results can be obtained. Of key importance is the fact that the well-known Jacod's condition implies the assumed domination by the decoupling measure. \par
Consider a probability space $(\Omega, \mathcal{F}, \mathbb{P})$ and right-continuous filtrations $(\mathcal{F}_t)_{t \geq 0}$ and $(\mathcal{H}_t)_{t \geq 0}$ with $\mathcal{F}_{\infty} = \bigvee_{t \geq 0} \mathcal{F}_t$, $\mathcal{H}_{\infty} = \bigvee_{t \geq 0} \mathcal{H}_t$. We define the enlarged filtration $\mG$ the usual way as $\mathcal{G}_t = \cap_{s>t} (\mathcal{F}_s \vee \mathcal{H}_s)$ for all $t \geq 0$. We now consider the two filtrations on identical copies of the underlying probability space and introduce the product space $\overline{\Omega} = \Omega \times \Omega$ with the product $\sigma$-algebra $\overline{\mathcal{F}} = \mathcal{F}_{\infty} \otimes \mathcal{H}_{\infty}$ and filtration $\overline{\mathcal{F}}_t = \cap_{s>t}(\mathcal{F}_s \otimes \mathcal{H}_s)$. The map $\psi : (\Omega, \mathcal{F}) \to (\overline{\Omega}, \overline{\mathcal{F}})$ (i.e. $\omega \mapsto (\omega, \omega)$) defines an \textbf{embedding} of $\Omega$ into $\overline{\Omega}$. The image measure of $\mathbb{P}$ under $\psi$ is defined as $\overline{\mathbb{P}} \coloneqq \mathbb{P}_{\psi} \coloneqq \mathbb{P} \circ \psi^{-1}$, and the equality 
\begin{align*}
    \int f(\omega, \omega')d\overline{\mathbb{P}}(\omega, \omega') = \int f(\omega, \omega) d\mathbb{P}(\omega)
\end{align*}
holds for all $\overline{\mathcal{F}}$-measurable functions $f : \overline{\Omega} \to \mathbb{R}$. A number of technical lemmas are necessary to take care of the measurability issues. We only list the lemmas to facilitate the understanding of the following more fundamental results, the proofs can be found in (Ankirchner, 2005, pp. 9-10). Note that $\mathcal{F}^{\mathbb{P}}$ represents the completion of filtration $\mathcal{F}$ by $\mathbb{P}$-null sets.  \par

\begin{lem} If $\overline{f} : \overline{\Omega} \to \mathbb{R}$ is $\overline{\mathcal{F}}^{\mathbb{P}}_t$-measurable, then the map $\overline{f} \circ \psi$ is $\mathcal{G}^{\mathbb{P}}_t$-measurable. 
\end{lem}

\par
\begin{lem} If $\overline{X}$ is $(\overline{\mathcal{F}}^{\mathbb{P}}_t)$-predictable, then $\overline{X} \circ \psi$ is $(\mathcal{G}^{\mathbb{P}}_t)$-predictable. 
\end{lem}

\par
\begin{lem} Let $\overline{Y}$ be $(\overline{\mathcal{F}}^{\mathbb{P}}_t)$-adapted. Then the process $Y = \overline{Y} \circ \psi$ is $(\mathcal{G}^{\mathbb{P}}_t)$-adapted. Moreover, if $\overline{Y}$ is a $(\overline{\mathcal{F}}^{\mathbb{P}}_t, \overline{\mathbb{P}})$-local martingale, then $Y$ is a $(\mathcal{G}^{\mathbb{P}}_t, \mathbb{P})$-local martingale. 
\end{lem}

\par
Now we can state a key theorem on the transfer of the semimartingale property. \par

\begin{thm} Let $\overline{Y}$ be a $(\overline{\mathcal{F}}^{\mathbb{P}}_t, \overline{\mathbb{P}})$-semimartingale. Then  $Y = \overline{Y} \circ \psi$ is a $(\mathcal{G}^{\mathbb{P}}_t, \mathbb{P})$-semimartingale. 
\end{thm}
\begin{proof}
Let $\overline{Y}$ be a $(\overline{\mathcal{F}}^{\mathbb{P}}_t, \overline{\mathbb{P}})$-semimartingale. Then $Y = \overline{Y} \circ \psi$ has cadlag paths $\mathbb{P}$-a.s., and $Y$ is $(\mathcal{G}^{\mathbb{P}}_t)$-adapted by the lemma above. By the Bichteler-Dellacherie-Mokobodzki theorem, it suffices to show that if a sequence of simple $\mathcal{G}_t$-adapted integrands $(\theta^n)$ (i.e. of the form $\sum_{1 \leq i \leq n} \mathbbm{1}_{]t_i,t_{i+1}]}\theta_i$, where $\theta_i$ is $\mathcal{G}_{t_i}$-measurable) coverges uniformly to $0$, then the respective simple integrals $(\theta^n \cdot Y)$ converge to 0 in probability with respect to $\mathbb{P}$. Note $\mathcal{G}_t = \psi^{-1}(\overline{\mathcal{F}}_t)$, hence there exists a sequence of simple $(\overline{\mathcal{F}}_t)$-adapted integrands $(\overline{\theta}^n)$ that converges uniformly to 0 such that  $\overline{\theta}^n \circ \psi = \theta^n$. Since $\overline{Y}$ is a semimartingale, then $(\overline{\theta}^n \cdot \overline{Y})$ converges to 0 in probability with respect to $\overline{\mathbb{P}}$, therefore  $(\theta^n \cdot Y)$ converges to 0 in probability with respect to $\mathbb{P}$. 
\end{proof}

\par
The discussion above establishes the preservation of the key properties under projection. We now consider the preservation under an embedding. Let $R$ be a probability measure on $\mathcal{H}_{\infty}$. Define the so-called \textbf{decoupling measure} to be 
\begin{align}
    \overline{\mathbb{Q}} \coloneqq \mathbb{P} \vert_{\mathcal{F}_{\infty}} \otimes R \vert_{\mathcal{H}_{\infty}}.
\end{align}

\begin{lem} Let $M$ be a right-continuous $(\mathcal{F}^{\mathbb{P}}_t, \mathbb{P})$-local martingale. Then the process $\overline{M}(\omega, \omega') = M(\omega)$ is a $(\mathcal{F}^{\overline{Q}}_t, \overline{Q})$-local martingale.
\end{lem}

\par
At this point a key assumption is introduced that is very much in the vein of Jacod's condition. \\
\textbf{Assumption.} $\overline{\mathbb{P}} \ll \overline{\mathbb{Q}}$ on $\overline{\mathcal{F}}$. \par
This assumption is satisfied when $R \sim \mathbb{P}$ and $(\mathcal{G}_t)$ is a filtration initially enlarged with a discrete random variable $G$, i.e. when  $\mathcal{H}_t = \sigma(G)$ for all $t \geq 0$. \par
Note that under this assumption the lemma above implies that a  $(\mathcal{F}^{\mathbb{P}}_t, \mathbb{P})$-local martingale $M$ extended to the product space $\overline{\Omega}$ as $\overline{M}$ remains a $(\overline{\mathcal{F}}_t^{\overline{\mathbb{P}}},\overline{\mathbb{P}})$-semimartingale. This in turn implies, by the main theorem above, that $M$ is a $(\mathcal{G}_t^{\mathbb{P}}, \mathbb{P})$-semimartingale, and $\mathcal{H}^{\prime}$ holds. The assumption also allows us to interpret the enlargement of filtrations as essentially equivalent to a change of measure from $\overline{\mathbb{Q}}$ to $\overline{\mathbb{P}}$ on $\overline{\Omega}$, hence Girsanov-type results can be obtained. We refer to (Ankirchner, 2005) for full proofs of the following statements. \par
We consider the cadlag $(\overline{\mathcal{F}}^{\overline{\mathbb{Q}}}_t)$-adapted density process $\overline{Z}_t \coloneqq \frac{d\overline{\mathbb{P}}}{d\overline{\mathbb{Q}}} \big \vert_{\overline{\mathcal{F}}^{\overline{\mathbb{Q}}}_t}$. The discussion above implies that $Z = \overline{Z} \circ \psi$ is a $(\mathcal{G}_t^{\mathbb{P}}, \mathbb{P})$-semimartingale. \par

\begin{lem} Let $\overline{X}$ and $\overline{Y}$ be $(\overline{\mathcal{F}}_t^{\overline{\mathbb{P}}},\overline{\mathbb{P}})$-semimartingales. If $X = \overline{X} \circ \psi$ and $Y = \overline{Y} \circ \psi$, then $[\overline{X}, \overline{Y}] \circ \psi = [X,Y]$ up to indistinguishability with respect to $\mathbb{P}$. 
\end{lem}

\par
The following theorems are the main Girsanov-type results. We use the standard notation where $\overline{U}_t = \Delta \overline{M}_{\overline{T}}\mathbbm{1}_{\{t \leq \overline{T}\}}$, $\overline{T} \coloneqq \text{inf}\{t > 0 : \overline{Z}_t = 0, \overline{Z}_{t-}>0\}$, $\widetilde{U}$ is the compensator ($(\overline{\mathcal{F}}^{\overline{\mathbb{Q}}}_t, \overline{\mathbb{Q}})$-predictable projection)  of $\overline{U}$, and $\hat{U} = \widetilde{U} \circ \psi$. 
\par

\begin{thm} If $M$ is a $(\mathcal{F}^{\mathbb{P}}_t, \mathbb{P})$-local martingale with $M_0 = 0$, then
\begin{align*}
    M - \frac{1}{Z} \cdot [M,Z] + \hat{U}
\end{align*}
is a $(\mathcal{G}_t^{\mathbb{P}}, \mathbb{P})$-local martingale.
\end{thm}

\par
\begin{thm} If $M$ is a continuous $(\mathcal{F}^{\mathbb{P}}_t, \mathbb{P})$-local martingale with $M_0 = 0$, then
\begin{align*}
    M - \frac{1}{Z} \cdot [M,Z]
\end{align*}
is a $(\mathcal{G}_t^{\mathbb{P}}, \mathbb{P})$-local martingale.
\end{thm}

\bigskip

\section{Martingale Representation Property under Enlargement}

\subsection{Preliminaries and Definitions}

We begin by setting up some notation and definitions (consult Protter, 2005, Medvegyev, 2007, or Cohen, Elliott, 2015, for more details). For $p \in [1, \infty)$, denote by $\mathcal{H}^p$ the space of martingales $M$ such that $\lVert M \rVert_{\mathcal{H}^p} \coloneqq \lVert M^*_{\infty} \rVert_p = \left(\mathbb{E}[\text{sup}_t\lvert M_t\rvert^p]\right)^{1/p} < \infty$. We denote by $\mathcal{H}^p_0$ the processes $M$ in $\mathcal{H}^p$ with $M_0 = 0$, and by $\mathcal{H}^p_{\text{loc}}$ the processes which are locally in $\mathcal{H}^p$. Denote the stopped process by $M^T = (M^T_t) = (M_{t \land T})$. \par

\begin{defn} A subspace $\mathcal{K} \subset \mathcal{H}^p$ is said to be $\textbf{stable}$ if \\
(i) it is closed under the $\mathcal{H}^p$-norm topology, \\
(ii) it is closed under stopping, i.e. for any stopping time $T$ and $M \in \mathcal{K}$ it follows that $M^T \in \mathcal{K}$, \\
(iii) if $M \in \mathcal{K}$ and $A \in \mathcal{F}_0$, then $\mathbbm{1}_AM \in \mathcal{K}$. 
\end{defn}

\par
If $\mathcal{X}$ is a subset of $\mathcal{H}^p_0$, then we denote by stable$(\mathcal{X})$ the smallest closed subspace of $\mathcal{H}^p_0$ which is closed under stopping and contains $\mathcal{X}$. \par

\begin{defn} Let $p \in [1, \infty)$. The subset $\mathcal{X} \subseteq \mathcal{H}^p_0$ is said to have the \textbf{Martingale Representation Property} (MRP) if $\mathcal{H}^p_0 = \text{stable}(\mathcal{X})$. 
\end{defn}

\par
Consider the set of processes $\mathcal{C} \coloneqq \{ X \in \mathcal{H}^p_0 : X = H \bullet M, \ M \in \mathcal{H}^p_0 \}$, \textcolor{black}{where $H \bullet M\coloneqq \left(\int_0^tH_sdM_s\right)_{t\in\mT},$}  which can be shown to be closed in $\mathcal{H}^p_0$. The following theorem is standard (e.g. Medvegyev, 2007, p. 343). \par

\begin{thm} Let $(M_i)_{i=1}^n$ be a finite subset of $\mathcal{H}^p_0$. Assume that for $i \neq j$ the martingales $M_i$ and $M_j$ are strongly orthogonal in the sense of local martingales (i.e. $[M_i,M_j] = 0$ for $i \neq j$). Then
\begin{align*}
    \text{stable}(M_1, ..., M_n) = \left\{ \sum_{i=1}^{n} H_i \bullet M_i : H_i \in \mathcal{L}^p(M_i) \right\}.
\end{align*}
In other words, the stable subspace generated by a finite set of strongly orthogonal $\mathcal{H}^p_0$-martingales is the linear subspace generated by the stochastic integrals $H_i \bullet M_i$, for predictable $H_i \in \mathcal{L}^p(M_i)$, i.e. $\lVert H_i \rVert \coloneqq \left( \mathbb{E}\left( \int_0^{\infty}H_i^p(s)d[M_i,M_i]_s \right) \right)^{1/p} < \infty$. 
\end{thm}

\par
\begin{defn} Let $(M_i)_{i=1}^n$ be a finite subset of $\mathcal{H}^p_0$. We say that $(M_i)_{i=1}^n$ has the \textbf{Predictable Representation Property} (PRP) if for every $X \in \mathcal{H}^p_0$ we have $X = \sum_{i=1}^{n} H_i \bullet M_i$, for some predictable $H_i \in \mathcal{L}^p(M_i)$. 
\end{defn}

\par
The martingale (predictable) representation property plays a central role in mathematical finance. Recall that a financial market is called \textbf{complete} if any payoff is attainable. More formally (Jarrow, 2021): \par

\begin{defn} Assume that the set of local martingale measures is not empty, i.e. $\mathfrak{M}_l \neq \emptyset$ ($\iff$ NFLVR). Choose $\mathbb{Q} \in \mathfrak{M}_l $. The market is said to be \textbf{complete} with respect to $\mathbb{Q}$ if given any $Z_T \in L^1_{+}(\mathbb{Q})$, there exists an $x \geq 0$ and $(\alpha_0, \alpha) \in \mathcal{A}(x)$ such that $x + (\alpha \bullet S)_T = Z_T$ and the value process defined by $X_t \coloneqq x + (\alpha \bullet S)_t$ for all $t \in [0,T]$ is a $\mathbb{Q}$-martingale, i.e. $X_t = \mathbb{E}_{\mathbb{Q}} [ Z_T \vert \mathcal{F}_t]$. 
\end{defn}

\par
Within this setup, the \textbf{Second Fundamental Theorem of Asset Pricing} asserts that completeness implies uniqueness of the local martingale measure, and, (partially) in the other direction, the existence of an equivalent martingale measure implies completeness. It is now easy to see that, essentially, completeness (i.e. the existence of a trading strategy that generates a given payoff) is the financial equivalent of the martingale (predictable) representation property (i.e. the existence of a predictable process that can be used to represent our martingale as a stochastic integral). \par
We would like to stress that completeness is a slightly stronger property than uniqueness of the local martingale measure, i.e. it is the set of \textit{martingale} measures that has to be a singleton for completeness to follow. This essentially stems from the theorem of Jacod-Yor whereby the collection $\mathcal{X} \subset \mathcal{H}^1_0$ possesses the  MRP if and only if the underlying probability measure $\mathbb{P}$ is an extremal point of $\mathfrak{M}^1(\mathcal{X}).$ This fact will play an important role in the settings of several theorems below, specifically in the context of progressive enlargement. \par
It is now only natural to ask whether this property would be preserved under an enlargement of the filtration. In the sequel, we consider initial and progressive enlargements separately.

\subsection{MRP under Initial Enlargement}

There are several important results on the stability of the MRP under initial enlargement with varying degrees of generality. We opted to cover both the well-known classical results and the more recent theorems in the most general setting. The rationale is that there are important ideas present in settings which might not be the most general, yet provide insights that are conceptually valuable in their own right. \par   
The key condition that allows to retain the MRP under initial enlargement is still Jacod's criterion. We begin with the discussion of the classical result of (Amedinger, 2000). Consider a stochastic basis $(\Omega, \mathcal{F}, \mathbb{F}, \mathbb{P})$ with a filtration that satisfies the usual hypotheses, $\mathcal{F}_0$ trivial and $\mathcal{F}_s = \mathcal{F}_T$ for all $s \geq T$ and some $T >0$. Consider a $d$-dimensional cadlag process $S = (S^1, ..., S^d)^\top$ and assume that there exists a probability measure $\mathbb{Q}^{\mathbb{F}}$ equivalent to $\mathbb{P}$ such that each component of $S$ is in $\mathcal{H}^2(\mathbb{Q}^{\mathbb{F}}, \mathbb{F})$. Let $Z^{\mathbb{F}}$ be the density process of $\mathbb{Q}^{\mathbb{F}}$ with respect to $\mathbb{P}.$ Consider an enlargement of $\mathbb{F}$ given by $\mathbb{G} \coloneqq \mathbb{F} \lor \sigma(G)$, where $G$ is a random variable taking values in a Polish space $(X, \mathscr{X}).$ Assume that Jacod's equivalence criterion holds, i.e. the regular conditional distributions of $G$ given $\mathcal{F}_t$, $t \in [0,T]$, are equivalent to the law of $G$ $\mathbb{P}$-a.s., i.e.
\begin{align*}
    \mathbb{P}(G \in \cdot \vert \mathcal{F}_t) \sim \mathbb{P}(G \in \cdot ) \ \text{for all} \ t \in [0,T] \ \mathbb{P}-\text{a.s.}
\end{align*}
\textcolor{black}{As was mentioned in the previous sections}, the existence of a good version of the conditional density process $(\omega,t,x) \mapsto p^x_t(\omega)$ follows from the equivalence assumption (see, e.g., Jacod, 1985, Lemma 1.8). An important step in the overall derivation is to show how the properties of the stochastic processes are transferred from $\mathbb{F}$ to $\mathbb{G}$ via the so-called \textit{martingale preserving probability measure} (Amedinger, 2000, Theorem 3.1., p. 104). \par

\begin{thm} Assume Jacod's equivalence criterion holds \textcolor{black}{and $p$ is the associated conditional density process}. Then \\
(i) $Z^{\mathbb{G}} \coloneqq \frac{Z^{\mathbb{F}}}{p^G}$ is a $\mathbb{G}$-martingale under $\mathbb{P}$; \\
(ii) the martingale preserving probability measure under initial enlargement given by
\begin{align*}
    \mathbb{Q}^{\mathbb{G}}(A) \coloneqq \int_A\frac{Z^{\mathbb{F}}_T}{p^G_T}d\mathbb{P} \quad \text{for} \ A \in \mathcal{G}_T
\end{align*}
has the following properties \\
(a) $\mathcal{F}_T$ and $\sigma(G)$ are independent under $\mathbb{Q}^{\mathbb{G}}$; \\
(b) $\mathbb{Q}^{\mathbb{G}} = \mathbb{Q}^{\mathbb{F}}$ on $(\Omega, \mathcal{F}_T)$ and $\mathbb{Q}^{\mathbb{G}} = \mathbb{P}$ on $(\Omega, \sigma(G))$, \\
i.e. for $A_T \in \mathcal{F}_T$ and $B \in \mathscr{X}$ we have 
\begin{align*}
    \mathbb{Q}^{\mathbb{G}}(A_T \cap \{G \in B\}) = \mathbb{Q}^{\mathbb{F}}(A_T)\mathbb{P}(G \in B) = \mathbb{Q}^{\mathbb{G}}(A_T)\mathbb{Q}^{\mathbb{G}}(G \in B).
\end{align*}
\end{thm}
\begin{proof}
    The key to the proof is the equality (F\"ollmer and Imkeller, 1993, p.574; Amedinger et al, 1998, Proposition 2.3):
\begin{align*}
  \mathbb{E}\left[\mathbbm{1}_{\{G \in B\}} \frac{1}{p^{\mathbb{G}_T}} \Bigg\vert\mathcal{F}_T\right] = \int_B \frac{1}{p^x_T(\omega)}p^x_T(\omega)\mathbb{P}(G \in dx) = \mathbb{P}(G \in B) 
\end{align*}
\par
We provide the details (see Es-Saghouani, 2006) to introduce some of the proof techniques. We begin with the second assertion. Consider the sets $A_T \in \mathcal{F}_T$ and $B \in \mathscr{X}$, condition on $\mathcal{F}_T$, use the definition of $p^{\mathbb{G}}_T$ and the equivalence condition to get
\begin{align*}
    \mathbb{Q}^{\mathbb{G}}\left[A_T \cap \{G \in B\}\right] &= \mathbb{E}\left[\mathbbm{1}_{A_T\cap\{G \in B\}} \frac{Z^{\mathbb{F}}_T}{p^G_T}\right] = \mathbb{E}\left[\mathbbm{1}_{A_T}\mathbb{E}\left[\mathbbm{1}_{\{G \in B\}} \frac{Z^{\mathbb{F}}_T}{p^G_T} \Bigg\vert \mathcal{F}_T \right]\right] \\
    &= \int_{A_T}\mathbb{E}\left[\mathbbm{1}_{\{G \in B\}} \frac{Z^{\mathbb{F}}_T}{p^G_T} \Bigg\vert \mathcal{F}_T \right](\omega)\mathbb{P}(d\omega) \\
    &= \int_{A_T}\int_B \frac{Z^{\mathbb{F}}_T(\omega)}{p^x_T(\omega)}\mathbb{P}(G \in dx \vert \mathcal{F}_T)\mathbb{P}(d\omega) \\
    &= \int_{A_T}\int_B \frac{Z^{\mathbb{F}}_T(\omega)}{p^x_T(\omega)}p^x_T(\omega) \mathbb{P}(G \in dx)\mathbb{P}(d\omega) \\
    &= \int_{A_T}\int_B Z^{\mathbb{F}}_T(\omega) \mathbb{P}(G \in dx)\mathbb{P}(d\omega) \\
    &= \int_{A_T} Z^{\mathbb{F}}_T(\omega)\int_B \mathbb{P}(G \in dx)\mathbb{P}(d\omega) = \mathbb{Q}^{\mathbb{F}}(A_T)\mathbb{P}(G \in B).
\end{align*}
Noting that $\mathbb{Q}^{\mathbb{G}}\left[\Omega \cap \{G \in B\}\right] = \mathbb{Q}^{\mathbb{G}}[G \in B] = \mathbb{P}(G \in B)$ and $\mathbb{Q}^{\mathbb{G}}\left[A_T \cap \{G \in X\}\right] = \mathbb{Q}^{\mathbb{G}}[A_T] = \mathbb{Q}^{\mathbb{F}}(A_T)$, we establish the second claim. To show the martingale property of $Z^{\mathbb{G}} \coloneqq \frac{Z^{\mathbb{F}}}{p^G}$, take $0 \leq s \leq t \leq T$, $A_s \in \mathcal{F}_s$, $B \in \mathscr{X}$, $A = A_s \cap \{G \in B\} \in \mathcal{G}_s$ and note
\begin{align*}
    \mathbb{E}\left[\mathbbm{1}_A \frac{Z^{\mathbb{F}}_t}{p^G_t}\right] &= \mathbb{Q}^{\mathbb{F}}[A_s]\mathbb{P}[G \in B] = \mathbb{E}\left[\mathbbm{1}_{A_s}Z^{\mathbb{F}}_s\mathbb{P}[G \in B]\right] \\
    &= \int_{A_s}\int_B \frac{Z^{\mathbb{F}}_s(\omega)}{p^x_s(\omega)}p^x_s(\omega) \mathbb{P}(G \in dx)\mathbb{P}(d\omega) \\
    &= \mathbb{E}\left[\mathbbm{1}_{A_s}\mathbb{E}\left[\mathbbm{1}_{\{G \in B\}} \frac{Z^{\mathbb{F}}_s}{p^G_s} \Bigg\vert \mathcal{F}_s \right]\right] = \mathbb{E}\left[\mathbbm{1}_A \frac{Z^{\mathbb{F}}_s}{p^G_s}\right].
\end{align*}
Since $Z^{\mathbb{F}} = 1 = p^G_0$ and hence $\frac{Z^{\mathbb{F}}}{p^G_0} = 1$, then $\mathbb{Q}^{\mathbb{G}}(A) \coloneqq \int_A\frac{Z^{\mathbb{F}}_T}{p^G_T}$ is indeed a probability measure on $(\Omega, \mathcal{G}_T)$.
\end{proof}

 \par
The following theorem shows that the martingale property is preserved when we swtich from $\mathbb{F}$ to $\mathbb{G}$ and from $\mathbb{Q}^{\mathbb{F}}$ to $\mathbb{Q}^{\mathbb{G}}$. \par

\begin{thm} Assume Jacod's equivalence criterion holds. Then $\forall p \in [1,\infty]$ we have
\begin{align*}
    \mathcal{H}^p_{(\text{loc})}(\mathbb{Q}^{\mathbb{F}}, \mathbb{F}) = \mathcal{H}^p_{(\text{loc})}(\mathbb{Q}^{\mathbb{G}}, \mathbb{F}) \subseteq \mathcal{H}^p_{(\text{loc})}(\mathbb{Q}^{\mathbb{G}}, \mathbb{G})
\end{align*}
and, in particular,
\begin{align*}
        \mathcal{M}^p_{(\text{loc})}(\mathbb{Q}^{\mathbb{F}}, \mathbb{F}) = \mathcal{M}^p_{(\text{loc})}(\mathbb{Q}^{\mathbb{G}}, \mathbb{F}) \subseteq \mathcal{M}^p_{(\text{loc})}(\mathbb{Q}^{\mathbb{G}}, \mathbb{G}).
\end{align*}
Furthermore, any $(\mathbb{Q}^{\mathbb{F}},\mathbb{F})$-Brownian motion is $(\mathbb{Q}^{\mathbb{G}},\mathbb{G})$-Brownian motion. 
\end{thm}
\begin{proof}
We follow the original proof of Amedinger (2000). By independence of $\sigma(G)$ and $\mathcal{F}_T$ under $\mathbb{Q}^{\mathbb{G}}$, and the fact that $\mathbb{Q}^{\mathbb{G}} = \mathbb{Q}^{\mathbb{F}}$ on $(\Omega, \mathcal{F}_T)$, it follows that any
$(\mathbb{Q}^{\mathbb{F}},\mathbb{F})$-martingale is a $(\mathbb{Q}^{\mathbb{G}},\mathbb{G})$-martingale. Note that $\mathbb{F}$-stopping times are also $\mathbb{G}$-stoppping times, hence $(\mathbb{Q}^{\mathbb{F}},\mathbb{F})$-localization carries over to $(\mathbb{Q}^{\mathbb{G}},\mathbb{F})$ and $(\mathbb{Q}^{\mathbb{G}},\mathbb{G})$. Since $\mathbb{Q}^{\mathbb{G}} = \mathbb{Q}^{\mathbb{F}}$ on $(\Omega, \mathcal{F}_T)$, the integrability conditions required for $\mathcal{H}^p$ and $\mathcal{M}^p$ are satisfied. Note that the quadratic variation of continuous martingales is computed pathwise and does not involve the filtration. Since $\mathbb{Q}^{\mathbb{G}} = \mathbb{Q}^{\mathbb{F}}$ on $(\Omega, \mathcal{F}_T)$, we have for $t \in [0,T]$ that 
\begin{align*}
   \langle W \rangle ^{(\mathbb{Q}^{\mathbb{G}},\mathbb{G})}_t = \langle W \rangle ^{(\mathbb{Q}^{\mathbb{G}},\mathbb{F})}_t = \langle W \rangle ^{(\mathbb{Q}^{\mathbb{F}},\mathbb{F})}_t = t
\end{align*}
and hence W is also a $(\mathbb{Q}^{\mathbb{G}},\mathbb{G})$-Brownian motion.
\end{proof}

\par
We are now ready to state the result on the transfer of the MRP from $\mathbb{F}$ to $\mathbb{G}$. We hence need to make an additional assumption that the MRP holds in $\mathbb{F}$ to begin with. Specifically, we assume that for any $F \in L^{\infty}(\mathcal{F}_T)$, there exists a process $\phi \in \mathcal{L}^2(S,\mathbb{Q}^{\mathbb{F}}, \mathbb{F})$ such that $F = \mathbb{E}_{\mathbb{Q}^{\mathbb{F}}}[F] + \int_0^T \phi^*_s dS_s$. Then the following theorem (Amedinger, 2000, Theorem 4.2., p. 106) holds. \par

\begin{thm} Assume Jacod's equivalence criterion holds and the MRP holds for the process $S \in \mathcal{H}^2_{(\text{loc})}(\mathbb{Q}^{\mathbb{F}}, \mathbb{F})$. Then the following statements are true. \\
(i) For any $K \in \mathcal{H}^2(\mathbb{Q}^{\mathbb{G}}, \mathbb{G})$, there exists a process $\psi \in \mathcal{L}^2(S,\mathbb{Q}^{\mathbb{G}}, \mathbb{G})$ such that 
\begin{align*}
    K_t = K_0 + \int_0^t \psi^*_s dS_s, \quad t \in [0,T].
\end{align*}
(ii) For any $K \in \mathcal{M}_{\text{loc}}(\mathbb{Q}^{\mathbb{G}}, \mathbb{G})$, there exists a process $\psi \in \mathcal{L}^1_{\text{loc}}(S,\mathbb{Q}^{\mathbb{G}}, \mathbb{G})$ such that 
\begin{align*}
    K_t = K_0 + \int_0^t \psi^*_s dS_s, \quad t \in [0,T].
\end{align*}
\end{thm}
\begin{proof}
 We provide a sketch of the proof, details can be found in the original paper. To prove the first claim it suffices to show that any random variable $H \in L^2(\mathbb{Q}^{\mathbb{G}}, \mathcal{G}_T)$ is of the form $H = \mathbb{E}_{\mathbb{Q}^{\mathbb{G}}}[H \vert \mathcal{G}_0] + \int_0^T \psi^*_s dS_s$ for some process $\psi \in \mathcal{L}^2(S,\mathbb{Q}^{\mathbb{G}}, \mathbb{G})$. Since $\mathcal{G}_T = \mathcal{F}_T \lor \sigma(G)$, and $V \coloneqq \left\{ h \in L^{\infty}(\mathcal{G}_T): h = \sum_{i=1}^mf_ig_i \ \text{with} \ f_i \in L^{\infty}(\mathcal{F}_T), g_i \in L^{\infty}(\sigma(G)), m \in \mathbb{N} \right\}$ is dense in $L^2(\mathbb{Q}^{\mathbb{G}}, \mathcal{G}_T)$, there exists a sequence of the form $(h_n)_{n \in \mathbb{N}} = (\sum_{i=1}^{m(n)}F^{i,n}G^{i,n})_{n \in \mathbb{N}}$ in $V$, with $F^{i,n} \in L^{\infty}(\mathcal{F}_T)$ and $G^{i,n} \in L^{\infty}(\sigma(G))$, such that $h_n$ converges to $H$ in $L^2(\mathbb{Q}^{\mathbb{G}})$. Since $\mathbb{Q}^{\mathbb{G}} = \mathbb{Q}^{\mathbb{F}}$ on $(\Omega, \mathcal{F}_T)$ and the MRP holds in $\mathbb{F}$, there exist processes $\phi^{i,n} \in \mathcal{L}^2(S,\mathbb{Q}^{\mathbb{G}}, \mathbb{F})$ such that the representation  $F^{i,n} = \mathbb{E}_{\mathbb{Q}^{\mathbb{G}}}[F^{i,n}] + \int_0^T (\phi^{i,n}_s)^* dS_s$ holds. It can then be argued that the representation $h_n = \mathbb{E}_{\mathbb{Q}^{\mathbb{G}}}[h_n \vert \mathcal{G}_0] + \int_0^T (\psi^n_s)^* dS_s$ holds, where $\psi^n \coloneqq \sum_{i=1}^{m(n)}G^{i,n}\phi^{i,n}$. Since the mapping $\theta \mapsto \int \theta^* dS$ is an isometry from $(\mathcal{L}^2(S,\mathbb{Q}^{\mathbb{G}}, \mathbb{G}), \lVert \cdot \rVert_{\mathcal{L}^2(S,\mathbb{Q}^{\mathbb{G}}, \mathbb{G})})$ to $(\mathcal{H}^2(\mathbb{Q}^{\mathbb{G}}, \mathbb{G}), \lVert \cdot \rVert_{\mathcal{H}^2(\mathbb{Q}^{\mathbb{G}}, \mathbb{G})})$, it follows that the space of stochastic integrals $\left\{ \int_0^T \theta^*_s dS_s : \theta \in \mathcal{L}^2(S,\mathbb{Q}^{\mathbb{G}}, \mathbb{G}) \right\}$ is closed in $L^2(\mathbb{Q}^{\mathbb{G}})$, hence there exists some process $\psi \in \mathcal{L}^2(S,\mathbb{Q}^{\mathbb{G}}, \mathbb{G})$ such that $\int_0^T(\psi^n_s)^*dS_s$ converges to $\int_0^T\psi^*_s dS_s$ in $L^2(\mathbb{Q}^{\mathbb{G}})$, and therefore, taking the limit, we get $H = \lim_{n \to \infty} h_n = \lim_{n \to \infty} (\mathbb{E}_{\mathbb{Q}^{\mathbb{G}}}[h_n \vert \mathcal{G}_0] + \int_0^T (\psi^n_s)^* dS_s) = \mathbb{E}_{\mathbb{Q}^{\mathbb{G}}}[H \vert \mathcal{G}_0] + \int_0^T \psi^*_s dS_s$. This proves the first claim. The second claim follows by localizaton and the right-continuity of $\mathbb{G}.$ 
 \end{proof}
 
 \par
Summarizing the above results, we see that Jacod's equivalence criterion allows us to define under $\mathbb{Q}^{\mathbb{G}}$ the martingale preserving probability measure on $(\Omega, \mathcal{F}_T)$ with the so-called \textbf{decoupling property}. Under this new measure $\mathbb{Q}^{\mathbb{G}}$ the $\sigma$-algebras $\mathcal{F}_T$ and $\sigma(G))$ become \textit{independent}, $\mathbb{Q}^{\mathbb{G}} = \mathbb{Q}^{\mathbb{F}}$ on $(\Omega, \mathcal{F}_T)$, and hence any $\mathbb{F}$-local martingale under $\mathbb{Q}^{\mathbb{F}}$ remains a $\mathbb{G}$-local martingale under $\mathbb{Q}^{\mathbb{G}}$. \par

\bigskip

We will now discuss conditions under which the PRP is preserved for $(\mathbb{P}, \mathbb{F}^{\sigma(\zeta)})$-local martingales when the original filtration is generated by a point process and a continuous martingale (e.g. a Brownian motion). The results were obtained in (Grorud, Pontier, 2001) and represent a notable extension of generality, although Jacod's stronger \textit{equivalence} condition is assumed. \par
Consider a stochastic basis $(\Omega, \mathcal{F}, \mathbb{F}, \mathbb{P})$, where $\mathbb{F}$ satisfies the usual hypotheses, and an initial enlargement given by $\mathbb{G} = \mathbb{F}^{\sigma(\zeta)} = \mathbb{F} \lor \sigma(\zeta)$ that is interpreted as private information available to the insider at $t=0$. Assume that the price process follows \par
\begin{align*}
    S^i_t = S^i_0 + \int_0^tS^i_s(b^i_sds+\sigma^i_sdW_s) + \int_0^tS^i_{s-}\int_O\phi^i(x,s)N(dx,ds), \quad 0 \leq t \leq T, i = 1,...,d,
\end{align*}
where $W$ is an $m$-dimensional Brownian motion, $N$ is an $n$-dimensional point process that has strictly positive intensity $\nu(x)dxds$, $x\in O$ an open set in $\mathbb{R}^n$, and for which exists the PRP.
We also assume (hypothesis $P$) that there exist two $\mathbb{G}$-predictable processes $l$ and $\delta$ such that
\begin{align*}
    B_t = W_t - \int_0^tl_sds, \quad t<T, \quad \text{is a} \ (\mathbb{G},\mathbb{P})\text{-Brownian motion,} \\
    \overline{N_t} = N_t - \int_0^t\int_O\nu(y)(1+\delta(y,s))dyds, \quad t<T, \quad \text{is a} \ (\mathbb{G},\mathbb{P})\text{-martingale.}
\end{align*}
\par

\begin{thm} Assume that hypothesis $P$ holds, $\zeta$ satisfies Jacod's equivalence criterion and $\mathbb{E}_{\mathbb{P}}\left(\mathcal{E}_t(-l \bullet B)\mathcal{E}_t(-\delta \bullet \overline{N})\right) = 1$ for all $t < T$, where $\mathcal{E}(\cdot)$ are the stochastic exponentials. Then there exists a probability measure $\mathbb{Q}$ equivalent to $\mathbb{P}$ such that under $\mathbb{Q}$, $\forall t < T$, $\mathcal{F}_t$ is independent of $\sigma(\zeta)$. 
\end{thm}

\par
The implied independence of $\sigma$-algebras unsurprisingly outputs a martingale representation result. We denote by $\mathbb{G}_A$ the filtration defined by $(\mathcal{G}_t = \mathcal{F}_t \lor \sigma(\zeta) ; t \leq A)$. \par

\begin{thm} Assume that either \par
(a) there exists a probability measure $\mathbb{Q}$ equivalent to $\mathbb{P}$ such that under $\mathbb{Q}$, $\forall t < T$, $\mathcal{F}_t$ is independent of $\sigma(\zeta)$ \par
or \par
(b) hypothesis $P$ and $\mathcal{E}_t(-l \bullet B)\mathcal{E}_t(-\delta \bullet \overline{N})$ is a $(\mathbb{G}_A,\mathbb{P})$-integrable martingale; \par
then $\forall A<T$, every $(\mathbb{G}_A,\mathbb{Q})$-local martingale can be represented as \par
\begin{align*}
    Z_t = \mathbb{E}_{\mathbb{Q}}(Z \vert \mathcal{G}_0) + \int_0^t\chi_sdW_s + \int_0^t\int_O\xi(y,s)(N(dy,ds) - \nu(y)dyds), \quad t\leq A
\end{align*}
for a unique predictable process $(\chi,\xi).$
\end{thm}
\begin{proof}
See (Jacod, Shiryaev, 2002, Theorem 4.34, p.176).
\end{proof} 

It can actually be shown that the existience of the equivalent probability measure $\mathbb{Q}$ implies hypothesis $P$ and the existence of predictable processes $l$ and $\delta$ such that $B$ is a $(\mathbb{G},\mathbb{P})$-Brownian motion, $\overline{N}$ is a $(\mathbb{G},\mathbb{P})$-martingale and $\mathbb{E}_{\mathbb{P}}\left(\mathcal{E}_t(-l \bullet B)\mathcal{E}_t(-\delta \bullet \overline{N})\right) = 1$ for all $t < T$. \par
Finally, we present another result (Callegaro et al, 2011) on the stability of the PRP under initial (as well as progressive) enlargement when Jacod's equivalence criterion is assumed. Consider a complete stochastic basis $(\Omega, \mathcal{F}, \mathbb{F}, \mathbb{P})$ and an enlargement given by $\mathbb{G} \coloneqq \mathbb{F} \lor \sigma(\zeta).$ $\mathbb{P}^*$ denotes the decoupling measure as defined above for the results of Amedinger. Also recall that any $(\mathbb{P},\mathbb{F})$-local martingale $X$ is a $(\mathbb{P}, \mathbb{G})$-semimartingale with canonical decomposition  \par
\begin{align*}
    X_t = \widetilde{X}_t + \int_0^t\frac{d\langle X, p^{\zeta} \rangle_s}{p_{s-}^{\zeta}},
\end{align*}
where $\widetilde{X}$ is a $(\mathbb{P}, \mathbb{G})$-local martingale \textcolor{black}{and $p$ is the usual conditional density process}. \par

\begin{thm} Assume that there exists a process $Z$ (not necessarily continuous) that possesses the PRP for $(\mathbb{P},\mathbb{F})$. Then every $\widetilde{X} \in \mathcal{M}_{\text{loc}}(\mathbb{P}^*, \mathbb{G})$ admits a representation \par
\begin{align*}
    \widetilde{X}_t = \widetilde{X}_0 + \int_0^t\widetilde{\Phi}_sdZ_s,
\end{align*}
where $\widetilde{\Phi} \in \mathcal{L}(Z,\mathbb{P}^*,\mathbb{G})$. In the special case where $\widetilde{X} \in \mathcal{M}^2(\mathbb{P}^*, \mathbb{G})$, it holds that \\ $\mathbb{E}_{\mathbb{P}^*}\left(\int_0^t(\widetilde{\Phi}_s)^2d[Z]_s\right)<\infty$, for all $t\geq0$, and the representation is unique. 
\end{thm}

\bigskip

\par
We now relax Jacod's equivalence criterion and consider the case where only absolute continuity is satisfied. The following is a well-known result on the predictable representation property under initial enlargement (see, e.g., Askamit, Jeanblanc, 2017, pp. 92-93). \par

\begin{thm} Assume that $\xi$ satisfies Jacod's absolute continuity condition and there exists a process $X \in \mathcal{M}_{loc}(\mathbb{F})$ such that every $Y \in \mathcal{M}_{loc}(\mathbb{F})$ can be represented as $Y = Y_0 + \phi \bullet X$ for some $X \in \mathcal{L}(X,\mathbb{F}),$ where $\mathcal{L}(X,\mathbb{F}) = \left\{ \phi \
\text{predictable} : \left(\int_0^{.}\phi^2_sd[X]_s\right)^{1/2} \text{locally integrable} \right\}$. Then every $Y \in \mathcal{M}_{loc}(\mathbb{F}^{\sigma(\xi)})$ admits a representation
\begin{align*}
    Y_t = Y_0 + (\Phi \bullet \overline{X})_t,
\end{align*}
where $\Phi \in \mathcal{L}(\overline{X},\mathbb{F}^{\sigma(\xi)}), Y_0 \in \mathcal{F}^{\sigma(\xi)}_0$ and $\overline{X}$ is the $\mathbb{F}^{\sigma(\xi)}$-local martingale component in the $\mathbb{F}^{\sigma(\xi)}$-semimartingale decomposition of $X$, given by 
\begin{align*}
    \overline{X}_t \coloneqq X_t - \int_0^t \frac{1}{p_{s-}(\zeta)}d[X, p(\zeta)]_s + \left(\mathbbm{1}_{\llbracket\tilde{R}(u),\infty\llbracket} \Delta X_{\tilde{R}(u)}\right)^{p,\mathbb{F}}_t\vert_{u = \zeta},
\end{align*}
where $\tilde{R}(u) \coloneqq R(u)_{\{p_{R(u)-}(u)>0\}}$, and $R(u) \coloneqq \text{inf} \{t : p_{t-}(u) = 0\},$ as defined in the previous section.
\end{thm}
\begin{proof} The proof is rather technical. A sketch can be found in (Askamit, Jeanblanc, 2017, p.93, Theorem 4.34). See (Fontana, 2017, p.13, Theorem 2.6) for a full proof with all the details.
\end{proof}
\par
\begin{example} (Fontana, 2017, p. 9) Consider a standard Brownian motion $W = (W_t)_{t\in[0,T]}$ on a stochastic basis $(\Omega, \mathcal{F}_T, \mathbb{F}, \mathbb{P})$ that satisfies the usual conditions, and define a process by $S = (S_t)_{t\in[0,T]} \coloneqq 1 + W$. Consider a discrete random variable $\xi \coloneqq \mathbbm{1}_{\{\text{inf}_{t \in [0,T]}S_t > 0\}}$. Thus, Jacod's criterion is satisfied, therefore 
\begin{align*}
p_t^1 = \frac{\mathbb{P}(\text{inf}_{u\in[0,t]}S_u>0 \vert \mathcal{F}_t)}{\mathbb{P}(\text{inf}_{u\in[0,t]}S_u>0)} = 1 + \frac{1}{\mathbb{P}(\text{inf}_{u\in[0,t]}S_u>0)}\sqrt{\frac{2}{\pi}}\int_0^{\sigma \land t}\frac{1}{\sqrt{T-s}}\text{exp}\left(-\frac{S_s^2}{2(T-s)}\right)dW_s, \\
p_t^0 = \frac{\mathbb{P}(\text{inf}_{u\in[0,t]}S_u\leq0 \vert \mathcal{F}_t)}{\mathbb{P}(\text{inf}_{u\in[0,t]}S_u\leq0)} = 1 + \frac{1}{\mathbb{P}(\text{inf}_{u\in[0,t]}S_u\leq0)}\sqrt{\frac{2}{\pi}}\int_0^{\sigma \land t}\frac{1}{\sqrt{T-s}}\text{exp}\left(-\frac{S_s^2}{2(T-s)}\right)dW_s,
\end{align*}
where $\sigma = \text{inf}\{t \in \mathbb{R}_+ \vert S_t = 0\}.$ Note that there is a positive probability that $p^1$ and $p^0$ hit 0 in a continuous way. Note that $W$ has the predictable representation property on $(\Omega, \mathcal{F}_T, \mathbb{F}, \mathbb{P})$. Since $p^{\xi}_t = p^1_t\mathbbm{1}_{\{\sigma > T\}} + p^0_t\mathbbm{1}_{\{\sigma \leq T\}}$, the predictable representation property holds in the enlarged stochastic basis $(\Omega, \mathcal{F}_T, \mathbb{F}^{\sigma(\xi)}, \mathbb{P})$ with respect to the following process
\begin{align*}
    \overline{S}_t &= S_t - \int_0^t\frac{1}{p^{\xi}_s}d\langle S, p^x\rangle_s^{\mathbb{F}}\vert_{x = \xi} \\
    &= S_t - \frac{1}{\mathbb{P}(\text{inf}_{u\in[0,t]}S_u>0)}\sqrt{\frac{2}{\pi}}\int_0^{\sigma \land t}\frac{1}{\sqrt{T-s}}\text{exp}\left(-\frac{S_s^2}{2(T-s)}\right)dW_s \\
    &+  \frac{1}{\mathbb{P}(\text{inf}_{u\in[0,t]}S_u\leq0)}\sqrt{\frac{2}{\pi}}\int_0^{\sigma \land t}\frac{1}{\sqrt{T-s}}\text{exp}\left(-\frac{S_s^2}{2(T-s)}\right)dW_s.
\end{align*}
\end{example}

\begin{example} (Fontana, 2017, p. 8) Consider a standard Poisson process $N = (N_t)_{t\in[0,T]}$ on a stochastic basis $(\Omega, \mathcal{F}_T, \mathbb{F}, \mathbb{P})$ that satisfies the usual conditions, with the sequence of jump times $\{\tau_n\}_{n \in \mathbb{N}}$. Consider the random variable $\xi \coloneqq N_T$. It is well-known that the corresponding conditional density process is given by 
\begin{align*}
    p^n_t &= \frac{\mathbb{P}(\xi = n \vert \mathcal{F}_t)}{\mathbb{P}(\xi = n)} = e^t\frac{(T-t)^{n - N_t}}{T_n}\frac{n!}{(n - N_t)!}\mathbbm{1}_{\{N_t\leq n\}}, \text{for all} \ t < T \\
    p^n_T &= e^TT^{-n}n!\mathbbm{1}_{\{N_T=n\}},
\end{align*}
hence Jacod's criterion is satisfied. Note that $p^n$ hits 0 in a jump at $\tau_{n+1}$ (if $\tau_{n+1} \leq T$), i.e. $R(n) = \tau_{n+1}$ for all $n \in \mathbb{N}$, \textcolor{black}{where $R$ is defined as above}. Since the compensated Poisson process $S \coloneqq N - t$ has the predictable representation property on $(\Omega, \mathcal{F}_T, \mathbb{F}, \mathbb{P})$, the property also holds  in the enlarged stochastic basis $(\Omega, \mathcal{F}_T, \mathbb{F}^{\sigma(\xi)}, \mathbb{P})$ with respect to the following process (we omit the details, which can be found in the original paper)
\begin{align*}
    \overline{S}_t = N_t - t - \sum_{n+1}^{N_T}\mathbbm{1}_{\{\tau_n\leq t\}}\left(1 - \frac{T - \tau_n}{N_t - n + 1}\right) + (t - \tau_{N_T}\land t).
\end{align*}
\end{example}

\bigskip

\subsection{MRP under Progressive Enlargement}
The preservation of the Martingale Representation Property under progressive enlargement requires finer analysis and the results are less general. For example, it is well-known that the Brownian motion loses the MRP under progressive enlargement by a random time which is not a stopping time with respect to its natural filtration. As before, we do not restrict ourselves to immediately stating the most general theorems available, and rather trace the important ideas and strands in the literature. \par
We begin with a number of negative results, i.e., following (Calzolari, Torti, 2016a), we investigate in greater detail the conditions under which the loss of the PRP occurs. \par
Consider a square-integrable semimartingale $X$ defined on a stochastic basis $(\Omega, \mathcal{F}_T, \mathbb{F}, \mathbb{P})$, and take $\mathbb{G}$ to be an arbitrary enlargement of $\mathbb{F}$. Denote by $\mathfrak{M}(X,\mathbb{G})$ the set of probability measures on $(\Omega, \mathcal{G}_T)$ under which $X$ is a $\mathbb{G}$-martingale and which are equivalent to $\mathbb{P}$.
\par

\begin{thm} Assume that $\mathcal{G}_0$ is $\mathbb{P}$-trivial and $u \coloneqq \text{inf}\{t \in [0,T] : \mathcal{F}_t \subsetneq \mathcal{G}_t \} = \text{min}\{t \in [0,T] : \mathcal{F}_t \subsetneq \mathcal{G}_t \}.$ If $\mathfrak{M}(X,\mathbb{G}) \neq \emptyset$, then for any $\mathbb{P} \in \mathfrak{M}(X,\mathbb{G})$ 
\begin{align*}
    K^2(\Omega,\mathbb{G},\mathbb{P},X) \subsetneq L^2_0(\Omega,\mathcal{G}_T,\mathbb{P}),
\end{align*}
where $ K^2(\Omega,\mathbb{G},\mathbb{P},X) \coloneqq \left\{\int_0^T\xi dX_s, \xi \in \mathcal{L}^2(X,\mathbb{P},\mathbb{F})\right\}$. \end{thm}
\begin{proof} See (Calzolari, Torti, 2016a, Theorem 3.4, p. 6). 
\end{proof}
\par
An immediate corollary of the strict inclusion in the preceding theorem is that if the set of equivalent probability measures is a singleton, then the PRP of $X$ is destroyed under enlargement, i.e. \par
\begin{cor} Assume $\mathfrak{M}(X,\mathbb{F}) = \{\mathbb{P}\}$. Then, under the hypotheses of the previous theorem, $X$ does not possess the PRP in $\mathbb{G}$. 
\end{cor}

\par
In the context of initial enlargement by a random variable $\xi$, assuming the existence of a probability measure $\mathbb{Q}$ equivalent to $\mathbb{P}$ under which $\mathcal{F}_T$ and $\sigma(\xi)$ are independent guarantees the preservation of the PRP, whereas now the same assumption results in the loss of that property. We quote the general result and a discussion on the resolution of the paradox (see Calzolari, Torti, 2016a, p. 17 for more details). \par

\begin{thm} Assume $\mathfrak{M}(X,\mathbb{F}) = \{\mathbb{P}\}$ and $\mathbb{G} \coloneqq \mathbb{F} \lor \mathbb{K}$, where $\mathbb{K}$ is a filtration on $(\Omega, \mathcal{F}, \mathbb{P})$ that satisfies the usual hypotheses and such that $\mathcal{K}_t \neq \mathcal{K}_0$ for some $t \in (0,T]$. Let $\mathbb{Q}$ be a probability measure on $(\Omega, \mathcal{F})$ equivalent to $\mathbb{P}$ under which $\mathcal{F}_T$ and $\mathcal{K}_T$ are independent. Then $X$ does not possess the PRP in $\mathbb{G}$. 
\end{thm}

\par
The resolution of the paradox lies in the observation that initial and progressive enlargements have fundamentally different impacts on the information set available at time 0. In the case of initial enlargement $\mathcal{F}_0$ is augmented by information contained in $\sigma(\xi)$, whereas under progressive enlargement $\mathcal{F}_0$ remains trivial. Under the assumption that $\mathfrak{M}(X,\mathbb{F})$ is a singleton the initial $\sigma$-algebra $\mathcal{F}_0$ is trivial under $\mathbb{P}$ and, as is well-known, $X$ possesses the PRP in $\mathbb{F}$ under $\mathbb{P}$. Intuitively, this means that all randomness in $\mathbb{F}$-local martingales is contained in $X$ (Calzolari, Torti, 2016a, p. 17). In the enlarged filtration the randomness contained in $X$ is no longer sufficient to represent all $\mathbb{G}$-local martingales. Under initial enlargement this results in the initial values of the martingales becoming random variables, since $\mathcal{F}_0$ is no longer trivial, whereas under progressive enlargement additional stochastic processes are required as stochastic integrators. 
\par

\bigskip

\par
The preceding discussion leads us to the conclusion that more assumptions are needed to salvage (some form of) the PRP. In the following, we discuss a number of results not necessarily in an increasing (or nested) order of generality. \par
Within the setup of the previous theorems, \textcolor{black}{where $X$ is a square-integrable semimartingale}, take $\mathbb{H}$ to be a filtration on $(\Omega, \mathcal{F}, \mathbb{P})$ that satisfies the usual conditions, with $\mathcal{H}_0$ trivial. Take $Y$ to be a square-integrable $(\mathbb{P},\mathbb{H})$-semimartingale enjoying the PRP in $\mathbb{H}$ under an equivalent probability measure. Denote by $M$ and $N$ the martingale parts of $X$ and $Y$. Denote by $\overline{K^X(\mathbb{F})}$ the closure in $L^1(\Omega,\mathcal{F}_T,\mathbb{P})$ of the set $\left\{\int_0^T\xi dX_s, \ \xi \ \ \mathbb{F}-\text{predictable, simple, bounded} \right\}$, and similarly $\overline{K^Y(\mathbb{H})}$. Denote by $\alpha$ the predictable process associated with $L^2([0,T], \mathcal{B}([0,T]), d\langle M \rangle_t)$ and by $\delta$ the predictable process associated with $L^2([0,T], \mathcal{B}([0,T]), d\langle N \rangle_t)$, such that $A_t = \int_0^t\alpha_sd\langle M \rangle_s$, $D_t = \int_0^t\delta_sd\langle N \rangle_s$. Finally, define the enlarged filtration to be $\mathbb{G} \coloneqq \mathbb{F} \lor \mathbb{H}$. We are going to show that there is a probability measure $\mathbb{Q} \in \mathfrak{M}(X, Y, [X,Y], \mathbb{G})$ such that $(X, Y, [X,Y])$ possesses the PRP in $\mathbb{G}$ under $\mathbb{Q}$.   \par

\begin{thm} Assume that the following conditions hold \\
(i) $\mathfrak{M}(X,\mathbb{F}) = \{\mathbb{P}^X\}$, $\mathfrak{M}(Y,\mathbb{H}) = \{\mathbb{P}^Y\}$ \footnote{\textcolor{black}{The superscripts emphasize the fact that these measures are defined for different processes and different filtrations. Note both sets are singletons and the two measures can be restrictions of a single measure $\mP$ to the corresponding filtrations.}}, \\
(ii) $\overline{K^X(\mathbb{F})}\cap L_+^1(\Omega,\mathcal{F}_T,\mathbb{P}) = \{0\}$, $\overline{K^Y(\mathbb{H})}\cap L_+^1(\Omega,\mathcal{H}_T,\mathbb{P}) = \{0\}$, \\
(iii) $\alpha \Delta M < 1$ $\mathbb{P}$-a.s., $\delta \Delta N < 1$ $\mathbb{P}$-a.s. \\
(iv) 
\begin{align*}
& \mathbb{E}^{\mathbb{P}}\left[\text{exp}\left\{\frac{1}{2}\int_0^T\alpha_t^2d\langle M^c \rangle_t +  \int_0^T\alpha_t^2d\langle M^d \rangle_t  \right\} \right] < \infty \\
& \mathbb{E}^{\mathbb{P}}\left[\text{exp}\left\{\frac{1}{2}\int_0^T\delta_t^2d\langle N^c \rangle_t +  \int_0^T\delta_t^2d\langle N^d \rangle_t  \right\} \right] < \infty. 
\end{align*}
(v) $M$ and $N$ are $(\mathbb{P}, \mathbb{G})$-strongly orthogonal martingales. \par Then the following assertions hold \\
(1) $\mathcal{F}_T$ and $\mathcal{H}_T$ are independent under $\mathbb{P}$, $\mathbb{G}$ satisfies the usual hypotheses and every $W \in \mathcal{M}^2(\mathbb{P}.\mathbb{G})$ has the unique representation
\begin{align*}
    W_t = W_0 + \int_0^t\gamma^W_sdM_s + \int_0^t\kappa^W_sdN_s + \int_0^t\phi^W_sd[M,N]_s \quad \mathbb{P}-\text{a.s.,}
\end{align*}
where $\gamma^W \in \mathcal{L}^2(M,\mathbb{P},\mathbb{G})$, $\kappa^W \in \mathcal{L}^2(N,\mathbb{P},\mathbb{G})$ and $\phi^W \in \mathcal{L}^2([M,N],\mathbb{P},\mathbb{G})$; \\
(2) there exists a probability measure $\mathbb{Q}$ on $(\Omega, \mathcal{G}_T)$ such that $(X,Y, [X,Y])$ possesses the PRP in $\mathbb{G}$ under $\mathbb{Q}$, i.e. every $Z \in \mathcal{M}^2(\mathbb{Q},\mathbb{G})$ has the unique representation
\begin{align*}
    Z_t = Z_0 + \int_0^t\eta^Z_sdX_s + \int_0^t\theta^Z_sdY_s + \int_0^t\zeta^Z_sd[X,Y]_s \quad \mathbb{Q}-\text{a.s.,}
\end{align*}
where $\eta^Z \in \mathcal{L}^2(X,\mathbb{Q},\mathbb{G})$, $\theta^Z \in \mathcal{L}^2(Y,\mathbb{Q},\mathbb{G})$ and $\zeta^Z \in \mathcal{L}^2([X,Y],\mathbb{Q},\mathbb{G})$. 
\end{thm}
\begin{proof} See (Calzolari, Torti, 2016a, Theorem 4.13, p. 14).
\end{proof}

\par
An important special case is worth stating separately as a corollary. \par

\begin{cor} Under the hypotheses of the previous theorem, $[X,Y]_t \equiv 0$ $\mathbb{P}$-a.s. if and only if the $(\mathbb{P},\mathbb{G})$-semimartingale $(X,Y)$ possesses the PRP in $\mathbb{G}.$ 
\end{cor}

\par
These results can be extended to the multidimensional setting. The transfer, however, is not entirely straightforward, due to the well-known differences between \textit{vector} stochastic integrals and multidimensional \textit{componentwise} stochastic integrals. We only provide the statements of the main result and refer to (Calzolari, Torti, 2016b) for the full details. \par
In the following $X$ and $Y$ denote an $m$-dimensional square-integrable $(\mathbb{P}, \mathbb{F})$-semimartingale and an $n$-dimensional square-integrable $(\mathbb{P}, \mathbb{H})$-semimartingale respectively, with the canonical decompositions $X = X_0 + M + A$ and $Y = Y_0 + N + D$, \textcolor{black}{where $M$ and $N$ are local martingales and $A$ and $D$ are finite variation processes}.
$\mathcal{L}^2(M, \mathbb{P}, \mathbb{G})$ denotes the set of $\mathbb{G}$-predictable $m$-dimensional processes $\xi = (\xi_t)_{t \in [0,T]} = ((\xi_t^1,...,\xi^m_t))_{t \in [0,T]}$ such that $\lVert \xi \rVert_{\mathcal{L}^2(M, \mathbb{P}, \mathbb{G})} = \mathbb{E}^{\mathbb{P}}[(\xi \bullet X)_T] < \infty$, where $\xi \bullet X$ denotes the vector stochastic integral, and is thus a one-dimensional stochastic process. The Radon-Nikodym derivate processes are denoted as $L_t^X \coloneqq \frac{d\mathbb{P}^X}{d\mathbb{P} \vert \mathcal{F}_t}$, similarly for $Y$. $[M, N]^V$ denotes the process 
\begin{align*}
    ([M^1,N^1], ..., [M^1,N^n], [M^2,N^1], ..., [M^2,N^n], ..., [M^m,N^1], ..., [M^m,N^n]).
\end{align*}
\par
The first multidimensional result is nearly identical to the one-dimensional case, and only requires an interpretation of stochastic integrals in the vector sense. 
\par
\begin{thm} Assume that the following conditions hold \\
(i) $\mathfrak{M}(X,\mathbb{F}) = \{\mathbb{P}^X\}$, $\mathfrak{M}(Y,\mathbb{H}) = \{\mathbb{P}^Y\}$, \\
(ii) $L_T^X \in L_{\text{loc}}(\Omega, \mathcal{F}_T, \mathbb{P})$, $L_T^Y \in L_{\text{loc}}(\Omega, \mathcal{H}_T, \mathbb{P})$, \\
(iii) for any $i \in (1, ..., m)$ and $j \in (1, ..., n)$, $M^i$ and $N^j$ are real $(\mathbb{P}, \mathbb{G})$-strongly orthogonal martingales. \par
Then the following statements hold. \\
(1) $\mathcal{F}_T$ and $\mathcal{H}_T$ are independent under $\mathbb{P}$; \\
(2) $\mathbb{G}$ satisfies the standard hypotheses; \\
(3) every $W \in \mathcal{M}^2(\mathbb{P}.\mathbb{G})$ has the unique representation
\begin{align*}
    W_t = W_0 + (\gamma^W \bullet M)_t + (\kappa^W \bullet N)_t + (\phi^W \bullet [M,N]^V)_t \quad \mathbb{P}-\text{a.s.,}
\end{align*}
where $\gamma^W \in \mathcal{L}^2(M,\mathbb{P},\mathbb{G})$, $\kappa^W \in \mathcal{L}^2(N,\mathbb{P},\mathbb{G})$ and $\phi^W \in \mathcal{L}^2([M,N]^V,\mathbb{P},\mathbb{G})$; \\
(4) there exists a probability measure $\mathbb{Q}$ on $(\Omega, \mathcal{G}_T)$ such that $(X,Y, [X,Y]^V)$ possesses the PRP in $\mathbb{G}$ under $\mathbb{Q}$, i.e. every $Z \in \mathcal{M}^2(\mathbb{Q},\mathbb{G})$ has the unique representation
\begin{align*}
    Z_t = Z_0 + (\eta^Z \bullet X)_t + (\theta^Z \bullet Y)_t + (\zeta^Z \bullet [X,Y])_t \quad \mathbb{Q}-\text{a.s.,}
\end{align*}
where $\eta^Z \in \mathcal{L}^2(X,\mathbb{Q},\mathbb{G})$, $\theta^Z \in \mathcal{L}^2(Y,\mathbb{Q},\mathbb{G})$ and $\zeta^Z \in \mathcal{L}^2([X,Y]^V,\mathbb{Q},\mathbb{G})$. 
\end{thm}
\begin{proof} See (Calzolari, Torti, 2016b, Theorem 3.4, p. 12). 
\end{proof}

\par
We conclude with an example that combines the ideas discussed above (e.g. the decoupling measure of Amedinger, 2000) with the multidimensional theorem. \par

\begin{example} (Calzolari, Torti, 2016b, Proposition 4.4, p. 20). Consider a stochastic basis $(\Omega, \mathcal{F}, \mathbb{F}, \mathbb{P})$, let $\tau$ be a continuous random variable such that $\mathbb{P}( \tau \in \cdot \vert \mathcal{F}_t) \sim \mathbb{P}(\tau \in \cdot)$, for every $t \in [0,T]$, $\mathbb{P}$-a.s. Let $X$ be an $m$-dimensional square-integrable $(\mathbb{P}, \mathbb{F})$-semimartingale that satisfies (i) and (ii) above. Let $\mathbb{G}$ be defined by $\mathcal{G}_t \coloneqq \cap_{s>t}\mathcal{F}_s\lor\sigma(\tau \land s).$ Define a real process $\lambda$ given by 
\begin{align*}
    \lambda_t = \frac{p_t(t)}{\mathbb{P}(\tau > t \vert \mathcal{F}_t)}, \quad t \in [0,T],
\end{align*}
where the process $(p_t(u))_u$ is given by 
\begin{align*}
    \int_t^Tp_t(u)du = \mathbb{P}(\tau > t \vert \mathcal{F}_t).
\end{align*}
Assume that $\mathbb{F}$ is immersed in $\mathbb{G}$ under $\mathbb{P}$. Then it can be shown that the pair 
\begin{align*}
    \left( M, \mathbbm{1}_{\{\tau \leq \cdot \}} - \int_0^{\tau \land \cdot} \lambda_udu \right)
\end{align*}
is a $(\mathbb{P}, \mathbb{G})$-basis of multidimensional martingales. We focus on the main ideas only. In the setting above, we take $\mathbb{H}$ to be the natural filtration of the process $(\mathbbm{1}_{\{\tau \leq t\}})_t$, and the process $N$ to be given by 
\begin{align*}
    N_t = \mathbbm{1}_{\{\tau \leq t\}} - \int_0^{\tau \land t}\frac{dF_u}{1 - F_u},
\end{align*}
where $F$ is the continuous cdf of $\tau$. It is known that $N$ is a $(\mathbb{P}, \mathbb{H})$-martingale that possesses the PRP in $\mathbb{H}$ under $\mathbb{P}\vert_{\mathcal{H}_T}$, and $M$ possesses the PRP in $\mathbb{F}$ under $\mathbb{P}\vert_{\mathcal{F}_T}$. We know from the results in (Amedinger, 2000) covered above that under Jacod's equivalence condition there exists a decoupling probability measure $\mathbb{P}^*$ on $(\Omega, \mathcal{G}_T)$ under which $\mathbb{F}$ and $\sigma(\tau)$ are independent, and $\mathbb{P}^*\vert_{\mathcal{F}_T} = \mathbb{P}\vert_{\mathcal{F}_T}$, $\mathbb{P}^*\vert_{\mathcal{H}_T} = \mathbb{P}\vert_{\mathcal{H}_T}$. Since $\mathcal{H}_T = \sigma(\tau)$, it follows from the independence that under $\mathbb{P}^*$ for any $i \in (1,...,m)$, $M^i$ and $N$ are strongly orthogonal martingales in $\mathbb{G}$, hence the theorem above applies and $(M,N)$ constitutes a $(\mathbb{P}^*, \mathbb{G})$-multidimensional basis. A bit of more work is required to show that the result carries from $(M,N)$ to the pair $\left( M, \mathbbm{1}_{\{\tau \leq \cdot \}} - \int_0^{\tau \land \cdot} \lambda_udu \right)$ under $\mathbb{P}$. We skip this part.
\end{example}

\par
\bigskip

We give another representation result from default risk modeling due to (Kusuoka, 1999) that is slightly less general than the previous example but is important in it own right. We begin by establishing a number of preliminary results that are standard in the literature on default times and can be found, e.g., in (Aksamit, Jeanblanc, 2017, pp. 31-32). Here the filtration $\mathbb{A}$ is the natural filtration generated by the increasing process $A_t = \mathbbm{1}_{\{\tau \leq t\} }$. \par

\begin{thm} The process $M$ given by 
\begin{align*}
    M_t \coloneqq A_t - \int_0^{\tau\land t}\frac{dF(s)}{1 - F(s-)}
\end{align*}
is an $\mathbb{A}$-martingale. \par
If, in addition, $F = \mathbb{P}(\tau \leq t)$ is absolutely continuous with respect to the Lebesgue measure with density $f$, then the process given by
\begin{align*}
    M_t \coloneqq A_t - \int_0^{\tau\land t}\lambda(s)ds = A_t - \int_0^t\lambda(s)(1 - A(s))ds
\end{align*}
is an $\mathbb{A}$-martingale, where $\lambda(t) = \frac{f(t)}{1 - F(t)}\mathbbm{1}_{\{F(t)<1\}}$ is a deterministic non-negative function of $t$ known as the $\textit{intensity rate}$ of $\tau$. 
\end{thm}

\par
The fundamental importance of this martingale stems from the fact that it possesses the MRP in $\mathbb{A}$ as attested by the following theorem. 
\par

\begin{thm} Any $\mathbb{A}$-local martingale can be represented as a stochastic integral with respect to $M$ as defined above. 
\end{thm}

\par
Consider a stochastic basis $(\Omega, \mathcal{F}, \mathbb{F}, \mathbb{P})$, where the filtration $\mathbb{F}$ is taken to be the $d$-dimensional \textit{weak Brownian filtration} that satisfies the usual conditions, i.e., \textcolor{black}{by definition, there exists a $d$-dimensional Brownian motion $\b$ that possesses the PRP in $\mF$, i.e. every $\mF$-local martingale can be represented as a stochastic integral with respect to $\b$ for some predictable integrand process $h$}. Consider the stochastic processes given by $A^i_t = \mathbbm{1}_{\{\tau_i \leq t\}} $, where the the distributions of the default times $\tau_i$, $i=1,2,...,N$, are continuous and $\mathbb{P}(\tau_k = \tau_l) = 0$ for all $k \neq l$ holds.

Consider an enlarged filtration given by $\mathcal{G}_t = \mathcal{F}_t \lor \sigma(\tau_i \land t; i = 1,2,...,N)$. Assume that the $\mathcal{H}$-hypothesis holds, and there exist $\mathbb{G}$-progressively measurable processes $\lambda^i : [0,T) \times \Omega \mapsto [0,T], k=1,2,..,N$, such that 
\begin{align*}
    M^i(t) = A^i_t - \int_0^t(1 - A^i(s))\lambda^i(s)ds
\end{align*}
are $\mathbb{G}$-martingales under $\mathbb{P}$. Then the following martingale representation result holds. \par

\begin{thm} For any $(\mathbb{P},\mathbb{G})$-square integrable martingale $(Y_t)_{t\in[0,T]}$, there are $\mathbb{G}$-predictable processes $f : [0,T] \times \Omega \mapsto \mathbb{R}^d$ and  $\tilde{f}^i : [0,T] \times \Omega \mapsto \mathbb{R}$,$i =1,...,N,$ such that
\begin{align*}
    \mathbb{E}\left[\int_0^T\lvert f(t)\rvert^2dt\right]<\infty, \quad \mathbb{E}\left[\int_0^T\lvert \tilde{f}^i(t)\rvert^2 \lambda^i(t)dt\right]<\infty, \quad i =1,...,N
\end{align*}
and the following MRP holds
\begin{align*}
    Y_t &= Y_0 + \int_0^tf(s)dB(s) + \sum_{i=1}^{N}\int_0^t\tilde{f}^i(s)dM^i(s) \\
    &= Y_0 + (f \bullet B)_t + (\tilde{f} \bullet M)_t, \quad t\in[]0,T].
\end{align*}
\end{thm}

\par
In particular, when the $\mathbb{G}$-martingale $Y$ is of the form $\mathbb{E}(X (1 - \mathbbm{1}_{\{\tau \leq T\}}) \vert \mathcal{G}_t) = \mathbb{E}(X\mathbbm{1}_{\{T < \tau\}} \vert \mathcal{G}_t)$, where $X$ is $\mathcal{F}_t$-measurable, or for a martingale of the form $\mathbb{E}(h_{\tau} \vert \mathcal{G}_t)$, then following representation holds (see also Jeanblanc, Yor, Chesney, 2009). \par

\begin{thm} Suppose the $\mathcal{H}$-hypothesis holds, $G = 1 - F$ is continuous, where $F_t \coloneqq \mathbb{P}(\tau \leq t \vert \mathcal{F}_t)$, and every $\mathbb{F}$-martingale is continuous. Then the martingale $Y_t = \mathbb{E}(h_{\tau} \vert \mathcal{G}_t)$, where $h$ is an $\mathbb{F}$-predictable process such that $\mathbb{E}\lvert h_{\tau} \rvert) < \infty $, admits the following representation as the sum of a $\mathbb{G}$-continuous martingale and a $\mathbb{G}$-purely discontinuous martingale, i.e.
\begin{align*}
    Y_t = m_0^h + (G^{-1} \bullet m^h)_{t\land\tau} + ((h-J) \bullet M)_{t\land\tau},
\end{align*}
where $m^h$ is the continuous $\mathbb{F}$-martingale given by 
\begin{align*}
    m^h_t = \mathbb{E}\left(\int_0^{\infty}h_udF_u \Bigg\vert \mathcal{F}_t\right),
\end{align*}
and $h - J$ is the $\mathbb{G}$-purely discontinuous martingale, where $J$ given by 
\begin{align*}
    J_t = G_t^{-1}(m^h_t - \int_0^th_udF_u),
\end{align*}
and $J_u = H_u$ on the set $\{u < \tau\}$.
\end{thm}

\par
Assume the MRP holds in $\mathbb{F}$ for a family of $\mathbb{F}$-local martingales $X = (X^1, ..., X^n)$, and there exists a process $\widetilde{X} = ( \widetilde{X}^1, ..., \widetilde{X}^n)$ that is the $\mathbb{G}$-local martingale component in the $\mathbb{G}$-semimartingale decomposition of $X$. \par
Under Jacod's criterion (i.e. for $\mathcal{J}$-times or \textit{initial} times), from the discussion in the previous section on initial enlargement it follows that $(\widetilde{X}, M)$ has the MRP in $\mathbb{G}$, where $M = A - (1 - A_-)\frac{1}{Z_-}\bullet A^p$. \par
We provide a version of this result due to (Callegaro et al, 2011). Consider a complete stochastic basis $(\Omega, \mathcal{F}, \mathbb{F}, \mathbb{P})$ and an enlargement given by $\mathcal{G}_t = \mathcal{F}_t \lor \mathcal{A}_t, \ \mathcal{A}_t = \sigma(\mathbbm{1}_{\{\tau\leq s\}}; s\leq t),$ i.e. the smallest filtration with respect to which a random time $\tau$ is a stopping time. As above, $\mathbb{P}^*$ denotes the decoupling measure. As usual, we define the $(\mathbb{P}^*,\mathbb{G})$-martingale \par
\begin{align*}
    M_t \coloneqq A_t - \int_0^{\tau\land t}\lambda(s)\nu(ds), \quad t\geq0,
\end{align*}
where $\lambda(t) = \frac{1}{1-F(t)}$, $F(t) = \mathbb{P}(\tau\leq t)$ and \textcolor{black}{$\nu$ is the law of $\t$}.
\par

\begin{thm} Assume that there exists a process $Z$ (not necessarily continuous) that possesses the PRP for $(\mathbb{P},\mathbb{F})$. Then every $X \in \mathcal{M}_{\text{loc}}(\mathbb{P}^*, \mathbb{G})$ admits a representation \par
\begin{align*}
    X_t = X_0 + \int_0^t\Phi_sdZ_s + \int_0^t\Psi_sdM_s,
\end{align*}
for some processes $\Phi \in \mathcal{L}(Z,\mathbb{P}^*,\mathbb{G})$ and $\Psi \in \mathcal{L}(M,\mathbb{P}^*,\mathbb{G})$ . In the special case where $X \in \mathcal{M}^2(\mathbb{P}^*, \mathbb{G})$, it holds that, for all $t\geq0$, $\mathbb{E}_{\mathbb{P}^*}\left(\int_0^t\Phi_s^2d[Z]_s\right)<\infty$, $\mathbb{E}_{\mathbb{P}^*}\left(\int_0^t\Psi_s^2\lambda(s)d[M]_s\right)<\infty$ and the representation is unique. 
\end{thm}

\par
As the authors point out, intuitively, since $\mathbb{G} = \mathbb{F} \lor \mathbb{A}$ consists of two filtrations, we need another martingale to salvage the PRP in the enlarged filtration, and a good candidate (that turns out to do the job) seems to be the compensated martingale of $A.$ This is a recurring theme in the topic and is termed the \textbf{multiplicity} of the filtration which is defined to be the minimal number of martingales that allow the PRP. Typically, an enlargement of the filtration results in an increase of the multiplicity (Aksamit, Jeanblanc, 2017, p.99).

\par
A classical result by (Barlow, 1978) for honest times $\tau$ states that under assumption (C) all continuous $\mathbb{G}$-martingales are generated by $\widetilde{X}$. Moreover, for an honest time $\tau$, all square integrable $\mathbb{G}$-martingales are generated by the familes $\left\{\widetilde{X} : X \ \text{is a bounded} \ \mathbb{F}\text{-martingale}\right\}$ and $\left\{ v\mathbbm{1}_{\llbracket \tau, \infty \llbracket} - (v\mathbbm{1}_{\llbracket \tau, \infty \llbracket})^{\mathbb{G},p} : v \ \text{is a bounded} \ \mathcal{G}_{\tau}\text{-measurable random variable}\right\}$. We provide just the statements of the main results and refer to the original paper for full details. \par
Consider a stochastic basis $(\Omega, \mathcal{F}, \mathbb{F}, \mathbb{P})$, where $\mathbb{F}$ satisfies the usual conditions, and an honest time $\tau : \Omega \mapsto [0,\infty]$. Let $\mathbb{G}$ be an enlargement of $\mathbb{F}$ given by 
\begin{align*}
    \mathcal{G}_t = \left\{ \Lambda \in \mathcal{F} : \Lambda = (E \cap \{L \leq t\}) \cup (F \cap \{L>t\}) \ \text{for some} \ E,F \in \mathcal{F}_t \right\} \quad \text{for all} \ t\in\mathbb{R}_+.
\end{align*}
Then it follows that $\mathbb{G}$ also satisfies the usual conditions and $\tau$ is a $\mathbb{G}$-stopping time. \par
Now define $A^o$ and $\prescript{o}{}A$ to be the optional and dual optional projections of $A = (\mathbbm{1}_{\llbracket \tau, \infty \llbracket})_{t \geq 0}$. The following theorem is a crucial ingredient in the derivation of the main result. \par

\begin{thm} Let $M$ be a square integrable $\mathbb{F}$-martingale, and $\widetilde{M}$ be defined by
\begin{align*}
    \widetilde{M}_t = M_t + \int_0^t \left(\frac{1 - A_{s-}}{1 - A^o_{s-}} - \frac{A_{s-}}{A^o_{s-}}\right)d\langle M, A^o - \prescript{o}{}A \rangle_s
\end{align*}
Then $\widetilde{M}$ is a square-integrable $\mathbb{G}$-martingale. 
\end{thm}

\par
\begin{thm} Suppose that $\{M^i : i \in I\}$ is a finite collection of continuous $\mathbb{F}$-local martingales that possesses the PRP in $\mathbb{F}$, i.e. such that if $Y$ is any continuous $\mathbb{F}$-local martingale then there exist $\mathbb{F}$-predictable processes $\phi^i$, $i\in I$, such that
\begin{align*}
    Y_t = \sum_{i \in I} \int_0^t \phi^i_sdM^i_s.
\end{align*}
Then the finite collection  $\{\widetilde{M}^i : i \in I\}$ possesses the PRP in $\mathbb{G}$, i.e. if $Z$ is any continuous $\mathbb{G}$-local martingale, then there exist $\mathbb{G}$-predictable processes $\psi$, $i \in I$, such that 
\begin{align*}
    Z_t = \sum_{i \in I} \int_0^t \psi^i_sd\widetilde{M}^i_s.
\end{align*}
\end{thm}

\par
 The analysis can be extended to the case of purely discontinuous $\mathbb{G}$-martingales. Providing all the details even for the statements would take us too far, we only give the bare essentials to understand the result (see Barlow, 1978, for the full details). $\mu$ is an $\mathbb{F}$-optional integer-valued random measure on $(0,\infty)\times E$, where $E$ is a Lusin space, and $\nu$ is its $\mathbb{F}$-dual predictable projection. $\widetilde{\mu}$ and $\widetilde{\nu}$ are defined similarly for the enlarged filtration $\mathbb{G}$. Then the following results hold. \par
 
\begin{thm} If the finite collection $\{M^i : i \in I\ ; \mu - \nu\}$ has the martingale representation property for $\mathbb{F}$-martingales, then $\{\widetilde{M}^i : i \in I\ ; \widetilde{\mu} - \widetilde{\nu}\}$ has the martingale representation property for $\mathbb{G}$-martingales. 
\end{thm}

\par
\begin{thm} If $\mu - \nu$ has the martingale representation property for purely discontinuous $\mathbb{F}$-martingales, then $\widetilde{\mu} - \widetilde{\nu}$ has the martingale representation property for purely discontinuous $\mathbb{G}$-martingales.
\end{thm}

\par

\bigskip

\par
For a general random time $\tau$ (Jeanblanc, Song, 2015)  show that when $\mathbb{F}$ is immersed in $\mathbb{G}$ and $X_{\tau}$ is $\mathcal{F}_{\tau-}$-measurable, then $\mathcal{G}_{\tau} = \mathcal{G}_{\tau-}$ and $(\widetilde{X},M)$ has the MRP in  $\mathbb{G}$, where $M = A - (1 - A_-)\frac{1}{Z_-}\bullet A^p$. They also show that if $\tau$ is an honest time and $\mathbb{F}$ is a Brownian filtration, then the collection of stochastic processes can be enlarged to salvage the PRP in $\mathbb{G}$. Specifically, the following statement holds (as usual, we work on a complete stochastic basis $(\Omega, \mathcal{F}, \mathbb{F}, \mathbb{P}))$. \par

\begin{thm} Assume that $\tau$ is an $\mathbb{F}$-honest time, the MRP holds for a $d$-dimensional $\mathbb{F}$-adapted cadlag process $X$ and $\mathbb{F}$ is a Brownian filtration. Then there exists a bounded $\mathbb{G}$-martingale $Y$ such that the MRP holds for $(\widetilde{X}, M, Y)$ in $\mathbb{G}.$ 
\end{thm}
\begin{proof} We provide a sketch of the main arguments and refer to the original paper for the missing details and appropriate context. \par
Since $\mathbb{F}$ is a Brownian filtration and $\tau$ is an honest time, then there exists a set $A\in\mathcal{G}_{\tau}$ such that $\mathcal{G}_{\tau} = \mathcal{G}_{\tau-} \lor \sigma(A)$, and any $\xi \in \mathcal{G}_{\tau}$ has the representation $\xi = \xi^{\prime}\mathbbm{1}_A + \xi^{\prime\prime}\mathbbm{1}_{A^c}$. If $0 = \xi^{\prime}p + \xi^{\prime\prime}(1-p)$, where $p = \mathbb{P}(A \vert \mathcal{G}_{\tau-})$, then 
\begin{align*}
    \xi = \left( -\mathbbm{1}_{\{p>0\}}\xi^{\prime\prime}\frac{1}{p} + \mathbbm{1}_{\{p=0\}}\xi^{\prime}\frac{1}{1-p} \right)((1-p)\mathbbm{1}_A - p\mathbbm{1}_{A^c}).
\end{align*}
Define $Y = ((1-p)\mathbbm{1}_A - p\mathbbm{1}_{A^c})H$, where $H \coloneqq \mathbbm{1}_{\llbracket\tau,\infty\llbracket}$, and let $F$ be an $\mathbb{F}$-predictable proccess such that $F_{\tau} = \left( -\mathbbm{1}_{\{p>0\}}\xi^{\prime\prime}\frac{1}{p} + \mathbbm{1}_{\{p=0\}}\xi^{\prime}\frac{1}{1-p} \right)$. It now follows that $Y$ is a bounded $\mathbb{G}$-martingale, $F$ is $Y$-integrable and the following decomposition formula for any bounded $\mathbb{G}$-martingale $Z$ holds
\begin{align*}
    Z_t = Z_0 + \int_0^t(J^{\prime}\mathbbm{1}_{\{s\leq\tau\}} + J^{\prime\prime}\mathbbm{1}_{\{\tau<s\}})d\widetilde{X}_s + \int_0^tK_s\mathbbm{1}_{\{0<s\leq\tau\}}dM_s + \int_0^tF_sdY_s,
\end{align*}
where we have used the fact that for any $\mathbb{G}$-predictable process $J$ there exist $\mathbb{F}$-predictable processes $J^{\prime}$ and $J^{\prime\prime}$, such that $J\mathbbm{1}_{(0,\infty)} = J^{\prime}\mathbbm{1}_{(0,\tau\rrbracket} + J^{\prime\prime}\mathbbm{1}_{\rrbracket\tau,\infty)}$.
This shows that the MRP holds for $(\widetilde{X}, M, Y)$ in $\mathbb{G}.$ 
\end{proof}
\par

\bigskip

\par
We now turn our to attention once more to general random times. Consider a stochastic basis $(\Omega, \mathcal{F}, \mathbb{F}, \mathbb{P})$ and an $\mathcal{F}$-measurable random time $\tau$. Take $\mathbb{G}$ to be an enlargement of $\mathbb{F}$ by the random time $\tau$. Recall that the process $L_t = \mathbbm{1}_{\{\tau > 0\}}H_t - \int_0^{t\land\tau}\frac{dA_s}{Z_{s-}}$ is a $\mathbb{G}$-local martingale, where $H \coloneqq \mathbbm{1}_{\llbracket \tau, \infty \llbracket}$. Denote by $\mathcal{I}(\mathbb{P}, \mathbb{F}, M)$ the family of multi-dimensional $\mathbb{F}$-predictable processes that are $M$-integrable under $\mathbb{P}$. We also use the following definition of the $\mathcal{H}^{\prime}$-hypothesis: there exists a map $\Gamma$ from $\mathcal{M}_{loc}(\mathbb{F})$ into the space of cadlag $\mathbb{G}$-predictable processes with finite variation, such that for any $X \in \mathcal{M}_{loc}(\mathbb{F})$, $\Gamma(X)_0 = 0$ and $\widetilde{X} \coloneqq X - \Gamma(X)$ is a $\mathbb{G}$-local martingale. The operator $\Gamma$ is called the \textbf{drift operator}.  We assume throughout that (i) the MRP holds for a $d$-dimensional $\mathbb{F}$-adapted cadlag process $X$ and (ii) $\mathcal{H}^{\prime}$ is satisfied for $\mathbb{G}$. Then the following general theorem holds (see Jeanblanc, Song, 2015, p. 8, Theorem 3.1). \par

\begin{thm} For any bounded $\mathcal{G}_{\tau}$-measurable random variable $\zeta$, there exists an $\mathbb{F}$-predictable process $J$ such that $J\mathbbm{1}_{\llbracket 0, \tau \rrbracket} \in \mathcal{I}(\mathbb{P}, \mathbb{G}, \widetilde{X})$ and 
\begin{align*}
    \mathbb{E}_{\mathbb{P}}[\zeta \vert \mathcal{G}_t] = \mathbb{E}_{\mathbb{P}}[\zeta \vert \mathcal{G}_0] + \int_0^t(1 - H_{s-})J_sd\widetilde{X}_s + \mathbbm{1}_{\{\tau > 0 \}}(\zeta - Y_{\tau})H_t - \int_0^tK_s(1 - H_{s-})\frac{1}{Z_{s-}}dA_s,
\end{align*}
for $t \geq 0$, where $K$ is a bounded $\mathbb{F}$-predictable process such that $K_{\tau}\mathbbm{1}_{\{0 < \tau < \infty\}} = \mathbb{E}_{\mathbb{P}}[(\zeta - Y_{\tau})\mathbbm{1}_{\{0 < \tau < \infty\}} \vert \mathcal{F}_{\tau-}]$, and $Y_t \coloneqq \frac{\mathbb{E}_{\mathbb{P}}[\zeta\mathbbm{1}_{\{t<\tau} \vert \mathcal{F}_t]}{Z_t}\mathbbm{1}_{\{t<R\}}$, for $0 \leq t < \infty$. 
\end{thm}

\par
If we further suppose that $\{0 < \tau  < \infty\} \cap G_{\tau} = \{0 < \tau  < \infty\} \cap G_{\tau-}$ under $\mathbb{P}$, then it follows easily that the MRP holds in the stopped filtration $\mathbb{G}^{\tau}$ for the stopped process $(\widetilde{X}, L)$. More specifically, the following is true. \par

\begin{thm} Under the assumptions above, MRP$(\mathbb{P}, \mathbb{G}^{\tau}, (\widetilde{X}^{\tau},L))$ and $\mathbbm{1}_{\{0 < \tau < \infty\}}X_{\tau}$ is $\mathcal{G}_{\tau-}$-measurable if and only if $\{0 < \tau  < \infty\} \cap G_{\tau} = \{0 < \tau  < \infty\} \cap G_{\tau-}$ under $\mathbb{P}$. 
\end{thm}

\par
The preservation of the MRP beyond the default time $\tau$ is a subtler problem that we choose not to cover. We end the discussion of the preservation of the MRP under progressive enlargement with a particularly vivid example from (Jeanblanc, Song, 2011, 2015). \par

\begin{example} (An evolution model.) Consider a stochastic basis $(\Omega, \mathcal{F}, \mathbb{F}, \mathbb{P})$, where $\mathbb{F}$ satisfies the usual hypotheses. Consider an $\mathbb{F}$-adapted continuous increasing process $\Lambda$ and a non-negative $\mathbb{F}$-local martingale $N$, with $\Lambda_0 = 0$, $N_0 = 1$, and $0 \leq N_t e^{-\Lambda_t} \leq 1$ for all $0\leq t<\infty$. Define a product space given by $([0,\infty]\times\Omega,\mathscr{B}([0,\infty])\otimes\mathcal{F}_{\infty}),$ and the corresponding canonical projection maps $\pi(s,\omega) = \omega$ and $\tau(s,\omega) = s.$ We can now pull back the filtration onto the product space via $\widehat{\mathbb{F}} = \pi^{-1}(\mathbb{F})$ and the probability measure on $\pi^{-1}(\mathcal{F}_{\infty})$ via $\widehat{\mathbb{P}}(\pi^{-1}(A)) = \mathbb{P}(A)$ for all $A \in \mathcal{F}_{\infty}.$ We thus obtain a probability structure $([0,\infty]\times\Omega, \widehat{\mathbb{F}}, \widehat{\mathbb{P}})$ that is isomorphic to $(\Omega, \mathbb{F}, \mathbb{P}).$ In the sequel, $(\widehat{\mathbb{F}}, \widehat{\mathbb{P}})$ is denoted simply as $(\mathbb{F}, \mathbb{P}),$ and an $\mathcal{F}_{\infty}$-measurable random variable $\xi$ on $\Omega$ is identified with $\xi \circ \pi$ on the product sample space $[0,\infty]\times\Omega$. Within this setup, an interesting problem is posed.\par
\textit{Problem $\mathcal{P}^*$} Construct a probability measure $\mathbb{Q}$ on $([0,\infty]\times\Omega,\mathscr{B}([0,\infty])\otimes\mathcal{F}_{\infty})$ such that \par
(i) $\mathbb{Q} \vert_{\mathcal{F}_{\infty}} = \mathbb{P} \vert_{\mathcal{F}_{\infty}}$ (the restriction condition); \par
(ii) $\mathbb{Q}[\tau > t \vert \mathcal{F}_t] = N_te^{-\Lambda_t}$ for all $0\leq t < \infty$ (the projection condition). \par
It can be shown that under certain assumption there exist infinitely many solutions to the Problem $\mathcal{P}^*.$ For example, for any $(\mathbb{P}, \mathbb{F})$-local martingale $Y$, and any bounded differentiable function $f$ with bounded continuous derivative and $f(0) = 0$, there exists a solution $\mathbb{Q}^{\#}$ such that for any $u \in \mathbb{R}_+^*$ the martingale $M^u_t = \mathbb{Q}^{\#}\left[\tau \leq u \vert \mathcal{F}_t\right]$, $t\geq u$, satisfies the following evolution equation \par
\[
\begin{cases}
dX_t = X_t\left(-\frac{e^{-\Lambda_t}}{1 - Z_t}dN_t + f(X_t - (1-Z_t))dY_t\right), \quad u \leq t < \infty, \\
X_u = 1 - Z_u,
\end{cases}
\]
\end{example}
\textcolor{black}{where $Z$ is the Azema supermartingale associated with $\t$ with the multiplicative decomposition $Z=Ne^{-\L}$}.
\par
We now consider an enlarged by a random time $\tau$ filtration $\mathbb{G}$ on the product space given by $\mathcal{G}_t = \mathcal{F}_t \lor \sigma(\tau \land t)$. Then the following theorems hold. \par

\begin{thm} Assume the above conditions hold, and for each $0 < t < \infty$, the map $u \mapsto M^u_t$ is continuous on $(0,t].$ Then, for any $(\mathbb{P}, \mathbb{F})$-local martingale $X$, the process 
\begin{align*}
    \Gamma(X)_t &\coloneqq \int_0^t\mathbbm{1}_{\{s\leq\tau\}}\frac{e^{-\Lambda_s}}{Z_s}d\langle N,X \rangle_s - \int_0^t\mathbbm{1}_{\{\tau < s\}}\frac{e^{-\Lambda_s}}{1 - Z_s}d\langle N,X \rangle_s \\ 
    &+ \int_0^t\mathbbm{1}_{\{\tau<s\}}(f(M^{\tau}_s - (1-Z_s)) + M^{\tau}_s f^{\prime}(M^{\tau}_s - (1-Z_s)))d\langle Y,X \rangle_s, \quad 0\leq t <\infty,
\end{align*}
is a well-defined $\mathbb{G}$-predictable process with finite variation, and the process $\widetilde{X} = X - \Gamma(X)$ is a $(\mathbb{Q}^{\#}, \mathbb{G})$-local martingale. \end{thm}

\par
\begin{thm} \textcolor{black}{Let $X$ an $\mF$-adapted cadlag process.} Assume the following conditions hold: \par
(i) all $(\mathbb{P}, \mathbb{F})$-local martingales are continuous; \par
(ii) for each $0 < t < \infty$, the map $u \mapsto M^u_t$ is continuous on $(0,t]$; \par
(iii) $X$ possesses the MRP in $\mathbb{F}$ under $\mathbb{P}$; \par
(iv) $0 < Z_t < 1$ for $0 < t < \infty$. \par
Then $(\widetilde{X}, L)$ possesses the MRP in $\mathbb{G}$ under $\mathbb{Q}^{\#}$, where $L$ is given by 
\begin{align*}
    L_t = \mathbbm{1}_{\{\tau >0\}}H_t - \int_0^{\tau \land t}\frac{dA_s}{Z_{s-}}, \quad t\geq0,
\end{align*}
\textcolor{black}{where $A$ is the predictable increasing process in the Doob-Meyer decomposition of $Z$ and $H\coloneqq \1_{\llbracket \t,\infty)}.$}
\end{thm}

\bigskip

\section{Levy Processes, Processes with Independent Increments and Point Processes} 

In this section we focus on a number of results that can be considered as either of more specialized nature (compared to general semimartingale results) or as generalizations of some of the examples we have encountered before (e.g. enlargement in a Brownian motion setting or the Kusuoka model and its extensions by Calzolari and Torti, to name a few). We have decided to group these results together in a separate section since they revolve around the ideas of filtration enlargement in the context of Levy processes, processes with independent increments and point processes that share and employ similar mathematical techniques. As we have done before, we consider the preservation of the semimartingale property and the predictable representation property in separate subsections. \par

\bigskip

\subsection{$\cH \ \text{and} \ \cH^{\prime}$ hypotheses under enlargement}

One of the first general results in this area was established in (Jacod, Protter, 1988), which we follow in the sequel. We begin by setting up the notation and definitions first. Let $Z$ denote a Levy process on $[0,1]$, i.e. a process with stationary and independent increments, cadlag paths, with $Z_0 = 0$ a.s. $Z^c$ will denote the local martingale part of $Z$ with respect to its natural filtration $\mF$. Thus, excluding the degenerate case of $Z^c \equiv 0$, it follows that the normalized process $Z^c/\s$ is a standard Wiener process for some $\s>0.$ Let $\m$ be the jump measure of $Z$ given by 
\begin{align*}
    \m(\w; dt \times dx) = \sum_{s>0, \D Z_s(\w)\neq 0}\e_{(s,\D Z_s(\w))}(dt \times dx),
\end{align*}
\textcolor{black}{where $\e_a$ denotes the Dirac measure at $a$.}
\par

The $\mF$-compensator of the jump measure is given by 
\begin{align*}
    \n(\w; dt \times dx) = dt \otimes F(dx) 
\end{align*}
for some non-random measure $F$ on $\mR$ such that\textcolor{black}{ $\int_{\mR}\min(x^2,1)F(dx)<\infty$}. Thus, for every $\a>0$ we have the standard decomposition for $Z$
\begin{align*}
    Z_t = b_{\a}t + Z^c_t + \int_0^t\int_{|x|\leq\a}x(\m - \n)(ds \times dx) + \sum_{0<s<t}\D Z_s\1_{\{|\D Z_s|>\a\}},
\end{align*}
where $b_{\a}\in\mR$. \par
We define the filtration $\mG$ to be the smallest filtration that satisifes the usual conditions relative to which $Z$ is adapted and $Z_1$ is $\cG_0$-measurable. The filtration $\mH$ is defined to be the smallest filtration that satisfies the usual conditions, with respect to which $Z$ is adapted, $Z_1^c$ and $\sum_{0 < s \leq 1}\1_A(\D Z_s)$ are $\cH_0$-measurable for all Borel sets $A$ at a positive distance away from $0.$ Note that $\mH$ is larger than $\mG$ since $\int f(x) \m((0,1] \times dx) = \sum_{0<s\leq1}f(\D Z_s)$ is $\cH_0$-measureable for all Borel functions $f$ that vanish on a neighborhood of $0$ and the integral representation of $Z$ above is the $L^2$-limit of 
\begin{align*}
    \int_0^t\int_{1/n<|x|\leq\a}x(\m-\n)(ds \times dx) = \sum_
    {0<s\leq t}\D Z_s\1_{\{1/n<|\D Z_s| \leq \a\}} - t \int_{1/n<|x|\leq\a}xF(dx)
\end{align*}
as $n \to \infty$, hence $Z_1$ is $\cH_0$-measurable. \par
The main result of Jacod and Protter concerning enlargement is the following. \par

\begin{thm} Let $Z$ be a Levy process. Then every $\mF$-semimartingale is an $\mH$-semimartingale on $[0,1)$. 
\end{thm}

\par
Before proving the result, we focus on a result that generalizes the Brownian motion case we have met before. While rather standard and well-known in the modern literature, it captures the key aspects of some of the essential ideas involved in initial enlargement by the terminal value - an idea that is in some sense foundational to the notion of insider information. Moreover, this example presents a useful exercise in obtaining "tangible derivations" that shed light on the more \textcolor{black}{abstract} (and often cluttered in terms of notation and technical conditions that usually accompany greater generality) theoretical results, and thus has a lot of instructive value. It is due to E. Kurtz, we follow the exposition in (Jacod, Protter, 1988, p.624).  \par

\begin{thm} Assume that the Levy process $Z$ is integrable, i.e. $\mE |Z_t| < \infty$ for all $t$. Then 
\begin{align*}
    M_t = Z_t - \int_0^t\frac{Z_1-Z_s}{1-s}ds 
\end{align*}
is a $\mG$-martingale on $[0,1)$. 
\end{thm}
\begin{proof} We begin by assuming a stronger integrability condition and require that $\mE Z_t^2 < \infty$ for all $t.$ Next define
\begin{align*}
    Y_i = Z_{(i+1)/n} - Z_{i/n}.
\end{align*}
Then we note that for rational $0\leq s<t\leq 1, s=j/n, t=k/n,$ we have 
\begin{align*}
    Z_1 - Z_s &= \sum_{i=j}^{n-1}Y_i \\
    Z_t - Z_s &= \sum_{i=j}^{k-1}Y_i.
\end{align*}
Since $Z$ is a Levy process, the random variables $Y_i$ are i.i.d., and, moreover, integrable. Note
\begin{align*}
    \mE(Z_t-Z_s | Z_1 - Z_s) &= \mE \left(\sum_{i=j}^{k-1}Y_i \Bigg| \sum_{i=j}^{n-1}Y_i \right) \\
    &= \frac{k-j}{n-j}\sum_{i=j}^{n-1}Y_i \\
    &= \frac{t-s}{1-s}(Z_1-Z_s).
\end{align*}
By the independent increments property, $\mE(Z_t-Z_s|\cG_s) = \mE(Z_t-Z_s|Z_1-Z_s)$, and hence $\mE(Z_t-Z_s|\cG_s) = \frac{t-s}{1-s}(Z_1-Z_s)$ for all rational $0\leq s < t \leq 1$. Since $Z_t - \mE Z_t$ is an $\mF$-martingale, the family $(Z_t)_{0\leq t \leq 1}$ is uniformly integrable and $Z$ has cadlag paths, hence the equality holds for all real $0 \leq s<t \leq 1$. We now proceed to check the martingale properties for $M$ as defined in the statement of the theorem. For fixed $0\leq s<t <1$, the conditional Fubini theorem gives 
\begin{align*}
    \mE(M_t-M_s | \cG_s) &= \mE (Z_t-Z_s|\cG_s) - \int_s^t\frac{1}{1-u}\mE (Z_1-Z_u|\cG_s)du \\
    &=\frac{t-s}{1-s}(Z_1-Z_s) - \int_s^t\frac{1}{1-u}\frac{1-u}{1-s}(Z_1-Z_s)du \\
    &= 0.
\end{align*}
The indepedent and stationary increments property gives 
\begin{align*}
    \mE |Z_1-Z_s| \leq \mE \left[(Z_1-Z_s)^2\right]^{1/2}\leq \a (1-s)^{1/2}
\end{align*}
for some $\a$ and $0\leq s \leq 1$, which gives the bound
\begin{align*}
    \mE \int_0^1\frac{|Z_1-Z_s|}{1-s}ds \leq \a \int_0^1\frac{\sqrt{1-s}}{1-s}ds < \infty.
\end{align*}
We now need to relax the integrability condition to $\mE |Z_t|<\infty$ for all $0\leq t \leq 1.$ As a first step, we subtract away the big jumps of the process. Set 
\begin{align*}
    J_t^1 &= \sum_{0<s\leq t} \D Z_s \1_{\{\D Z_s>1\}}, \\
    J_t^2 &= -\sum_{0<s\leq t} \D Z_s \1_{\{\D Z_s<-1\}},
\end{align*}
and note that $Y = Z - J^1 + J^2$ is a Levy process with bounded jumps and hence square integrable. The independence of $Y, J^1$ and $J^2$ implies that $Y_t - \int_0^t\frac{Y_1-Y_s}{1-s}ds$ is a $\mG$-martingale on $[0,1]$. Similarly, it can be shown that if $\mE \int_0^1\frac{|J_1^i-J^i_s|}{1-s}ds<\infty$, then $J^i_t - \int_0^t\frac{J^i_1-J^i_s}{1-s}ds$ is a martingale on $[0,1]$. To obtain the necessary bound, note that by the independent and stationary increments we have 
\begin{align*}
    \mE \int_0^1\frac{|J^i_1-J^i_s|}{1-s}ds &= \mE \int_0^1\frac{J^i_1-J^i_s}{1-s}ds \\
    &=\int_0^1\mE \frac{J^i_1-J^i_s}{1-s}ds \\
    &=\a_i\int_0^1\frac{1-s}{1-s}ds = \a_i<\infty,
\end{align*}
which concludes that $M$ is a $\mG$-martingale on $[0,1]$. 
\end{proof}

\par
We note that a key step in the proof is checking the integrability conditions, which in general is not a trivial matter.
\par

The following result is also of interest. \par

\begin{thm} Let $Z$ be a Levy process. \par
(i) The process
\begin{align*}
    \wh{Z}^c_t = Z^c_t -\int_0^t\frac{Z^c_1 - Z^c_s}{1-s}ds
\end{align*}
is an $\mH$-martingale on $[0,1]$ with quadratic variation $\langle\wh{Z}^c,\wh{Z}^c\rangle = \langle Z^c, Z^c \rangle .$ \par
(ii) the $\mH$-compensator $\r$ of the jump measure $\m$ on $[0,1]$ is given by 
\begin{align*}
    \r(\w; dt \times dx) = dt \times \frac{\m(\w;(t,1]\times dx)}{1-t}\1_{[0,1)}(t).
\end{align*}
\end{thm}
\begin{proof}
    See (Jacod, Protter, 1988, Theorem 2.9).
\end{proof} 

We now proceed to the proof of the main result of Jacod and Protter. \par
\begin{proof} \textit{(the Main Thoerem)}. It suffices to show hat any square-integrable $\mF$-martingale $M$ on $[0,1)$ (either continuous or purely discontinuous) with $M_0=0$ is an $\mH$-semimartingale on $[0,1)$. The two cases are treated separately. \par
(i) Let $M$ be a continuous square-integrable $\mF$-martingale on $[0,1]$ with $M_0=0$. By the martingale representation theorem, there exists a predictable process $H$ such that 
\begin{align*}
    \mE\int_0^1H^2_sd\langle Z^c,Z^c\rangle_s = \s^2\mE \int_0^1H^2_sds<\infty,
\end{align*}
for which 
\begin{align*}
    M_t = \int_0^tH_sdZ^c_s,
\end{align*}
where we have used the fact that $\langle Z^c,Z^c\rangle
_t = \s^2t, \ \s\geq0$. \par
    By the result stated above, $\langle\wh{Z}^c,\wh{Z}^c\rangle = \langle Z^c, Z^c \rangle,$ hence $\wh{M}_t = \int_0^t H_s d\wh{Z}^c_s$ is a well-defined stochastic integral and an $\mH$-martingale on $[0,1).$ Setting $C = Z - \wh{Z}^c$, we can also define the Stieltjes integral $D_t = \int_0^tH_sdC_s$ on $[0,1),$ and obtain the decomposition $M = \wh{M} + D.$ This decomposition is explicit in the case of bounded $H$, which follows from Stricker's theorem, which guarantees that $\int_0^1 H_sdZ^c_s$ takes the same value in $\mF$ and $\mH$. The unbounded case is treated by truncation, i.e. setting $H^n \coloneqq H\1_{\{|H|\leq n\}}$, $M^n_t \coloneqq \int_0^tH^n_sdZ^c_s, \wh{M}^n_t \coloneqq \int_0^t H^n_sd\wh{Z}^c_s,$ and $D^n_t \coloneqq \int_0^tH^n_sdC_s.$ We thus obtain the decomposition $M^n = \wh{M}^n + D^n,$ where $M^n, \wh{M}^n$ and $D^n$ converge in probability to $M, \wh{M}$ and $D$ respectively. \par
    (ii) Now let $M$ be a purely discontinuous square-integrable $\mF$-martingale on $[0,1)$ with $M_0=0.$ By the martingale representation theorem, there exists a predictable process $W$ on $\W \times [0,1) \times \mR$ such that 
    \begin{align*}
        M_t = \int_0^t\int_{\mR}W(s,x)(\m - \n)(ds \times dx),
    \end{align*}
    and 
    \begin{align*}
        \mE \int_0^1\int_{\mR}W(s,x)^2\m(ds\times dx) = \mE \int_0^1\int_{\mR}W(s,x)^2dsF(dx)<\infty,
    \end{align*}
    where $\m,\n$ and $F$ are as defined at the beginning of the section. Moreover, the following process is well defined and is an $\mH$-martingale on $[0,1)$
    \begin{align*}
        \wh{M}_t = \int_0^t\int_{\mR}W(s,x)(\m - \r)(ds\times dx),
    \end{align*}
    where $\r$ is as defined in the preceding theorem. We now define a sequence of stochastic processes given by the following
    \begin{align*}
         M_t^n &= \int_0^t\int_{|x|>1/n}W(s,x)(\m - \n)(ds\times dx), \\
         \wh{M}_t^n &= \int_0^t\int_{|x|>1/n}W(s,x)(\m - \r)(ds\times dx), \\
         C_t^n &= \int_0^t\int_{|x|>1/n}W(s,x)(\r - \n)(ds\times dx),
    \end{align*}
    with $M^n_0=\wh{M}^n_0=C^n_0=0$ These processes are of finite variation on $[0,t]$ for any $t<1.$ Note the decomposition $M^n = \wh{M}^n + C^n,$ where $M_n$ is an $\mF$-martingale on $[0,1]$ and $\wh{M}^n$ is an $\mH$-martingale on $[0,1).$ By the standard results from random measure theory, $M^n_t$ and $\wh{M}^n_t$ converge in $L^2$ as $n\to\infty$ to $M_t$ and $\wh{M}_t$ respectively, hence $C^n_t$ converges in $L^2$ to $C_t =M_t - \wh{M}_t.$ It remains to show that $C$ has paths of finite variation on $[0,t]$ for any $t<1$. Note that 
    \begin{align*}
        C^n_t = \int_0^tU_s^nds,
    \end{align*}
    where $U$ is given by
    \begin{align*}
        U^n_s(\w) \coloneqq \int_{|x|>1/n}\frac{1}{1-s}\m(\w;(s,1]\times dx)W(\w;s,x) - \int_{|x|>1/n}W(\w;s,x)F(dx)
    \end{align*}
    Letting $n>m\geq0$ and using the convention $1/0 = +\infty$, we get 
    \begin{align*}
        N_t^{n,m,s}(\w) = \frac{1}{1-s}\1_{\{s<t\}}\int_s^t\int_{1/n<|x|\leq1/m}W(\w;u,x)(\m-\n)(\w; du\times dx)
    \end{align*}
    Since $N_t^{n,m,s}(\w)$ is the stochastic integral with respect to $\m-\n$ of the function 
    \begin{align*}
        (\w, u,x) \mapsto W^{n,m,s}(\w;u,x) = \frac{1}{1-s}\1_{\{s<u\}}W(\w;u,x)\1_{\{1/n<|x|\leq1/m\}},
    \end{align*}
    hence $N^{n,m,s}$ is an $\mF$-martingale, and 
    \begin{align*}
        \mE(N_1^{n,m,s})^2 &= \mE \int_0^1\int_{\mR}(W^{n,m,s}(u,x))^2du F(dx) \\
        &= \frac{1}{1-s}\mE \int_{1/n<|x|\leq1/m} W(s,x)^2F(dx).
    \end{align*}
    Since $N_1^{n,m,s} = U^n_s - U^m_s$ by construction, we have for $t<1$ the expression
    \begin{align*}
        \mE\int_0^t(U^n_s-U^m_s)^2ds = \int_0^t\frac{ds}{1-s}\mE \int_{1/n<|x|\leq 1/m}W(s,x)^2F(dx),
    \end{align*}
\noindent
tends to $0$ as $n,m \to \infty$ since $\mE \int_0^1\int_{\mR}W(s,x)^2\m(ds\times dx) = \mE \int_0^1\int_{\mR}W(s,x)^2dsF(dx)<\infty$. This implies that $U^n$ converges to a limit $U$ in $L^2(\W\times[0,t],\mP(d\w)\otimes du)$ and $C_t = \int_0^tU_sds$, from which the claim follows. 
\end{proof}

\par

\bigskip

We now discuss a generalization of the above theorem to semimartingales with independent increments obtained by L. Gal'chuk. We will focus on the main theorem and its proof (which follows the same pattern as that of of Jacod and Protter) and state the auxiliary lemmas without proofs, which can be found in (Gal'chuk, 1993). \par
Let $(\W,\cF,\mF,\mP)$ be a stochastic basis, where $\mF$ satisfies the usual conditions. Let $X$ be a stochastic process with independent increments on $[0,T]$, $T<\infty$ adapted to $\mF$, which we take to be the natural filtration generated by $X.$ We also assume $X$ to be an $\mF$-semimartingale, and thus has the Levy-Khintchine representation 
\begin{align*}
    X_t = v_t + X^c_t + \int_{|x|\leq 1}x(p-\p)((0,t],dx) + \int_{|x|>1}xp((0,t],dx),
\end{align*}
where $v$ is a continuous deterministic process of finite variation, $X^c$ is a continuous Gaussian martingale with independent increments, $p((0,t],A)$ is a Poisson random measure defined by 
\begin{align*}
    p(\w, (0,t], A) = \sum_{0<s\leq t}\1_{(\w,s:\D X_s\in A)}(\w,s), \ \ A\in\cB_0,
\end{align*}
$\p = \mE p$ is the $\mF$-compensator of $p$, and $\cB_0$ is the collection of sets from $\cB(\mR)$ lying away from $0$. The enlarged filtration $\mG$ is defined via
\begin{align*}
    \cG_t \coloneqq \cF_t \lor \s\{X^c_T, p((0,T],A), A\in\cB_0\}.
\end{align*}
We introduce the following disintegration formula for the compensator
\begin{align*}
    \p(dt,dx)= \P(t,dx)da_t,
\end{align*}
where $a$ is some continuous increasing function and $\P(t,dx)$ is a kernel from $[0,T]$ to $\mR$. We will also use the notation $\p^s(dx) \coloneqq \p((s,T],dx)$. Moreover, we impose an additional condition on $X$.\par
\textbf{(C)} For any fixed $s\in[0,T]$, the measure $\p$ is absolutely continuous with respect to $\p^s\otimes a$ with the density 
\begin{align*}
    \G(s,t,x) = \frac{d\p}{d(\p^s\otimes a)}, \ \ s\leq t \leq T.
\end{align*}
    Note that the density $\G(s,t,x)$ is measurable in $(t,x)$ and continuous in $s$, hence measurable in $(s,t,x).$ \par
Note that every $\mF$-semimartingale $Y$ has the representation
\begin{align*}
    Y_t = V_t + \int_0^tf_sdX^c_s + \int_0^t\int_{\mR}g(s,x)(p-\p)(ds,dx),
\end{align*}
where $V$ is a finite variation process, $f$ and $g$ are $\mF$-predictable, and satisfy
\begin{align*}
    \int_0^Tf^2_sd\langle X^c\rangle_s<\infty, \quad \int_0^T\int_{\mR}g^2(s,x)\p(ds,dx)<\infty \quad a.s.
\end{align*}
    If $Y$ is an $\mF$-semimartingale of finite variation on $[0,T]$, then it is also a $\mG$-semimartingale of finite varitation on $[0,T]$. We are now ready to state the main result of Gal'chuk. \par

\begin{thm} Let $X = (X_t, t\in [0,T])$ be a stochastic process with independent increments that is continuous in probability and satisfies the condition $(C)$. \par
(i) if $Y$ is a continuous $\mF$-semimartingale on $[0,T]$, then it is a $\mG$-semimartingale on $[0,\t)$, where $\t = \inf \{s : b_s = b_T\}$, where $b = \langle X^c\rangle.$ \par
$Y$ is a $\mG$-semimartingale on $[0,\t]$ if and only if
\begin{align*}
    \int_0^{\t}\left| \frac{d\langle Y^c,X^c\rangle}{db} (s) \right| \frac{db_s}{\sqrt{b_T-b_s}}<\infty \ \ a.s.
\end{align*}
\par
(ii) If $Y$ is a purely discontinuous $\mF$-martingale with the representation
\begin{align*}
    Y_t = \int_0^t\int_{\mR}g(s,x)(p-\p)(ds,dx),
\end{align*}
and if 
\begin{align*}
    \int_0^T\int_{\mR}g^2(s,x)\G(s,s,x)\p(ds,dx)<\infty \ \ a.s.,
\end{align*}
then $Y$ is a $\mG$-semimartingale on $[0,T].$
\end{thm}

\par

Note that if $Y$ is a $\mG$-semimartingale on $[0,\t]$, then it can be extended to the whole interval $[0,T]$ by setting $Y_t = Y_{\t}$ for $t\in(\t,T].$ \par
To prove the theorem we need a number of auxiliary lemmas (some of which are interesting in their own right), as well as the theorem of Kurtz that we have encountered before. \par

\begin{lem} Let $(p_t)$ be a Poisson process with the $\mF$-compensator $(\p_t), t\in[0,T].$ Then, for $s\leq t \leq \d = \inf\{s : \p_s = \p_T\}$ it holds
\begin{align*}
    \mE (p_t-p_s|\cG_s) = \frac{p_T-p_s}{\p_T-\p_s}(\p_t-\p_s) \ \ a.s.
\end{align*}
\end{lem}

\par
Under the condition $(C)$, for $A\in\cB_0, s\in[0,T]$, set
\begin{align*}
    &\d_A \coloneqq \inf\left(t>0: \p((0,t],A)=\p((0,T],A)\right), \\
    &\p^s(A) \coloneqq \p((s,\d_A],A) = \p((s,T],A), \\
    &p^s(A) \coloneqq p((s,T],A).
\end{align*}
Let $\m^s(\D,A)$ be the regular version of the conditional expectation $\mE [p(\D,A)|\cG_s]$, $\D = (u,v]\subseteq (s,\d_A].$ Then, by the lemma above, 
\begin{align*}
    \m^s(\D,A)  = \begin{cases}
    \frac{p^s(A)}{\p^s(A)}\p(\D,A), \quad \text{if} \quad \p^s(A)\neq0, \\
    0 \quad \text{otherwise}.
    \end{cases}
\end{align*}
    For $\D\subseteq(\d_A,T]$ we \textcolor{blue}{set} $0/0=0.$ Finally, for any bounded positive $\cB([0,T])$-measurable function $f$ define 
\begin{align*}
    \m^s(f,A) &\coloneqq \int_s^T\frac{p^s(A)}{\p^s(A)}\int_Af(v)\G(s,v, x)\p^s(dx)da_v, \\
    \wh{\m}(f,A) &\coloneqq \int_s^T\int_Af(v)\G(s,v,x)p^s(dx)da_v. 
\end{align*}
\par

\begin{lem} For $s\in[0,T], A\in\cB_0, p^s(A)<\infty$ a.s. and any positive bounded $\cB([0,T])$-measurable function $f$ it holds that
\begin{align*}
    \m^s(f,A) = \wh{\m}^s(f,A) \quad \text{a.s.}
\end{align*}
\end{lem}

\par
\begin{lem} For $A\in\cB_0$, the process $(\wh{p}((0,t],A))$ given by 
\begin{align*}
    \wh{p}((0,t],A) \coloneqq p((0,t],A) - \int_0^t\int_A\G(u,u,x)p((0,t],dx)da_u
\end{align*}
is a $\mG$-martingale on $[0,T].$
\end{lem}

\par
\begin{lem} The measure given by
\begin{align*}
    \m(u,dx) = \G(u,u,x)p((0,T],dx)da_u
\end{align*}
is the $\mG$-compensator of the measure $p(dt,dx).$
\end{lem}

\par

\bigskip

We are now ready to prove the theorem of Gal'chuk. Note the similarity of the steps and techniques of the proof to the theorem of Jacod and Protter. \par

\begin{proof} It suffices to consider the case of a square-integrable $\mF$-martingale $m=(m_t),m
_0=0$, which is known to have the representation
\begin{align*}
    m = m^c + m^d,
\end{align*}
where
\begin{align*}
    m^c_t &= \int_0^tf_sdX^c_s, \\
    m^d_t &= \int_0^t\int_{\mR}g(s,x)(p-\p)(ds,dx),
\end{align*}
where $f$ is some $\cP(\mF)$-measurable function and $g$ is a function measurable with respect to $\cP(\mF)\otimes\cB(\mR)$, where $\cP(\mF)$ is the $\mF$-predictable $\s$-algebra, and these functions satisfy 
\begin{align*}
    \mE \int_0^Tf^2_sd\langle X^c\rangle_s<\infty, \quad \mE \int_0^T\int_{\mR}g^2(s,x)\p(ds,dx)<\infty.
\end{align*}
\par
As before, we consider two cases.\par
(i) When $Y = m^c$, from the integrability condition $\mE \int_0^Tf^2_sd\langle X^c\rangle_s<\infty, t\in[0,T]$, it follows that $\wh{Y}_t = \int_0^Tf_sd\wt{X}^c_s$ is well-defined and is a $\mG$-martingale on $[0,T].$ By the lemma above, the process $C = X^c - \wt{X}^c$ is of finite variation, hence the Stieltjes integral $D_t = \int_0^tf_sdC_s$ is well-defined on $[0,\t).$ It remains to show that $Y = \wh{Y} + D$. By Stricker's theorem, when $f$ is bounded, the stochastic integral $\int_0^tf_sdX^c_s$ takes the same value in $\mF$ and $\mG$. In the unbounded case, we use truncation by setting
\begin{align*}
    f^n &\coloneqq f\1_{\{|f|\leq n\}}, \quad Y^n_t \coloneqq \int_0^t f^n_sdX^c_s, \\
    \wh{Y}^n_t &\coloneqq \int_0^tf^n_sd\wt{X}^c_s, \quad D^n_t \coloneqq \int_0^t f^n_sdC_s.
\end{align*}
We thus have the decomposition $Y^n = \wh{Y}^n + D^n$, and $Y^n, \wh{Y}_n, D^n$ converge in probability as $n\to\infty$ to $Y,\wh{Y}, D$ respectively, which proves the validity of the decomposition.\par
(ii) When $Y = m^d$, define a process by 
\begin{align*}
    \wh{Y}_t \coloneqq \int_0^t\int_{\mR}g(s,x)(p-\m)(ds,dx), \quad t\in[0,T].
\end{align*}
To show that the object is well defined and $\wh{Y}$ is a $\mG$-martingale, we need to check that 
\begin{align*}
    J_t = \mE \int_0^t\int_{\mR}g^2(s,x)\m(ds,dx)<\infty, \quad t\in[0,T].
\end{align*}
By the definition of the measure $\m$ and Fubini's theorem we have 
\begin{align*}
    J_t &= \mE \int_0^t\int_{\mR}g^2(s,x)\G(s,s,x)p((s,T], dx)da_s \\
    &= \mE \int_0^t\int_0^T\int_{\mR}g^2(s,x)\G(s,s,x)\1_{\{s<u\leq T\}}p(du,dx)da_s \\
    &= \mE \int_0^T\int_{\mR}\left[\int_0^tg^2(s,x)\G(s,s,x)\1_{\{s<u\leq T\}}da_s\right]p(du,dx).
\end{align*}
Since the function $\int_0^tg^2(s,x)\G(s,s,x)\1_{\{s<u\leq T\}}da_s$ is predictable and positive, and $\p$ is the $\mF$-compensator of the measure $p$, then 
\begin{align*}
    J_t = \mE \int_0^T\int_{\mR}\left[\int_0^tg^2(s,x)\G(s,s,x)\1_{\{s<u\leq T\}}da_s\right]\p(du,dx).
\end{align*}
By another application of Fubini's theorem, we get 
\begin{align*}
    J_t &= \mE \int_0^t\int_{\mR}g^2(s,x)\G(s,s,x)\p((s,T],dx)da_s \\
    &= \mE \int_0^t\int_{\mR}g^2(s,x)\p(ds,dx)<\infty.
\end{align*}
\par
As before, it now remains to show that the process $Y-\wh{Y}$ has finite variation. Let $n\in\mN$ and $A^n \coloneqq \{x : |x|>1/n\}.$ Similarly to the proof of Jacod and Protter, define
\begin{align*}
    Y^n_t &\coloneqq \int_0^t\int_{A^n}g(s,x)(p-\p)(ds,dx), \\
    \wh{Y}^n_t &\coloneqq \int_0^t\int_{A^n}g(s,x)(p-\m)(ds,dx), \\
    C^n_t &\coloneqq \int_0^t\int_{A^n}g(s,x)(\m-\p)(ds,dx), \quad t\in[0,T],\\
    Y^n_0 &= \wh{Y}^n_0=C^n_0=0.
\end{align*}
Note that these processes have finite variation, $Y^n$ is a square-integrable $\mF$-martingale, $\wh{Y}^n$ is a square-integrable $\mG$-martingale, and $Y^n_t, \wh{Y}^n_t, C^n_t$ converge in $L^2$ as $n\to\infty$ to $Y_t, \wh{Y}_t, C_t$ respectively, $t\leq T.$ To show that $C$ has finite variation, note 
\begin{align*}
    C^n_t = \int_0^t\int_{A^n}g(s,x)\G(s,s,x)\left[p((s,T],dx) - \p((s,T],dx)\right]da_s,
\end{align*}
and, for $n>m, A^{mn}\coloneqq \{1/n<|x|\leq 1/m\},$ we have
\begin{align*}
    C^n_t - C^m_t &= \int_0^t\int_{A^{mn}}g(s,x)\G(s,s,x)\left[p((s,T],dx) - \p((s,T],dx)\right]da_s \\
    &= \int_0^t\left[\int_s^T\int_{A^{mn}}g(s,x)\G(s,s,x)(p-\p)(du,dx)\right]da_s.
\end{align*}
Since for a fixed $s$ the process $\int_s^t\int_{A^{mn}}g(s,x)\G(s,s,x)(p-\p)(du,dx), s\leq t \leq T,$ is a square-integrable martingale, we can estimate the variation of $|C^n_t - C^m_t|$ on $[0,t]$ in the following way 
\begin{align*}
    \mE |C^n_t - C^m_t|^2 &\leq a_t \int_0^tda_s \mE\left|\int_s^T\int_{A^{mn}}g(s,x)\G(s,s,x)(p-\p)(du,dx)\right|^2 \\ 
    &\leq a_t \mE \int_0^tda_s \int_s^T\int_{A^{mn}}g^2(s,x)\G^2(s,s,x)\p(du,dx) \\
    &= a_t \mE \int_0^t\int_{A^{mn}}g^2(s,x)\G^2(s,s,x)\p((s,T],dx)da_s \\
    &= a_t \mE \int_0^t\int_{A^{mn}}g^2(s,x)\G(s,s,x)\p(ds,dx)<\infty,
\end{align*}
where we have used the fact that $\G(s,s,x)\p((s,T],dx)da_s = \p(ds,dx)$. It thus follows that on each interval $[0,t], t\leq T,$ the variation of $|C^n_t - C^m_t|$ converges in probability to $0$ as $n,m\to\infty,$ hence $C$ has finite variation on $[0,T].$ The last statement of part (i) of Gal'chuk's theorem follows from the following proposition which we do not prove and refer to the original paper for details. 
\end{proof}

\par
\begin{thm} Let $Y$ be a continuous $\mF$-local martingale with the integral representation 
\begin{align*}
    Y_t = \int_0^tf_sdX^c_s,
\end{align*}
where $f$ is an $\mF$-predictable process such that $\int_0^Tf^2_sdb_s<\infty$ a.s., for \textcolor{blue}{$b=\langle X^c\rangle$}. \par
Then, the following statements are equivalent:\\
(a) $Y$ is a $\mG$-semimartingale on $[0,\t];$ \\
(b) 
\begin{align*}
    \int_0^T|f_s|\frac{|X^c_T-X^c_s|}{b_T-b_s}db_s<\infty \quad a.s.
\end{align*} \\
(c)
\begin{align*}
    \int_0^t|f_s|(b_T-b_s)^{-1/2}db_s<\infty \quad a.s.,
\end{align*}
where $\t \coloneqq \inf\{s:b_s=b_T\}.$ Moreover, if these condition are satisfied, then $Y$ admits the following canonical decomposition in $\mG$
\begin{align*}
    Y_t = \int_0^tf_sd\wt{X}^c_s + \int_0^tf_s\frac{X^c_T-X^c_s}{b_T-b_s}db_s,
\end{align*}
\textcolor{black}{where $\wt{X}^c_t = X^c_t - \int_0^t\frac{X^c_T-X^c_u}{b_T-b_u}db_u.$}
\end{thm}

\bigskip

\subsection{Martingale Representation Theorems}

We begin by considering the preservation of the predictable representation property (PRP) in progressively enlarged Levy filtrations. This subsection is based on the work (Di Tella, Engelbert, 2020) that contains some of the most general results available at this point. \par
As usual, we work on a complete stochastic basis $(\W,\cF,\mF,\mP)$, where $\mF$ satisfies the usual conditions. $\cP(\mF)$ represents the $\s$-algebra of $\mF$-predictable subsets of $\mR_+\times\W$. Let $\t$ be a random time and define $\mG \coloneqq \mF \lor \mH$, where $\mH$ is the filtration generated by the default process $\1_{\lb\t,\infty)}$. We will assume throughout that $\t$ satisfies the following two key assumptions: \par

\bigskip

$(\cA) \ \t$ avoids $\mF$-stopping times, i.e. for every $\mF$-stopping time $\eta$, $\mP[\t=\eta<\infty] = 0$; \par
$(\cH)$ $\mF$ is immersed in $\mG$, i.e. every $\mF$-martingale is a $\mG$-martingale. \par

\bigskip

Let $L$ be a Levy process, i.e. an $\mF$-adapted stochastically continuous process with independent and stationary increments. Define $\mF^L$ to be the filtration generated by $L$ that satisfies the usual conditions. Denote by $\m$ the jump measure of $L$, i.e. a homogeneous Poisson random measure with respect to $\mF^L$. The predictable compensator of $\m$ is deterministic and denoted by $\l_+\otimes\n$, where $\l_+$ is the Lebesgue measure on $\mR_+$ and $\n$ is the Levy measure of $L$ such that $\n(\{0\})=0$ and $x\mapsto x^2\land1$ is integrable. The compensated Poisson random measure is denoted by $\ov{\m} \coloneqq \m - \l_+\otimes\n$. $W^{\s}$ denotes the Gaussian part of $L$, i.e. an $\mF^L$-Brownian motion with $\mE (W^{\s}_t)^2 = \s^2 t$, $\s\geq 0$. $\b\in\mR$ denotes the drift parameter of $L$. Thus, the triple $(\b,\s^2,\n)$ denotes the $\mF$-characteristics of $L.$ For every $t\geq0$, the characteristic function of $L_t$ is given by
\begin{align*}
    \mE \exp(iuL_t) = \exp(t\psi(u)),
\end{align*}
where 
\begin{align*}
    \psi(u) \coloneqq i\b u - \frac{1}{2}u^2\s^2 + \int_{\mR} (e^{iux} - 1 - iux\1_{\{|x|\leq 1\}})\n(dx), \ u\in\mR.
\end{align*}
\par
We introduce more notation, specific to this particular discussion. Denote by $\cG^2_{\loc}(\m)$ the linear space of $\cB(\mR_+)\otimes\cP(\mF)$-measureable mappings $G$ such that the increasing process $\sum_{s\leq\cdot}G^2(s,\w,\D L_s(\w))\1_{\{\D L_s(\w)\neq 0\}} $ is locally integrable. Then, for $G\in\cG^2_{\loc}(\m)$, the stochastic integral $\int_{\mR_+}\int_{\mR}G(s,x)\ov{\m}(ds,dx)$ is defined to be the unique purely discontinuous martingale $Z\in\cH^2_{0,\loc}(\mF)$ such that 
\begin{align*}
    \D Z_t(\w) = G(t,\w,\D L_t(\w))\1_{\{\D L_t(\w)\neq 0\}}, \ t\geq 0.
\end{align*}
\par
Note that for any $f\in L^2(\n), t\geq 0,$ the deterministic function $G_f(s,x)\coloneqq \1_{[0,t]}(s)f(x)$ is in $\cG^2_{\loc}(\m)$ since
\begin{align*}
    \mE \sum_{s\leq t}G_f^2(s,\w,\D L_s(\w))\1_{\{\D L_s(\w)\neq 0\}} = \mE \int_{\mR_+}\int_{\mR}G_f^2(s,x)\m(ds,dx) = t\int_{\mR}f^2(x)\n (dx)<\infty
\end{align*}
\par
Finally, we can define the process $X^f\in\cH^2_{0,\loc}(\mF)$ via 
\begin{align*}
    X_t^f \coloneqq \int_{[0,t]\times\mR}f(x)\ov{\m}(ds,dx), \ t\geq 0.
\end{align*}
\par
We now provide a number of theorems and lemmas necessary to state the main result of the subsection. The proofs can be found in the cited paper.  \par

\begin{thm} Let $(L,\mF)$ be a Levy process with the $\mF$-characteristics $(\b,\s^2,\n)$. For any $f\in L^2$, the following claims hold: \par
(i) $(X^f,\mF)$ is a Levy process and a true martingale; \par
(ii) for every $t\geq 0$, the identity $\mE (X^f_t)^2 = t\int_{\mR}f^2 d\n<\infty$ holds. Hence, for every deterministic $T\in\mR_+$, the stopped process $(X^f_{t\land T})_{t\geq 0}$ belongs to $\cH^2_0(\mF).$ \par
(iii) $\langle X^f, X^g \rangle_t = t\int_{\mR}fgd\n$ for every $f,g\in L^2(\n), t\geq 0.$ \par
(iv) $X^f$ and $X^g$ are orthogonal martingales if and only if $f,g\in L^2(\n)$ are orthogonal functions.
\end{thm}

\par

For $\cJ \subseteq L^2(\n)$, we define 
\begin{align*}
    \mathscr{X}_{\cJ}\coloneqq \{W^{\s}\}\cup\{X^f, f\in\cJ\}.
\end{align*}
\par
We also recall the notation for permissible integrands given by 
\begin{align*}
    L^2(\mF,X) = \left\{K : K \ \text{is} \ \mF\text{-predictable and} \quad \mE [K^2\bullet\langle X,X \rangle_{\infty}]<\infty\right\}.
\end{align*}
\par

\begin{thm} Let $\cJ \coloneqq \{f_n, n\geq 1\}\subseteq L^2(\n)$ be an orthonormal basis. Then the family $\mathscr{X}_{\cJ}$ possesses the PRP with respect to $\mF^L$, i.e. every $X\in\cH^2(\mF^L)$ can be written as 
\begin{align*}
    &X_t = X_0 + Z\cdot W^{\s}_t + \sum_{n=1}^{\infty}V^n\bullet X^{f_n}_t, \\
    Z\in L^2(\mF^L, &W^{\s}), V^n\in L^2(\mF^L, X^{f_n}), \ n\geq 1,  \ t\geq 0.
\end{align*}
\end{thm}

\begin{defn} Let $l^2$ denote the Hilbert space of sequences $v = (v^n)_{n\geq 1}$ for which the norm $\|v\|^2_{l^2}\coloneqq \sum_{n=1}^{\infty}(v^n)^2$ is finite. Denote by $M^2(\mF,l^2)$ the space of $l^2$-valued $\mF$-predictable processes $V$ such that $\int_0^{\infty}\|V_s\|^2_{l^2}ds$ is integrable. 
\end{defn}

\par
Note that $\sum_{n=1}^{\infty}V^n\bullet X^{f_n}\in\cH^2_0(\mF)$ if and only if $V = (V^n)_{n\geq 1}\in M^2(\mF,l^2).$ \par
As usual, we let $\L^{\mG}$ be the $\mG$-predictable compensator of $H = \1_{\lb\t,\infty)}$ given by
\begin{align*}
    \L^{\mG} = \int_0^{t\land\t}\frac{1}{A_{s-}}dH^p_s,
\end{align*}
where $A \coloneqq \prescript{o}{}(1 - H) = (\mP(\t>t|\cF_t), t\geq 0)$ is the Azema supermartingale and, as can be shown, under the assumptions for $\t$, $H^p = H^o = (1 - A)$. Recall also the $\mG$-local martingale $M$ given by 
\begin{align*}
    M_t \coloneqq H_t - \L_t^{\mG}, \quad t\geq 0.
\end{align*}
\par 
Note that this implies that $M\in\cH^2_0(\mG)$ and $\langle M, M \rangle = \L^{\mG}.$ \par

\begin{thm} Assume that the random time $\t$ satisfies the conditions $(\cA)$ and $(\cH)$ above. Let $(L,\mF^L)$ be a Levy process with $\mF^L$-characteristics $(\b,\s^2,\n)$. Let $\cJ\subseteq L^2(\n)$ be an orthonormal basis. Then, \par
(i) $(L,\mG)$ is a Levy process with $\mG$-characteristics $(\b,\s^2,\n)$; \par
(ii) $\mathscr{X}\coloneqq \mathscr{X}_{\cJ}\cup\{M\}$ is a family of $\mG$-martingales such that $\mE X^2_t<\infty, X\in\mathscr{X}, t\geq 0$; \par
(iii)For every $X\in\mathscr{X}_{\cJ},$ the identity $[X,M]=0$ holds (up to an evanescent set). In particular, $\mathscr{X}$ consists of pairwise orthogonal $\mG$-martingales. 
\end{thm}

\par
\begin{lem} Let $(\z^k)_{k\in\mN}$ be a sequence converging in $L^2(\W,\cG_{\infty},\mP)$ to $\z.$ If 
\begin{align*}
    \z^k = \mE \z^k + Z^k\bullet W^{\s}_{\infty} + \sum_{n=1}^{\infty}V^{n,k}\bullet X^{f_n}_{\infty} + U^k\bullet M_{\infty}, \\
    Z^k\in L^2(\mG,W^{\s}), V^k\in M^2(\mG,l^2), U^k\in L^2(\mG,M), k\in \mN,
\end{align*}
then there exist $Z\in L^2(\mG,W^{\s}), V\in M^2(\mG,l^2), U\in L^2(\mG,M)$ such that 
\begin{align*}
    \z = \mE \z + Z\bullet W^{\s}_{\infty} + \sum_{n=1}^{\infty}V^n\bullet X^{f_n}_{\infty} + U\bullet M_{\infty}.
\end{align*}
\end{lem}

\par 
We are now ready to state the main result on the preservation of the PRP property in the progressively enlarged Levy filtration $\mG.$ \par

\begin{thm} Let $(L,\mF^L)$ be a Levy process with $\mF^L$-characteristics $(\b,\s^2,\n)$. Denote by $\cJ = \{f_n, n\geq 1\}$ an orthonormal basis of $L^2(\n)$ and by $\mathscr{X}_{\cJ}\subseteq\cH^2_{0,\loc}(\mF^L)$ the family of martingales associated with $\cJ$. Assume that the random time $\t$ satisfies the assumptions $(\cA)$ and $(\cH)$ and let $\mG$ be the progressive enlargement of $\mF^L$ by $\t$. Then, the orthogonal family of $\mG$-martingales $\mathscr{X} = \mathscr{X}_{\cJ}\cup\{M\}\subseteq\cH^2_{0,\loc}(\mG)$ possesses the PRP with respect to $\mG$, i.e. every $X\in\cH^2(\mG)$ can be represented as 
\begin{align*}
    X_t = X_0 + Z\bullet W^{\s}_t + \sum_{n=1}^{\infty}V^n\bullet X_t^{f_n} + U\bullet M_t, \quad t\geq 0,
\end{align*}
where $Z\in L^2(\mG,W^{\s}), V\in M^2(\mG,l^2), U\in L^2(\mG,M)$. Moreover, this representation is unique, i.e. if there exists another triplet $Z'\in L^2(\mG,W^{\s}), V'\in M^2(\mG,l^2), U'\in L^2(\mG,M)$ such that the above representation holds for $(Z',V',U')$ instead of $(Z,V,U)$, then we have 
\begin{align*}
    \|Z-Z'\|_{L^2(\mG,W^{\s})}=0, \ \|V-V'\|_{M^2(\mG,l^2)}=0, \ \|U-U'\|_{L^2(\mG,M)}=0.
\end{align*}
\end{thm}
\par

In fact, the proofs for the above theorem work for any filtration $\mF$ satisfying the usual conditions and any countable family $\mathscr{X} = \{X^n, n\geq 1\}\subseteq\cH^2(\mF)$ of mutually orthogonal $\mF$-martingales possessing the PRP in $\mF$. We state it as the most general result in this context. \par

\begin{thm} Let $\mF$ be a filtration satisfying the usual conditions. Suppose $\mathscr{X} = \{X^n, n\geq 1\}\subseteq\cH^2(\mF)$ is an arbitrary family of mutually orthogonal $\mF$-martingales possessing the PRP with respect to $\mF$. Assume that the random time $\t$ satisfies the assumptions $(\cA)$ and $(\cH)$ and let $\mG$ be the progressive enlargement of $\mF$ by $\t$. Then, \par
(i) $\mathscr{X}$ is an orthogonal family of square-integrable $\mG$-martingales; \par
(ii) for every $n\geq 1$, the identity $[X^n,M] = 0 $ holds. Hence $X^n, M \in \cH^2(\mF)$ are orthogonal; \par
(iii) every $X\in\cH^2(\mG)$ can be represented as 
\begin{align*}
    X_t = X_0 + \sum_{n=1}^{\infty}V^n\bullet X_t^n + U\bullet M_t, \quad t\geq 0,
\end{align*}
where $V^n\in L^2(\mG,X^n), n\geq 1, U\in L^2(\mG,M)$. Furthermore, $V^n$ and $U$ are $\langle X^n,X^n \rangle\otimes\mP$-a.e. and $\langle M,M \rangle\otimes\mP$-a.e. unique on $\mR_+\times\W$ respectively.
\end{thm}

\bigskip

We finish this subsection with a discussion of a very interesting result on the enlargement of a filtration generated by a point process obtained in (Di Tella, Jeanblanc, 2020). We do not reproduce the notation that was introduced above in various settings and is obvious from the context. \par
Let $(\W,\cF, \mF, \mP)$ be a complete stochastic basis where $\mF$ satisfies the usual conditions. $\cO(\mF)$ and $\cP(\mF)$ represent the $\s$-algebra of $\mF$-optional and $\mF$-predictable sets of $\W\times\mR_+$ respectively. $\sA^+ = \sA^+(\mF)$ is the space of $\mF$-adapted integrable increasing processes, i.e. the space of increasing processes $X$ such that $\mE X_{\infty}<\infty.$ As usual, $\sA^+_{\loc} = \sA^+_{\loc}(\mF)$ is the localized version of $\sA^+.$ Recall that if $X\in\sA^+_{\loc}$, then there exists a unique $\mF$-predictable process $X^{p,\mF}\in\sA^+_{\loc}$ such that $X - X^{p,\mF}\in\cH^1_{\loc}(\mF)$. $X^{p,\mF}$ is called the $\mF$-compensator or the $\mF$-dual predictable projection of $X$, and we define $\ov{X}^{\mF}\coloneqq X - X^{p,\mF}.$ A point process $X$ with respect to $\mF$ is an $\mN$-valued and $\mF$-adapted increasing process such that the associated jump process satisfies $\D X \in \{0,1\}.$ Since a point process is locally bounded, it follows that $X\in\sA^+_{\loc}.$ $\m$ will denote a nonnegative random measure on $\mR_+\times E$, where $E$ is either $\mR^d$ or some Borel subset of it. We assume $\m(\w,\{0\}\times E) \equiv 0.$ Define $\wt{\W}\coloneqq \W\times\mR_+\times E, \wt{\cO}(\mF)\coloneqq \cO(\mF)\otimes\cB(E)$ and $\wt{\cP}(\mF)\coloneqq \cP(\mF)\otimes \cB(E).$ An $\wt{\cO}(\mF)$-measureable (resp. $\wt{\cP}(\mF)$-measureable) mapping $W$ from $\wt{\W}$ to $\mR$ will be called n $\mF$-optional (resp $\mF$-predictable) function. For an $\mF$-optional function $W$ we define
\begin{align*}
    W*\m(\w)_t\coloneqq \begin{cases}
    \int_{[0,t]\times E}W(\w,t,x)\m(\w,dt,dx), \ \text{if} \ \int_{[0,t]\times E}|W(\w,t,x)|\m(\w,dt,dx)<\infty \\
    +\infty, \quad \text{else}.
    \end{cases}
\end{align*}
\par
$\m$ is called an $\mF$-optional (resp. $\mF$-predictable) random measure if $W*\m$ is $\mF$-optional (resp. $\mF$-predictable), for every $\mF$-optional (resp. $\mF$-predictable) function $W.$ For an $\mR^d$-valued $\mF$
-semimartingale $X$, $\m^X$ will denote its jump measure given by
\begin{align*}
    \m^X(\w,dt,dx) = \sum_{s>0}\1_{\{\D X_s(\w)\neq 0\}}\d_{(s,\D X_s(\w))}(dt,dx),
\end{align*}
where $\d_a$ denotes the Dirac point mass at $a$. We recall that $\m^X$ is an integer-valued random measure with respect to $\mF$ and, in particular, an $\mF$-optional random measure. Moreover, if $\m^X(\w;[0,t]\times \mR^d)<\infty$, for every $\w\in\W$ and $t\in\mR_+$, then $\m^X$ is an $\mR^d$-valued \textit{marked point process} with respect to $\mF$. $(B^X, C^X,\n^X)$ will denote the $\mF$-predictable characteristics of $X$. Note that $\n^X$ is a predictable random measure that is characterized by two properties: for any $\mF$-predictable mapping $W$ such that $|W|*\m^X\in\sA^+_{\loc}$, it holds that $|W|*\n^X\in\sA^+_{\loc}$, and $W*\m^X - W*\n^X\in\cH^1_{\loc}(\mF).$  The proofs and discussions of these notions (and much more) can be found in (Jacod, Shiryaev, 2002, Ch.2). \par
Let $X$ be a point process and $\mX = (\sX_t)_{t\geq 0 }$ the filtration generated by it. We denote by $\sR$ a $\s$-field and call it the \textit{initial $\s$-field}. The \textit{initially enlarged} filtration $\mF$ is defined via $\cF_t\coloneqq \sR \lor \sX_t$. Note that $X$ remains a point process with respect to $\mF$, however, in general, the predictable compensator $X^{p,\mX}$ does not coincide with $X^{p,\mF}$. The following lemma is well-known ((i) Jacod, 1979; (ii) Jacod, Shiryaev, 2002). \par

\bigskip

\begin{lem} In the above setting it holds that \par
(i) $\mF$ is right-continuous; \par
(ii) for every $Y\in\cH^1_{\loc}(\mF)$ there exists $K\in L^1_{\loc}(\ov{X}^{\mF})$ such that $Y$ can be represented as
\begin{align*}
    Y = Y_0 + K\bullet\ov{X}^{\mF},
\end{align*}
i.e. the $\mF$-local martingale $\ov{X}^{\mF}\coloneqq X - X^{p,\mF}$ possesses the PRP with respect to $\mF.$
\end{lem}
\par

\bigskip

Now consider another point process $H$ and the filtration $\mH$ generated by it. We denote by $\mG \coloneqq \mF \lor \mH$ the \textit{progressive enlargement} of $\mF$ by $\mH$.  \par
Note that the default process $H \coloneqq \1_{\lb\t,\infty)}$ we have encountered on numerous occasions previously is a particular example of the point process $H$ considered here. We know that $\ov{X}^{\mF}$ has the PRP in $\mF$ and $\ov{H}^{\mH}$ has the PRP in $\mH$. The general result obtained in (Di Tella, Jeanblanc, 2020) concerns the preservation of the PRP in $\mG$. We now provide a number of theorems and lemmas that allow to state (and prove) their main theorem and the most general result of this subsection on point processes. \par

\begin{thm} Let $X$ and $H$ be two point processes with respect to $\mF$ and $\mH$, respectively. Then, the processes $X - [X,H]$, $H - [X,H]$ and $[X,H]$ are point processes with respect to $\mG$. Furthermore, they have pairwise no common jumps.
\end{thm}

\par

\begin{thm} Consider the $\mR^2$-valued $\mG$-semimartingale $\wt{X} = (X,H)^T.$ Then, the following statements hold \par
(i) the jump measure $\m^{\wt{X}}$ of $\wt{X}$ on $\mR_+\times E$, where $E\coloneqq \{(1,0);(0,1);(1,1)\}$, is given by 
\begin{align*}
    \m^{\wt{X}}(dt,dx_1,dx_2) = &d(X_t-[X,H]_t)\d_{(1,0)}(dx_1,dx_2) + \\ &d(H_t-[X,H]_t)\d_{(0,1)}(dx_1,dx_2) + d[X,H]_t\d_{(1,1)}(dx_1,dx_2).
\end{align*}
Furthermore, $\m^{\wt{X}}$ is an $\mR^2$-valued marked point process with respect to $\mG$. \par
(ii) the $\mG$-compensator $\n^{\wt{X}}$ of $\m^{\wt{X}}$ is given by
\begin{align*}
    \n^{\wt{X}}(dt,dx_1,dx_2) = &d(X_t-[X,H]_t)^{p,\mG}\d_{(1,0)}(dx_1,dx_2) + \\ &d(H_t-[X,H]_t)^{p,\mG}\d_{(0,1)}(dx_1,dx_2) + d[X,H]^{p,\mG}_t\d_{(1,1)}(dx_1,dx_2).
\end{align*}
\end{thm}

\par

\bigskip

Denote by $\wt{\mG}$ the smallest right-continuous filtration such that $\m^{\wt{X}}$ is optional and by $\mG^*$ its initial enlargement given by $\cG^*_t\coloneqq \sR \lor \wt{\cG}_t$. \par

\begin{lem} 
$\mG^*$ coincides with $\mG$. 
\end{lem}

\par

\bigskip

We differentiate two kinds of representations of martingales: one is with respect to a compensated random measure, and the other with respect to the following three locally bounded $\mG$-local martingales 
\begin{align*}
    Z^1\coloneqq \ov{X - [X,H]}^{\mG}, \ Z^2\coloneqq \ov{H - [X,H]}^{\mG}, \ Z^3\coloneqq \ov{[X,H]}^{\mG},
\end{align*}
where $\ov{Z}^{\mG}\coloneqq Z - Z^{p,\mG}.$ \par
Finally, we are in a position to state the main result of this subsetion. For a proof and detailed discussion we refer to the original paper.\par

\begin{thm} Let $Y\in\cH^1_{\loc}(\mG).$ Then, the following statements hold \par
(i) there exists a $\mG$-predictable function $W$ such that $|W|*\m^{\wt{X}}\in\sA^+_{\loc}(\mG)$ and $Y$ can be represented as
\begin{align*}
    Y = Y_0 + W*\m^{\wt{X}} - W*\n^{\wt{X}}.
\end{align*}
\par
(ii) $Y$ also has the following representation
\begin{align*}
    Y = Y_0 + K^1\bullet Z^1 + K^2\bullet Z^2 + K^3\bullet Z^3,
\end{align*}
where $K^i\in L^1_{\loc}(Z^i,\mG), i=1,2,3.$ 
\end{thm}
\par

\bigskip

As the authors point out, this result generalizes that of (Calzolari, Torti, 2016) that we have covered before since in the latter they make the very strong assumption that $\ov{X}^{\mF}$ and $\ov{H}^{\mH}$ remain $\mG$-local martingales and be orthogonal in $\mG$. Moreover, $\cF_0$ is not assumed to be trivial for the \textcolor{black}{extension} of the PRP to $\mG$ to hold. The theorem above also generalizes a number of other related results in the literature that we have not covered. The interested reader can conduct these comparisons on their own by relating the current results to those in (Di Tella, 2020), (Aksamit, Jeanblanc, Rutkowski, 2019) and (Xue, 1993). \par
We finish the discussion by stating a corollary to the main theorem related to the \textit{stable subspaces} generated by the collection of local martingales $\sZ\coloneqq\{Z^1, Z^2, Z^3\}$ in $\cH_0^p(\mG)$, for $p\geq 1$, and denoted by $\sL^p_{\mG}(\sZ)$.\par

\begin{cor} Let $p\geq 1$ and $\sZ\coloneqq\{Z^1, Z^2, Z^3\}$. Then \par
(i) the identity $\sL^p_{\mG}(\sZ)=\cH_0^p(\mG)$ holds; \par
(ii) if furthermore $p=2$ and the local martingales $Z^1,Z^2,Z^3\in\cH^2_{0,\loc}(\mG)$ are pairwise orthogonal, then every $Y\in\cH^2(\mG)$ can be represented as
\begin{align*}
      Y = Y_0 + K^1\bullet Z^1 + K^2\bullet Z^2 + K^3\bullet Z^3,
\end{align*}
where $K^i\in L^2(Z^i,\mG), i=1,2,3,$ which is an orthogonal decomposition of $Y$ in $(\cH^2(\mG), \|\cdot\|_2).$
\end{cor}

\par

\section{Conclusion}

We have provided a detailed account of the current state of the literature on filtration enlargements in the context of its applications in mathematical finance. As is evident from the decomposition formulas, enlargement of filtration sits at the intersection of the fundamental results of stochastic analysis such as Ito's theorem, the Doob-Meyer decomposition, change of measure techniques and (Markov) bridges. Our exposition of the connections of the theory of enlargement to the fundamental theorems of mathematical finance shows under what conditions enlargement results in the creation of arbitrage opportunities and when the no-arbitrage property is preserved. Likewise, we have studied the conditions that ensure the preservation of the martingale representation property in the market. Notably, our survey has revealed that the focus in the literature over the last decades has been on enlargements and \textit{filtration reduction} is a much less researched phenomenon, particularly in the context of finance. We plan to bridge that gap in a sequence of follow-up papers.

\end{document}